\documentclass[aihp,12pt]{imsart}

\RequirePackage[square,numbers,sort&compress]{natbib}

\RequirePackage[colorlinks,citecolor=blue,urlcolor=blue]{hyperref}
\RequirePackage{graphicx}

\usepackage[usenames,dvipsnames]{xcolor}                                   

\usepackage{cancel}
\usepackage{amsfonts,amsmath,amssymb}
\usepackage[english]{babel}
\usepackage[normalem]{ulem}
\usepackage{hyperref}

\usepackage{amsthm}
\usepackage{graphicx}
\usepackage{xfrac}

\usepackage[linesnumbered,noend]{algorithm2e}

\usepackage{subcaption,caption}
\usepackage{epsfig}
\usepackage{ulem}
\usepackage{tikz}
\usepackage{mathrsfs}
\usepackage{bbold}             

\usepackage{soul}

\usepackage{pgfplots}

\pgfplotsset{compat = newest}

\usepackage{bm}
\usetikzlibrary{shapes,arrows}
\usetikzlibrary{matrix}
\usetikzlibrary{calc,patterns,angles,quotes,arrows,automata}
\graphicspath{ {./imgs/}}

\newcommand{\Ac}{\mathcal{A}}
\newcommand{\Bc}{\mathcal{B}}

\newcommand{\Dc}{\mathcal{D}}
\newcommand{\Ec}{\mathcal{E}}
\newcommand{\Fc}{\mathcal{F}}
\newcommand{\Gc}{\mathcal{G}}

\newcommand{\Hc}{\mathcal{H}}
\newcommand{\Jc}{\mathcal{J}}
\newcommand{\Rc}{\mathcal{R}}

\newcommand{\Sc}{\mathcal{S}}

\newcommand{\Qc}{\mathcal{Q}}
\newcommand{\Ic}{\mathcal{I}}
\newcommand{\Kc}{\mathcal{K}}
\newcommand{\Lc}{\mathcal{L}}

\newcommand{\Xc}{\mathcal{X}}

\newcommand{\As}{\mathscr{A}}

\newcommand{\Ls}{\mathscr{L}}

\newcommand{\Ns}{\mathscr{N}}
\newcommand{\Vs}{\mathscr{V}}
\newcommand{\Ws}{\mathscr{W}}

\newcommand{\Rs}{\mathscr{R}}

\newcommand{\Ts}{\mathscr{T}}

\newcommand{\Fs}{\mathscr{F}}

\newcommand{\Rb}{\mathbb{R}}
\newcommand{\Cb}{\mathbb{C}}
\newcommand{\Pb}{\mathbb{P}}
\newcommand{\Eb}{\mathbb{E}}

\newcommand{\Nb}{\mathbb{N}}

\newcommand{\Bf}{\mathfrak{B}}

\newcommand{\Hf}{\mathfrak{H}}
\newcommand{\Df}{\mathfrak{D}}

\newcommand{\CE}{\mathbb{E}|}

\newcommand{\rs}{\mathcal{R}}
\newcommand{\es}{\mathcal{J}}

\newcommand{\um}{\frac{1}{2}}
\newcommand{\one}{\mathbb{1}}

\newcommand{\tr}{{\rm tr}}

\newcommand{\Span}{{\rm span}}

\newcommand{\alg}{{\rm alg}}

\newcommand{\ket}[1]{\left| #1 \right>}
\newcommand{\bra}[1]{\left< #1 \right|}

\newcommand{\inner}[2]{\left< #1,#2 \right>}

\newcommand{\expect}[1]{\left< #1 \right>}

\newcommand{\ketbra}[2]{\left\vert #1 \middle>\middle< #2 \right\vert}

\definecolor{darkviolet}{rgb}{0.58, 0.0, 0.83}
\definecolor{lavender}{rgb}{0.45, 0.31, 0.59}

\theoremstyle{theorem}
\newtheorem{lemma}{Lemma}
\newtheorem{proposition}{Proposition}

\newtheorem{theorem}{Theorem}
\newtheorem{corollary}{Corollary}

\theoremstyle{definition}
\newtheorem{definition}{Definition}
\newtheorem*{assumption*}{Assumptions}

\theoremstyle{remark}
\newtheorem{remark}{Remark}
\newtheorem{example}{Example}

\sethlcolor{yellow!50}

\begin{document}
\begin{frontmatter}
\title{Quantum model reduction for continuous-time quantum filters}

\begin{aug}

\author[A]{\inits{T.}\fnms{Tommaso }~\snm{Grigoletto*}\ead[label=e1]{tommaso.grigoletto@unipd.it}},
\author[B]{\inits{C.}\fnms{Clément}~\snm{Pellegrini}\ead[label=e2]{clement.pellegrini@math.univ-toulouse.fr}}
\and
\author[A,C]{\inits{F.}\fnms{Francesco}~\snm{Ticozzi}\ead[label=e3]{francesco.ticozzi@unipd.it}}

\address[A]{Department of Information Engineering, University of Padova, Italy\printead[presep={,\ }]{e1,e3}}

\address[B]{Institut de Mathématiques de Toulouse, UMR5219, Université de Toulouse, CNRS, UPS IMT,
F-31062 Toulouse Cedex 9, France\printead[presep={,\ }]{e2}}

\address[C]{Department of Physics and Astronomy, Dartmouth College, Hanover, New Hampshire 03755, USA}
\end{aug}

\begin{abstract}
The use of quantum stochastic models is widespread in dynamical reduction, simulation of open systems, feedback control and adaptive estimation. In many applications only part of the information contained in the filter's state is actually needed to reconstruct the target observable quantities; thus, filters of smaller dimensions could be in principle implemented to perform the same task. In this work, we propose a systematic method to find, when possible, reduced-order quantum filters that are capable of exactly reproducing the evolution of expectation values of interest. In contrast with existing reduction techniques, the reduced model we obtain is exact and in the form of a Belavkin filtering equation, ensuring physical interpretability. This is attained by leveraging tools from the theory of both minimal realization and non-commutative conditional expectations. The proposed procedure is tested on prototypical examples, laying the groundwork for applications in quantum trajectory simulation and quantum feedback control.
\end{abstract}
\end{frontmatter}

\section{Introduction}

Despite the comforting unitarity of quantum dynamics as prescribed by Schr\"odinger's equation \cite{von2018mathematical}, its stochastic extensions have emerged as natural candidates to model quantum measurement processes as dynamical systems \cite{percival1998quantum,adler2001martingale}, and to introduce spontaneous localization mechanisms in quantum theory \cite{ghirardi1985model,ghirardi1990markov}.  Probabilistic behavior can be introduced in Schr\"odinger's equation as a stochastic fluctuation of the Hamiltonian operator, forcing one to add some correction terms in order to maintain its state as a valid state vector \cite{barchielli2009quantum}. 

Independently, quantum stochastic evolutions of the same form have been derived in the 1970's as the result of the dynamical interaction of quantum systems with an infinite-dimensional environment, modeled as a quantum field in the framework of quantum probability \cite{parthasarathy2012introduction}. In the pioneering work of Belavkin, stochastic models of this type emerge as the quantum equivalent of a Kushner-Stratonovich equation \cite{Kushner,stratonovich1965conditional, Bucy}, i.e. the dynamical model for a quantum system undergoing indirect continuous observation \cite{BELAVKIN1992171,bouten_quantum_2006,bouten2007introduction,vanhandel2008stabilityquantummarkovfilters}. Later, similar models emerged in quantum optics, and have been used to model different types of measurements and their fluctuations \cite{carmichael1993open,wiseman1994quantum, wiseman2009quantum}. For a review of quantum optical models from a mathematical perspective and derivations of the models that avoid the need for noncommutative operator-valued processes, see \cite{barchielli2009quantum}. 

The potential of stochastic quantum evolutions as quantum filtering equations, which provide state estimation based on measurements, to support control algorithms have been already proposed in Belavkin's work \cite{belavkin2004towards}, and have been developed into a subfield of quantum control \cite{altafini2012modeling}. State-based feedback control based on stochastic models has been experimentally implemented on different platforms \cite{sayrin2011real, bonato2016optimized}. In any application of these models that requires numerical integration of the resulting SDEs, increasing the size of the system introduces a major hurdle. When integration has to be performed in real time, as in feedback protocols, this issue practically limits applicability to extremely small systems. 

\vspace{3mm} In this work we aim to construct smaller models that are able to {\em exactly} reproduce the output of interest for the model, which shall be assumed to be some linear functional of the state of the system. This is done by projecting the dynamics onto subspaces or algebras that contain the full trajectories of the observables of interest in Heisenberg picture. We show that the latter are contained in a Krylov-type subspace \cite{kry31} which is defined following a direct analogy with the observability analysis in linear system theory \cite{kalman1969topics}. The idea has been introduced for deterministic dynamics in \cite{tit2023,prxq2024} and for discrete time processes in \cite{letter2024}. The works \cite{legoll2010effective,hartman2020} are similar in spirit but are limited to the classical case. The quantum continuous-time scenario, as we shall see, presents peculiar challenges.

The main other approach that has been proposed to limit the size of the model to be integrated is the so-called {\em quantum projection filter}, by introducing parametrization of the state and constructing reductions to the corresponding manifolds approximations \cite{ramadan_exact_2023, nurdin_structures_2014, gao_design_2020, gao_improved_2020, handel_quantum_2005, gao_exponential_2018}. With respect to our case, however, the resulting models are approximate and do not guarantee to preserve the form of the dynamics, limiting their physical interpretability.

\vspace{3mm}{\em System and measurements:} We consider a finite-dimensional quantum system, \(\Hc\simeq\Cb^n\), whose state is described, for $t\geq0$, by the density operator $\rho_t\in\Df(\Hc)$ and subject to $p$ continuous measurement of homodyne type \cite{barchielli2009quantum} and $q$ measurements of counting type \cite{barchielli1991measurements11}.
For each measurement of homodyne type, labeled with $j=1,\dots,p$, the output signal is a scalar stochastic process $(Y_t^j)_{t\geq0}$, whose dynamics obeys the stochastic differential equation 
\begin{equation}\label{eq:output}
    dY_t^j = \tr[(D_j+D_j^*)\rho_t]dt + dW_{t}^j,
\end{equation}
where $\{(W_t^j)_{t\geq0}\}_{j=1}^p$ are independent Wiener processes and $D_j\in\Bf(\Hc)$ are operators that describe the effects of the measurement on the system. 
For each measurement of counting type, labeled with $j=1,\dots,q$, the output is a scalar counting process $(N_t^j)_{t\geq0}$ of stochastic intensity \(\int_{0}^t \tr[C_j\rho_{s} C_j^*]ds\), where $C_j\in\Bf(\Hc)$ are operators that describe the effect of measurement on the system. 
\bigskip

\noindent{\em Stochastic master equation.}
The evolution of the state \((\rho_t)_{t\geq0}\) under these assumptions is known as a {\em quantum trajectory} and is modeled by a jump-diffusion stochastic differential equation (SDE), known in the literature as stochastic master equation (SME) \cite{barchielli2009quantum} or quantum filtering equation \cite{BELAVKIN1992171,bouten2007introduction}. Given an Hermitian operator $H$, a set of arbitrary operators $\{L_j\}_{j=1}^m$, and the sets of operators that describe the relation between the state $\rho_t$ and the measurement outcomes $\{D_j\}_{j=1}^p,\{C_j\}_{j=1}^q$ the state $\rho_t$ satisfies the following stochastic differential equation:
\begin{equation}
\begin{split}
    d\rho_t = &\, \Lc(\rho_{t-}) dt\\& + \sum_{j=1}^{p} \left[ D_j\rho_{t-}+\rho_{t-} D_j^* -\tr[D_j\rho_{t-}+\rho_{t-} D_j^*]\rho_{t-} \right] \left(dY_t^j-\tr[D_j\rho_{t-}+\rho_{t-} D_j^*] dt\right)\\ & +  \sum_{j=1}^{q}\left[\frac{C_j\rho_{t-} C_j^*}{\tr[C_j\rho_{t-} C_j^*]}-\rho_{t-}\right]\left(dN_t^j-\tr[C_j\rho_{t-} C_j^*]dt\right),
\end{split}
    \label{eq:SME}
\end{equation}
where the operator $\mathcal L$ is the so called \emph{Lindblad} (or GKLS) generator and is defined as
\begin{equation}
    \begin{split}
        \Lc(\rho) \equiv&  -i[H,\rho] + \sum_{j=1}^m L_j\rho L_j^*-\frac{1}{2}\{L_j^* L_j,\rho\} + \sum_{j=1}^p D_j\rho D_j^*-\frac{1}{2}\{D_j^* D_j,\rho\} + \sum_{j=1}^q C_j\rho C_j^*-\frac{1}{2}\{C_j^* C_j,\rho\}.
    \end{split}
\end{equation}
If $\rho_0\in \Df(\Hc)$, then the solution $(\rho_t)_{t\geq 0}$ of Eq \eqref{eq:SME} is valued in $\Df(\Hc)$ \cite{pellegrini2010markov}. 

Note that, in Equation \eqref{eq:SME} one can include parameters that characterize the measurements efficiencies \cite{1amini2014stability}. We here decided to not include them in order to lighten the notations. Nonetheless the introduction of efficiency parameters does not change the results we derive next as those are simply scalar coefficients that are not affected by the involved reductions. 
\bigskip

\noindent{\em Ouput functionals}. 
In many cases of practical interest, we are not actually interested in all the information contained in $(\rho_t)_{t\geq0}$. In particular, quantum filtering equations are often used to estimate the state $\rho_t$ which is then used to to compute estimates of linear functionals of the state. Relevant cases include:
\begin{itemize}
    \item In simulation of {\em quantum trajectories} aimed to study the evolution of observables of interest, e.g. $\tr[O \rho_t]$, subject to continuous measurement \cite{tirrito_full_2023}. In the case of {\em non-demolition measurements} in continuous time (\cite{Benoist_2014,bauer2013repeated,bauer2011convergence,cardona2020exponential,benoist2024exponentially}, see Section \ref{sec:qnd} for details and extensions) the Hamiltonian and noise operators need to be diagonal, and the observation of interests correspond to the probability of finding the state in one of their  common eigen-subspaces: i.e. $p_j = \tr[\Pi_j \rho_t]$ with orthogonal projectors $\Pi_j$ such that $\sum_j \Pi_j = \one$.
    
    \item In {\em state reduction models} that mimic the asymptotic behavior induced by quantum measurements \cite{adler2001martingale, Adler_2008}, one obtains models that are equivalent to the non-demolition models, in which case the observables of interest are the  set of spectral projections. 
    \item  In many {\em feedback control scenarios}, see e.g \cite{amini2012stabilization,grigoletto2021stabilization,liang2022switching, benoist2017exponential}, the state $\rho_t$ is used to compute Lyapunov functions of the type $V(\rho_t) = \tr[K\rho_t]$ for some $K\in\Bf(\Hc)$ which is then used to perform feedback control on the system.
   
    \item In Montecarlo-type {\em simulations} of open quantum systems that employ quantum stochastic models to explore the evolution of expectation value of observables under Lindblad dynamics, in which case one could be interested in estimating $\Eb[\tr[O\rho_t]]$, \cite{PhysRevB.105.064305}.
    
\end{itemize}
In general we can assume to be interested in reproducing only the stochastic processes $\{(\Theta_t^j)_{t\geq0}\}_{j=1}^r$, that are defined as 
\begin{equation}
    \Theta_t^j \equiv \tr[O_j\rho_t],
    \label{eq:output_equation}
\end{equation} for a finite set of operators $\{O_j\}_{j=1}^r\subset\Bf(\Hc)$ and for all $t\geq0$. 

From a system-theoretic viewpoint \cite{kalman}, the paired equations \eqref{eq:SME}-\eqref{eq:output_equation} represent a stochastic filter, that takes the signals $(Y_t^j)_{t\geq0}$  and $(N_t^j)_{t\geq0}$ as inputs, and returns as an output the signals $(\Theta_t^j)_{t\geq0}$.  In the following, we denote by $\Sigma$ the filter described by equations \eqref{eq:SME} and \eqref{eq:output_equation}. Similarly, $\Sigma_L$ and $\Sigma_Q$ represent the two other filters, which we shall introduce next, where \eqref{eq:SME} is substituted with a minimal linear SDE and with a reduced order SME, respectively. We informally say that a filter $\Sigma$ is a \textit{quantum filter} if the stochastic differential equation that governs its dynamics is a stochastic master equation of the form \eqref{eq:SME}.
\bigskip

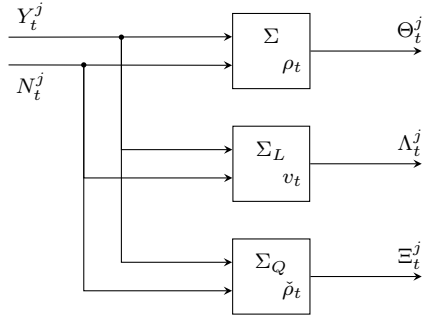
\begin{figure}[ht]
    \centering
    \begin{tikzpicture}[>=stealth]
        \coordinate (orig) at (0,0);
        \coordinate (sigma) at (4,0);
        \coordinate (sigma_c) at (4,-1.5);
        \coordinate (sigma_cc) at (4,-3);
        \coordinate (r) at (6,0);
        \coordinate (r_c) at (6,-1.5);
        \coordinate (r_cc) at (6,-3);

        \node[minimum width = 1cm, minimum height = 1cm,] (Orig) at (orig) {};
        \node[draw, minimum width = 1cm, minimum height = 1cm, align=center] (Sigma) at (sigma) {$\Sigma$\\$\qquad\rho_t$};
        \node[draw, minimum width = 1cm, minimum height = 1cm, align=center] (Sigma_c) at (sigma_c) {${\Sigma}_L$\\$\qquad v_t$};
        \node[draw, minimum width = 1cm, minimum height = 1cm, align=center] (Sigma_cc) at (sigma_cc) {${\Sigma}_Q$\\$\qquad\check{\rho}_t$};

        \path[draw,->] (Sigma) -- node[above, pos=0.9]{$\Theta_t^j$} (r);
        \path[draw,->] (Sigma_c) -- node[above, pos=0.9]{$\Lambda_t^j$} (r_c);
        \path[draw,->] (Sigma_cc) -- node[above, pos=0.9]{$\Xi_t^j$} (r_cc);
        \path[draw,->] (Orig.20) -- node[above, pos=0.1]{$Y^j_t$} (Sigma.160);
        \path[draw,->] (Orig.340) -- node[below, pos=0.1]{$N^j_t$} (Sigma.200);
        \path[draw,->] ($(Orig.20)+(1.5,0)$) |-  (Sigma_c.160);
        \path[draw,->] ($(Orig.340)+(1.0,0)$) |-  (Sigma_c.200);
        \path[draw,->] ($(Orig.20)+(1.5,0)$) |-  (Sigma_cc.160);
        \path[draw,->] ($(Orig.340)+(1.0,0)$) |-  (Sigma_cc.200);
        \path[fill] ($(Orig.20)+(1.5,0)$) circle[radius=1pt];
        \path[fill] ($(Orig.340)+(1.0,0)$) circle[radius=1pt];
    \end{tikzpicture}
    \caption{Schematic of the use of the original filter $\Sigma$ and reduced filters ${\Sigma_L}$ and $\Sigma_Q$.}
    \label{fig:schematic}
\end{figure}

\noindent{\em Reduction}. The objective of this work is to construct, when possible, a more computationally-efficient  {\em quantum filter}, denoted by $\Sigma_Q$, that, using the measurement process $\{(Y_t^j)_{t\geq0}\}_{j=1}^p$ and $\{(N_t^j)_{t\geq0}\}_{j=1}^q$ can {\em exactly} reproduce the output process $\{(\Theta_t^j)_{t\geq0}\}_{j=1}^r$ of the original filter $\Sigma$. 

The main result is provided in Section \ref{sec:reduced_quantum_filters}, where we construct a quantum filter ${\Sigma}_Q$, defined over a *-subalgebra $\check{\As}\subseteq\Bf(\check{\Hc})$ with $\check{\Hc}\subseteq\Hc$, such that: 
\begin{itemize}
    \item to each output process $(Y_t^j)_{t\geq0},$ for $j=1,\dots,p,$  is associated an operator \(\check{D}_j\in\Bf(\check{\Hc})\);
    \item to each output process $(N_t^j)_{t\geq0},$ with  $j=1,\dots,q,$ there is a set of associated operators $\{\check{C}_{j,k}\}_{k=1}^d\subset\Bf(\check{\Hc})$;
\end{itemize}
and whose state $\check{\rho}_t\in\Df(\check{\As})$ evolves according to 
\begin{equation}
    \begin{split}
        d\check{\rho}_t = &\,\check{\Lc}(\check{\rho}_{t-}) dt \\&+ \sum_{j=1}^{p}\left[\check{D}_j\check{\rho}_{t-}+\check\rho_{t-} {\check{D}_j}^* - \tr\left[\check{D}_j\check\rho_{t-}+\check\rho_{t-} {\check{D}_j}^*\right]\check{\rho}_{t-}\right] \left(dY_t^j - \tr[\check{D}_j\rho_{t-}+\check\rho_{t-} {\check{D}_j}^*]dt\right)\\ &+ \sum_{j=1}^{q} \left[  \frac{\sum_{k=1}^d\check{C}_{j,k}\check\rho_{t-}\check{C}_{j,k}^*}{\sum_{k=1}^d\tr[\check{C}_{j,k}\check\rho_{t-}\check{C}_{j,k}^*]} - \check{\rho}_{t-} \right]\left(dN_t^j-\sum_{k=1}^d\tr[\check{C}_{j,k}\check\rho_{t-}\check{C}_{j,k}^*]dt\right)
    \end{split}
    \label{eq:reduced_quantum_filter}
\end{equation}
where $\check{\Lc}$ is a Lindblad generator associated with an Hamiltonian $\check{H}\in\Bf(\check{\Hc})$, and noise operators \(\{\check{L}_{j,k}\}\subset\Bf(\check{\Hc})\), $\{\check{D}_{j}\}$ and $\{\check{C}_{j,k}\}\subset\Bf(\check{\Hc})$. 
Linear functionals of interest can be computed using the reduced filter $\Sigma_Q$ via 
\begin{equation}
     \Xi_t^j \equiv \tr[\check{O}_j\check{\rho}_t]\qquad\forall j=1,\dots,r
\end{equation}
where each $\check{O}_j\in\check{\As}$ is a reduced operator associated to $O_j$. The initial condition \(\rho_0\) is also reduced trough a linear map $\Rc_\As:\Df(\Hc)\to\Df(\check{\As})$, i.e. $\check{\rho}_0=\Rc_\As(\rho_0)$.

Under reasonable assumptions, we prove in Theorem \ref{thm:linear_solves_prob_1} that, initializing the reduced filter with the initial condition $\check{\rho}_0=\Rc_\As(\rho_0)$ almost surely,  we have $\Theta_t^j = \Xi_t^j$, for all initial conditions $\rho_0\in\Df(\Hc),$ for all times $t\geq0,$ and for all $j=1,\dots,r$. 
\bigskip

\noindent{\em Minimal linear filter}.
In the derivation of the quantum filter $\Sigma_Q$, it will prove instrumental to first derive  another filter, denoted by $\Sigma_L,$ as detailed in Section \ref{sec:observable_space}. This filter is {\em linear} (in the sense that its evolution is governed by a linear stochastic master equation) and is minimal in the dimension of the state (which will result to be $\kappa=\dim(\Ns^\perp)$, the dimension of the minimal operator subspace that contains all the observables of interest evolved in Heisenberg picture). What we obtain in this case is not necessarily a quantum filter, in the sense that its state is not necessarily a density operator, but is capable of exactly reproducing the processes of interest $\{(\Theta_t^j)_{t\geq0}\}$ {\em exactly} as the original filter $\Sigma$.
Although it does not provide any physical intuition on the model, it can be relevant for practical applications: for example, it allows one to efficiently implement a filter that estimates the processes $\{(\Theta_t^j)_{t\geq0}\}$ on a classical computer.
\bigskip

A schematic representation of the use of the three filters is shown in Figure \ref{fig:schematic}. All the reduced filters we derive are built to correctly work for all initial conditions $\rho_0\in\Df(\Hc)$: if one were to consider only pure states and a purity-preserving measurement (i.e. with perfect detection efficiency), one could in principle obtain smaller linear filters, but they would only work for this restricted scenario.

\section{Notation and problem setting}
\label{sec:problem_setting}

In this work, we are concerned with finite-dimensional Hilbert spaces $\Hc\simeq\Cb^{n}$. The algebra of bounded operators on $\Hc$ is denoted by $\Bf(\Hc)\simeq\Cb^{n\times n}$, the set of Hermitian operators by $\Hf(\Hc) = \{X\in\Bf(\Hc)|X=X^*\}$, and those that are positive semi-definite by $\Hf_{\geq}(\Hc)=\{X \in\Bf(\Hc)|X=X^* \geq0\}$, while the set of density operators $\Df(\Hc)=\{\rho\in\Bf(\Hc)|\rho=\rho^*\geq0, \tr(\rho)=1\}$. 
With some abuse of notation, given an operator space $\Vs\subseteq\Bf(\Hc)$ we denote by $\Df(\Vs) \equiv \Df(\Hc) \cap \Vs$ the set of density operators contained in $\Vs$.
Throughout the paper, with only few exceptions, we employ the calligraphic notation for superoperators, e.g. $\Lc,\Dc,\Gc,\Kc$ and the script notation for operator spaces, e.g. $\As,\Vs,\Ws$. 

For any density operator $\rho\in\Df(\Hc)$, and operators $C, D\in\Bf(\Hc)$, and $H\in\Hf(\Hc)$ we define the following super-operators: 
\begin{align}
\begin{aligned}
\begin{split}
    [H,\cdot]: \Df(\Hc) &\to \Bf(\Hc)\\ \rho&\mapsto H\rho -\rho H,
\end{split}\qquad
&\begin{split}
    \Kc_C: \Df(\Hc) &\to \Bf(\Hc)\\ \rho&\mapsto C\rho C^*,
\end{split}\\
\begin{split}
    \Gc_{D}: \Df(\Hc) &\to \Bf(\Hc)\\ \rho&\mapsto D\rho + \rho D^*,
\end{split}\qquad
&\begin{split}
    \Dc_C: \Df(\Hc) &\to \Bf(\Hc)\\ \rho&\mapsto C\rho C^* - \frac{1}{2}\{C^* C, \rho\},
\end{split}
\end{aligned}
\label{eq:fundamental_super_operators}
\end{align}
and denote by $\Ic$ the identity superoperator. 

\subsubsection{Stochastic master equation}
In this section,  we rigorously introduce all the necessary concepts and assumptions regarding the dynamics of interest, and outline the derivation of continuous-time quantum filters \eqref{eq:SME} starting from a linear stochastic differential equation for the stochastic evolution operator. This allows us to introduce two intermediate models that are going to be used in the key results of the paper and other derivations later in the paper: the linear SDE \eqref{eq_defS} for the stochastic evolution superoperator and the linear SDE \eqref{eq:Belavkin-Zakai} for the un-normalized state. The latter, in particular, allow us to directly use the tools from control theory to reduce the model. Equation \eqref{eq:SME} is then re-obtained as the SME \eqref{eq:SME_2}. We refer the reader to \cite{barchielliholevo,barchielli2009quantum,Benoist2021} and references therein for more details on the derivation. 

The first step consists in introducing a matrix valued stochastic process called the \emph{stochastic evolution operator}. To this end, let $(\Omega,\mathcal F,(\mathcal F_t)_{t\geq0},\mathbb Q)$ be a filtered probability space with standard Brownian motions $\{(Y_t^j)_{t\geq0}\}_{j=1}^p$ and standard Poisson processes  $\{(N_t^j)_{t\geq0}\}_{j=1}^q$, with intensity $1$, such that the full family $\{(Y_t^j)_{t\geq0}\}_{j=1}^p\cup\{(N_t^j)_{t\geq0}\}_{j=1}^q$ is independent. The filtration $(\mathcal F_t)_{t\geq0}$ is assumed to satisfy the standard conditions, we denote $\mathcal F_{\infty}$ by $\mathcal F$, and the processes $\big(Y_t^j\big)_{t\geq0}$ and $\big(N_t^j-t\big)_{t\geq0}$ are $(\mathcal F_t)_{t\geq0}$-martingales under $\mathbb Q$. 

On $(\Omega,\mathcal F,(\mathcal F_t)_{t\geq0},\mathbb Q)$, for $s\in\mathbb R_+$, let $(\mathcal A_t^s)_{t\in[s,\infty)}$ be the solution to the following SDE:
\begin{equation}\label{eq_defS}
d \mathcal A_{t}^s= \mathcal L\circ\mathcal A_{t-}^{s}\,d t+\sum_{i=1}^p \mathcal G_{D_i}\circ \mathcal A_{t-}^{s}\,d Y^i_t+\sum_{j=1}^q(\mathcal K_{C_j}-\mathcal I)\circ \mathcal A_{t-}^{s}\,d N^j_t,\qquad \mathcal A_s^s=\mathcal I,
\end{equation}
with
\begin{equation}
    \begin{split}
        \Lc(\rho) \equiv&  -i[H,\rho] + \sum_{j=1}^m L_j\rho L_j^*-\frac{1}{2}\{L_j^* L_j,\rho\} + \sum_{j=1}^p D_j\rho D_j^*-\frac{1}{2}\{D_j^* D_j,\rho\} + \sum_{j=1}^q C_j\rho C_j^*-\frac{1}{2}\{C_j^* C_j,\rho\}.
    \end{split}
\end{equation}
where the operators $\{L_j\}_{j=1}^m,\{D_j\}_{j=1}^p,\{C_j\}_{j=1}^q \subset\Bf(\Hc)$ and the Hamiltonian $H$ are the same we mentioned before.

A fundamental relation satisfied by $(\mathcal A_{t}^s)$ is 
$$\mathcal A_t^s\circ\mathcal A_s^r=\mathcal A_t^r,\quad \forall\, 0\leq r\leq s\leq t.$$
We then call the superoperator $\mathcal A_t^s$ the stochastic evolution superoperator, as in \cite{barchielli2009quantum}.  This superoperator is also called the propagator in the physics literature. 

We can now derive the linear stochastic master equation. To this end, let $(\tau_t)_{t\geq0}$ be a family of un-normalized states valued in $\Hf_{\geq0}(\Hc)$ (Hermitian, positive semi-definite) , such that $\tau_0=\rho_0$ and defined by
$$\tau_t=\mathcal A_t^0(\rho_0),\quad \forall t\geq0.$$ Using Ito rules, that is \[dY_t^jdY_t^k=\delta_{jk}dt,\qquad dN_t^jdN_t^k=\delta_{jk}dN_t^k,\qquad dY_t^idN_t^k=0,\qquad dY_t^idt=0\qquad\text{and}\qquad dN_t^kdt=0,\] we can show that $(\tau_t)_{t\geq0}$ satisfies the linear stochastic differential equation (SDE): 
\begin{equation}
\begin{split}
    d\tau_t &= \left[\Lc +\Ic - \sum_{j=1}^q \Kc_{C_j} \right](\tau_{t-}) dt + \sum_{j=1}^{p}\Gc_{D_j}(\tau_{t-}) dY_t^j +\sum_{j=1}^q \left[\Kc_{C_j}-\Ic\right](\tau_{t-}) dN_t^j\\
    &=\Lc (\tau_{t-}) dt + \sum_{j=1}^{p}\Gc_{D_j}(\tau_{t-}) dY_t^j +\sum_{j=1}^q \left[\Kc_{C_j}-\Ic\right](\tau_{t-}) (dN_t^j-dt).
\end{split}
    \label{eq:Belavkin-Zakai}
\end{equation}
In the sequel, we shall denote 
\begin{equation}
    \Qc(\tau) \equiv \left[\Lc +\Ic - \sum_{j=1}^q \Kc_{C_j} \right](\tau),
\end{equation}
for all positive operators $\tau$. Under the probability measure $\mathbb Q$, the processes $\{(Y_t^j)_{t\geq0}\}_{j=1}^p$ and $\{(N_t^j-t)_{t\geq0}\}_{j=1}^q$ are martingales. Using the fact that $\tr[\Lc(\tau_t)]=0$ then the process $(\tr(\tau_t))_{t\geq0}$ is also a martingale satisfying $\mathbb E[\tr(\tau_t)]=1, t\geq0$. This is the key to making the following Girsanov change of measure. Let $T\geq0$, one can consider the change of probability measure
$$\Pb^{\rho_0}_T(F)=\mathbb E_{\mathbb Q}[\mathbf 1_F\tr(\tau_t)], \qquad \forall F\in\mathcal F_T,$$
where $\mathbf 1_F$ is the indicator function for $F\in\mathcal F_T$.
The family of probability measures $(\Pb^{\rho_0}_T)_{T\geq 0}$ constructed by this procedure is consistent: that is, for all $F\in\mathcal F_s$, we have have $\Pb^{\rho_0}_t(F)=\Pb^{\rho_0}_s(F),$  for all $t\geq s$. In this way, one can extend this family of probability and define a unique probability measure $\Pb^{\rho_0}$ such that  $$\Pb^{\rho_0}(F)=\Pb^{\rho_0}_t(F),\qquad\forall F\in\mathcal F_t.$$  Note that intrinsically this probability depends on the initial state $\rho_0$. Nevertheless, all our study is robust with respect to this dependency, and the choice does not influence the reduction procedure.  

Now consider the process $(\rho_t)_{t\geq0}$ defined by 
$$\rho_t=\frac{\tau_t}{\tr(\tau_t)},\qquad t\geq0,$$ we obtain \eqref{eq:SME} satisfied by $(\rho_t)_{t\geq0}$ from \eqref{eq:Belavkin-Zakai}
using the Ito stochastic calculus for jump and diffusion processes.

Again, using Ito stochastic calculus, one can show that $(\rho_t)_{t\geq0}\subset\Df(\Hc)$ satisfies the stochastic master equation (SME) \cite{barchielli2009quantum,bouten2007introduction} or quantum filtering equation \cite{BELAVKIN1992171}:
\begin{equation}
    \begin{split}
        d\rho_t =\,& \Lc(\rho_{t-}) dt\\& + \sum_{j=1}^{p} \left[ \Gc_{D_j}(\rho_{t-}) -\tr[\Gc_{D_j}(\rho_{t-})]\rho_{t-} \right] \left(dY_t^j-\tr[\Gc_{D_j}(\rho_{t-})dt]\right)\\&  + \sum_{j=1}^{q}\left[\frac{\Kc_{C_j}(\rho_{t-})}{\tr[\Kc_{C_j}(\rho_{t-})]}-\rho_{t-}\right]\left(dN_t^j-\tr[\Kc_{C_j}(\rho_{t-})]dt\right).
    \end{split}
    \label{eq:SME_2}
\end{equation}
We thus recover the equation \eqref{eq:SME} presented in Introduction. 

A particular feature of the Girsanov change of measure is that the processes $(W_t^i)_{t\geq0}$, $i=1,\ldots,p$ defined by
$$W_t^i = Y_t^i-\int_0^t\tr[\Gc_{D_j}(\rho_s)]ds,$$
are independent Brownian motions under probability $\mathbb P^{\rho_0}$. The processes $(N_t^j)_{t\geq0}$, $j=1,\ldots,q$ are Poisson processes with intensity
$$t\rightarrow \tr[\Kc_{C_j}(\rho_t)],$$ which in particular implies that the processes defined by
$$N_t^j-\int_0^t\tr[\Kc_{C_j}(\rho_s)]ds$$ are martingales under $\mathbb P^{\rho_0}.$

\subsubsection{Linear functionals}

In many application scenarios for SMEs, one is not interested in the entire information contained in the state $\rho_t$ but only in a limited set of processes that depend linearly on the state $\rho_t$. In this work we focus on \textit{linear functionals} of the state as they cover many cases of interest. Namely, we assume to be interested in reproducing only the stochastic processes $\{\Theta_t^j\}_{j=1}^r$, that are defined as 
\begin{equation}
    \Theta_t^j \equiv \tr[O_j\rho_t],
    \label{eq:output}
\end{equation} for a finite set $\{O_j\}_{j=1}^r\subset\Bf(\Hc)$ and for all $t\geq0$. 

Often times, in practical situations, the set $\{O_j\}_{j=1}^r$ is composed of Hermitian matrices, i.e. {\em observables}. For this reason, in the following we refer to the set $\{O_j\}$ as the set of observables of interest and assume that $\{O_j\}_{j=1}^r\subset\Hf(\Hc)$. Note, however, that this extra assumption is only made for convenience of presentation and can easily be lifted. 
A particularly relevant example is when $\Hc = \Hc_S\otimes\Hc_E,$ and the output of interest is the reduced state on $\Hc_S$. In this case, the output of the linear map $\tr_E$ can be equivalently obtained by choosing $O_j = S_j\otimes \one_E$ where $\{S_j\}$ form an Hermitian basis for $\mathfrak{B}(\Hc_S)$ \cite{prxq2024}. 

We here make two assumptions on this set of operators. 
\begin{assumption*}
    The set $\{O_j\}$ is such that:
    \begin{enumerate}
        \item $\{D_j+D_j^*\}_{j=1}^p\cup\{C_j^* C_j\}_{j=1}^q\subseteq\Span\{O_j\}_{j=1}^r$; \label{ass:observables}
        \item $\one\in\Span\{O_j\}_{j=1}^r$. \label{ass:identity}
    \end{enumerate}
\end{assumption*}
Assumption \ref{ass:observables} will prove to be necessary to ensure that the statistics of the measurement process $\{(Y_t^j)_{t\geq0}\}_{j=1}^p$ and $\{(N_t^j)_{t\geq0}\}_{j=1}^q$ are preserved by the reduced model. In fact, Assumption \ref{ass:observables} is equivalent to requiring that the reduced filter is capable of reproducing just the expectation values of the drift terms in the measurement signals. This fact is showcased in concrete examples in Section \ref{sec:examples}.
On the other hand, Assumption \ref{ass:identity} is technical and derives from the fact that linear functionals can be computed equivalently from states $\rho_t$ or unnormalized states $\tau_t$, since 
\[\tr[O_j\rho_t] = \frac{\tr[O_j\tau_t]}{\tr[\tau_t]}.\]
Assumption \ref{ass:identity} thus allow us to ensure that requiring to preserve \(\tr[O_j\tau_t]\) implies that \(\tr[O_j\rho_t]\) is also preserved. Note that Assumption \ref{ass:identity} is satisfied in most cases of practical interest, as it physically means that we can test the presence of the system undergoing continuous measurements. Note that one can always include more operators into the set of observables of interest $\{O_j\}_{j=1}^r,$ so that both Assumption \ref{ass:observables} and \ref{ass:identity} are satisfied, at the cost of potentially obtaining a larger model. 

\section{Non-observable space and linear reduced filters}
\label{sec:observable_space}

We next introduce the notion of indistinguishable states, which we leverage in the reduction of the models. This concept is well known in the literature on control theory \cite{kalman,kalman1969topics,wonham,marrobasile} and has also been used in the context of quantum filtering \cite{vanhandel2008stabilityquantummarkovfilters}. The definitions that we give here are dual with respect to what is given in \cite{vanhandel2008stabilityquantummarkovfilters} (are given in Schr\"{o}dinger instead of Heisenberg picture).
\begin{definition}[Indistinguishable states and non-observable subpace.]
    We say that two states $\rho_0$ and $\rho_1$ are \textit{indistinguishable from} $\{O_j\}_{j=1}^r$ if we have 
    \begin{equation}
        \tr\left[O_j \mathcal A_t^0(\rho_0)\right] = \tr\left[O_j \mathcal A_t^0(\rho_1) \right]
    \end{equation}
    for all $t\geq0$, and all operators $\{O_j\}_{j=1}^r$.
    
    \noindent The {\em non-observable space} is then defined as the set of operators that are indistinguishable from 0:
    \begin{equation}
        \Ns \equiv\{X\in\Bf(\Hc)|\quad \tr[O_j \mathcal A_t^0 (X)] = 0,\, \forall t\geq0, \, \forall O_j\}.
        \label{eq:non_observable_definition}
    \end{equation}
\end{definition}

The connection between indistinguishable states and the non-observable subspace comes naturally: because of linearity of the map $A_t^0 \cdot {A_t^0}^*$, we have that two states $\rho_0,\,\rho_1$ are indistinguishable if their difference belongs to the non-observable subspace, that is, $\rho_0-\rho_1\in\Ns$. Verifying that $\Ns$ is, in fact, an operator space is also trivial.

Now we shall explore the properties of $\mathcal N$ and we start with a technical lemma that we only prove and express in the case $p=q=1$ since the generalization is straightforward.

\begin{lemma}
\label{lem:trick}
    Assume that $p=q=1$ and let $D$ and $C$ the corresponding measurement operators. Assume that, for some operators $O,X_0\in\Bf(\Hc)$ we have $\tr[OX_t]=0$ (where $X_t\equiv \mathcal A_t^0 (X_0)$) for all $t\geq0$. Then: 
    \[\tr[O\Lc(X_t)]=0, \qquad\text{and}\qquad \tr[O\Gc_D(X_t)]=0 \qquad\text{and}\qquad \tr[O\Kc_C(X_t)]=0\] for all $t\geq0$. 
\end{lemma}
\begin{proof}
    From the assumption that $\tr[O X_t]=0$, we have
    \[\tr[O X_0] + \int_0^t \tr[O \Lc(X_{s-})]ds + \int_0^{t} \tr[O\Gc_D(X_{s-})]dW_s + \int_{0}^t \tr[O(\Kc_{C}-\Ic)(X_{s-})](dN_s-ds) = 0.\]
    Note that, by assumption $\tr[O X_0]=0$. Then, as the above quantities are complex-valued then we should take the real $\Re$ and $\Im$ imaginary part which yields 
    \[\int_0^t \Re(\tr[O \Lc(X_{s-})])ds + \int_0^{t} \Re(\tr[O\Gc_D(X_{s-})])dW_s + \int_{0}^t \Re(\tr[O(\Kc_{C}-\Ic)(X_{s-})])(dN_s-ds) = 0,\]
    and similarly for the imaginary part. 
    Now, computing the conditional quadratic variation and using Ito's rules we obtain
    \[\int_0^t \Re(\tr[O \Gc_D(X_{s-})])^2 ds + \int_{0}^t \Re(\tr[O(\Kc_{C}-\Ic)(X_{s-})])^2 ds = 0.\]
    Hence
    \[\Re(\tr[O\Gc_{D}(X_t)])=0\qquad\text{and}\qquad \Re(\tr[O\Kc_{C}(X_t)])=0.\] for almost all $t\geq0$. More precisely, since the involved processes are c\`adl\`ag the two equality hold for all $t\geq0$. The same reasoning holds for the imaginary part thus \(\tr[O\Gc_{D}(X_t)]=0\) and \(\tr[O\Kc_{C}(X_t)]=0\), for all $t\geq0$. Finally coming back to the first equation we have 
    \[\int_{0}^t \tr[O\Lc(X_{s-})]ds = 0\] which directly implies $\tr[O\Lc(X_t)] = 0$ for all $t\geq0$ (again using the c\`adl\`ag property).
\end{proof}

We next list the main properties of the non-observable subspace. 
 
\begin{proposition}
        \label{prop:non-observable}
    \label{prop:observable_construction}
    Provided the SDE \eqref{eq:Belavkin-Zakai}, the non-observable subspace  $\Ns\subseteq\Bf(\Hc)$ is the largest operator subspace such that the following properties simultaneously hold: 
    \begin{enumerate}
        \item $\tr[O_jX]=0$ for all $X\in\Ns$ and for all $j=1,\dots,r$, i.e. $\Ns\subseteq\cap_{j=1}^r\{X\in\Bf(\Hc)|\, \tr[O_j X]=0\}$;
        \item It is $\Lc$-invariant, i.e. $\Lc(X)\in\Ns,$ for all $ X\in\Ns$;
        \item It is $\Gc_{D_j}$-invariant, i.e. \(\Gc_{D_j}(X)\in\Ns\), for all $X\in\Ns$, for all $j=1,\dots,p$;
        \item It is $\Kc_{C_j}$-invariant, i.e. \(\Kc_{C_j}(X)\in\Ns\), for all $X\in\Ns$, for all $j=1,\dots,q$.
    \end{enumerate}     

    Furthermore, if we denote by $\Ns^\perp$ the orthogonal complement w.r.t. $\inner{\cdot}{\cdot}_{HS}$, i.e. $\Bf(\Hc)=\Ns\oplus\Ns^\perp$ with $\kappa\equiv\dim(\Ns^\perp)$, and denote by $\Ts$ the super-operator algebra $$\Ts \equiv {\rm Alg}\left(\{\Ic, {\Lc}^*\}\cup\{\Gc_{D_j}^*\}_{j=1}^p \cup\{\Kc_{C_j}^*\}_{j=1}^q \right),$$ closed with respect to linear combinations and composition (i.e. for any $\Ac,\Bc\in\Ts$ and $\alpha,\beta\in\Cb$ we have $\alpha\Ac+\beta\Bc\in\Ts$ and $\Ac\Bc\in\Ts$) we have that 
    \begin{equation}
        \Ns^\perp = \Ts \left[\Span\{O_j\}_{j=1}^r\right]\,,
        \label{eq:observable_diffusive}
    \end{equation}
    where the right hand side of the equation is intended as $\Span\left\{\Ec[O], \, \Ec\in\Ts, \, O\in\Span\{O_j\}_{j=1}^r\right\}$.
\end{proposition}
\begin{proof}
 By definition, $\Ns$ is contained in $\cap_{j=1}^r\{X\in\Bf(\Hc)|\, \tr[O_j X]=0\}$ since $A_{0}^0=\one$.

We next want to prove that $\Ns$ is $\Lc$-invariant, $\Gc_{D_j}$-invariant for all $j=1,\dots,p$ and $\Kc_{C_j}$-invariant for all $j=1,\dots,q$, that is, for all $X\in\Ns$ we want to show that $\tr[O_j\mathcal A_t^0(\Lc(X))]=0$, $\tr[O_j\mathcal A_t^0(\Gc_{D_j}(X))]=0$, $j=1,\ldots,p$, $\tr[O_j\mathcal A_t^0(\Kc_{C_j}(X))]=0$, $j=1,\ldots,q$  for all $t\geq0$.

Consider then $X\in\Ns$. By definition, $\forall t,s>0$,  we have 
    \begin{eqnarray*}
    0=\tr[O \mathcal A_{t+s}^0(X)]
    =\tr[O \mathcal A_{t+s}^s\circ{\mathcal A_s^0}(X)]
    =\tr[{\mathcal A_{t+s}^s}^* (O){\mathcal A_s^0}(X)].
    \end{eqnarray*} 
Let us then define the shift operator $\theta^t$ acting on $\Omega$ as
$$\theta^t\omega(s)\equiv\omega(t+s)-\omega(t).$$
Now for all measurable functions $h\in\mathcal F$, multiplying by $h\circ \theta^s$, the previous inequality, we have
   $$\tr[{(\mathcal A_{t+s}^s})^* (O){\mathcal A_s^0}(X)]\cdot(h\circ \theta^s)=0,$$
where $\cdot$ is the scalar multiplication.
Now remarking that $\mathcal A_{t+s}^s=\mathcal A_t^0\circ\theta^s$, we have
$$\tr[(h\circ \theta^s)\cdot({\mathcal A_t^0\circ\theta^s})^* (O)  {\mathcal A_s^0}(X)]=0.$$
By noting that the random variable $(h\circ \theta^s)\cdot({\mathcal A_t^0\circ\theta^s})^* (O)$  is independent of $\mathcal F_s$, we have
\begin{eqnarray*}
0&=\,&\mathbb E_{\mathbb Q}\left[\tr[(h\circ \theta^s)\cdot({\mathcal A_t^0\circ\theta^s})^* (O) {\mathcal A_s^0}(X)]\vert\mathcal F_s\right]\\
&=\,&\tr\left[\mathbb E_{\mathbb Q}\left[(h\circ \theta^s)\cdot({\mathcal A_t^0\circ\theta^s})^* (O)\vert \mathcal F_s\right] {\mathcal A_s^0}(X)\right]\\
&=\,&\tr\left[\mathbb E_{\mathbb Q}\left[(h\circ \theta^s)\cdot({\mathcal A_t^0\circ\theta^s})^* (O)\right] {\mathcal A_s^0}(X)\right]\\
&=\,& \tr\left[\mathbb E_{\mathbb Q}\left[h\cdot(({\mathcal A_t^0})^* (O))\right] {\mathcal A_s^0}(X)\right].
\end{eqnarray*}

Then by using Lemma \ref{lem:trick} with $O$ replaced by $\mathbb E_{\mathbb Q}\left[h\cdot({\mathcal A_t^0})^* (O)\right]$, we have for example for $\mathcal L$
$$\tr\left[\mathbb E_{\mathbb Q}\left[h\cdot({\mathcal A_t^0})^* (O)\,\,\mathcal L(X)\right]\right]=0.$$
Therefore for all measurable functions $h$ and all $t\geq 0,$ we have
$$\mathbb E_{\mathbb Q}\left[h\cdot\tr\left[({\mathcal A_t^0})^* (O)\,\,\mathcal L(X)\right]\right]=0.$$
The since $h$ is arbitrary it follows that
$$0=\tr[({\mathcal A_t^0})^* (O)\,\,\mathcal L(X)]=\tr[O\mathcal A_t^0 (\mathcal L(X))],$$
which was the required results. The same holds for the other super-operators.

To prove that $\Ns$ is indeed the largest subspace such that properties 1--4 hold one can recur to a common argument from the system-theoretic literature (see e.g. \cite[Property 2.6.8]{marrobasile} or \cite{kalman,kalman1969topics,wonham,vanhandel2008stabilityquantummarkovfilters}) which we include next for completeness: Assume $\Ws\subseteq\Bf(\Hc)$ is an operator space such that properties 1--4 hold for $\Ws$, then it is easy to prove that $\Ws\subseteq\Ns$. 

Let us now consider $\Ns^\perp$. By definition of $\Ns$ and of orthogonal complement (w.r.t. $\inner{\cdot}{\cdot}_{HS}$) we have that 
\[\Ns^\perp = \Span\{({\mathcal A_t^0})^* (O_j), \, t\geq0,\,\forall j=1,\dots,r \}.\]
By common properties of the orthogonal complement, we have that properties 1--4 imply that $\Ns^\perp$ is the smallest operators subspace such that: \begin{itemize}\item1a) $\Ns^\perp \supseteq \Span\{O_j\}_{j=1}^r$; \item 2a) \(\Ns^\perp\) is $\Lc^*$-invariant; \item 3a) \(\Ns^\perp\) is $\Gc_{D_j}^*$-invariant for all $j=1,\dots,p$; 
\item 4a) \(\Ns^\perp\) is $\Kc_{C_j}^*$-invariant for all $j=1,\dots,q$.
\end{itemize}
Consider then the super-operator algebra $$\Ts \equiv {\rm Alg}\left(\{\Ic, {\Lc}^*\}\cup\{\Gc_{D_j}^*\}_{j=1}^p \cup\{\Kc_{C_j}^*\}_{j=1}^q \right).$$ By properties 2a--4a we then have that $\Ns^\perp$ is invariant under the action of any super-operator contained in $\Ts$, i.e. $\forall X\in\Ns^\perp$, and $\forall\Ac\in\Ts$, $\Ac(X)\in\Ns^\perp$. Combining this with property 1a, the statement naturally follows.
\end{proof}

We can observe that Assumption \ref{ass:identity} directly implies that $\one\in\Ns^\perp$. 
We shall further notice that the superoperator algebra $\Ts$ includes the super-operator Lie algebra $\Ls = {\rm LieAlg}\{\Tilde{\Lc}^*,\Gc_{D_j}^*\}$ closed with respect to linear combination and the operation $[\Ac,\Bc] = \Ac\Bc-\Bc\Ac$ as it was defined in \cite{amini2019estimation}. The reason why we use the super-operator algebra $\Ts$ instead of the Lie algebra $\Ls$ as it is commonly done in bilinear system theory, see e.g. \cite{elliott2009bilinear}, is because we are here interested in the operator {\em space} that is generated by observables of interest evolved in Heisenberg picture instead of the set of observables itself. 

\subsection{Reduced linear filters}

In this subsection we formalize the intuition that $\Ns^\perp$ contains all the necessary degrees of freedom and we can thus restrict the original quantum filter onto it to obtain a reduced filter that correctly reproduces the processes $\{(\Theta_t^j)_{t\geq0}\}_{j=1}^r$. 

More precisely, let $\Ns$ be the non-observable subspace defined in equation \eqref{eq:non_observable_definition} and let $\Ns^\perp$ its orthogonal complement (w.r.t. $\inner{\cdot}{\cdot}_{HS}$), i.e. $\Bf(\Hc)=\Ns\oplus\Ns^\perp$. 
Let then $\Pi_{\Ns^\perp}$ be the orthogonal projector onto $\Ns^\perp$ and let $\Rc_{\Ns^\perp}:\Bf(\Hc)\to\Cb^\kappa$ and $\Jc_{\Ns^\perp}:\Cb^\kappa\to\Ns^\perp$ be full rank factors such that \[\Pi_{\Ns^\perp}=\Jc_{\Ns^\perp}\Rc_{\Ns^\perp}\qquad\text{ and }\qquad\Rc_{\Ns^\perp}\Jc_{\Ns^\perp} = \one_\kappa\in\Cb^{\kappa \times \kappa}\] where $\kappa=\dim(\Ns^\perp)$. 
Notice that the choice of the two factors $\Rc_{\Ns^\perp}$ and $\Jc_{\Ns^\perp}$ of $\Pi_{\Ns^\perp}$ is not unique and what follows works for any choice of the factors. None the less a possible choice of these two factors can be constructed as follows: let $\{E_k\}$ be an orthonormal operator basis for $\Ns^\perp$ and let $\{e_k\}$ be an orthonormal vector basis for $\Cb^\kappa$, then \[\Rc_{\Ns^\perp}(X) = \sum_{k} e_k \inner{E_k}{X}_{HS}\qquad\text{ and }\qquad\Jc_{\Ns^\perp}(x) = \sum_k E_k \inner{e_k}{x}_{\Cb^\kappa}\] for all $x\in\Cb^\kappa$ and where by $\inner{\cdot}{\cdot}_{\Cb^\kappa}$ we intend the standard euclidean inner product for $\Cb^\kappa$.

The restriction of SDE \eqref{eq:Belavkin-Zakai} onto the subspace $\Ns^\perp$ is given by the process $(v_t)\in\Cb^\kappa$ with initial condition $v_0=\Rc_{\Ns^\perp}(\rho_0)$ and evolving through the SDE
\begin{equation}
        dv_t = Q v_{t-} dt + \sum_{j=1}^p G_j v_{t-} dY_t^j + \sum_{j=1}^q (K_j-\one_\kappa) v_{t-} dN_t^j,
    \label{eq:linear_filter}
\end{equation}
where 
\begin{align}
    Q \equiv \Rc_{\Ns^\perp} {\Qc}\Jc_{\Ns^\perp}\in\Cb^{\kappa \times \kappa}, && 
    G_j \equiv \Rc_{\Ns^\perp}\Gc_{D_j}\Jc_{\Ns^\perp}\in\Cb^{\kappa \times \kappa}, && 
    K_j \equiv \Rc_{\Ns^\perp}\Kc_{C_j}\Jc_{\Ns^\perp}\in\Cb^{\kappa \times \kappa}.
    \label{eq:reduced_op_linear}
\end{align}
Notice that the driving increments of the SDEs \eqref{eq:linear_filter} and \eqref{eq:Belavkin-Zakai} are the same: $\{dW_t^j\}_{j=1}^p$ and $\{dN_t^j\}_{j=1}^{q}$. 
The observables can be reduced as well by considering \[\inner{O_j}{\tau_t}_{HS} = \inner{O_j}{\Pi_{\Ns^\perp}(\tau_t)}_{HS} = \inner{\Jc_{\Ns^\perp}^*(O_j)}{\Rc_{\Ns^\perp}(\tau_t)}_{\Cb^\kappa} =  \inner{\zeta_j}{\Rc_{\Ns^\perp}(\tau_t)}_{\Cb^\kappa}\]
where we defined $\zeta_j \equiv \Jc_{\Ns^\perp}^*(O_j)\in\Cb^\kappa$ so that $\Lambda_t^j=\inner{\zeta_j}{v_{t}}_{\Cb^\kappa}$.
We next prove that $\tr[O_j\tau_t] = \inner{\zeta_j}{v_t}_{\Cb^\kappa}$ for all $t\geq0$.
\begin{theorem}
\label{thm:linear_filter_reduction}
    Consider the two processes $(\tau_t)_{t\geq0}$ and $(v_t)_{t\geq0}$ driven by the linear SDEs \eqref{eq:Belavkin-Zakai} and \eqref{eq:linear_filter} with the same output processes $\{(Y_t^j)_{t\geq0}\}_{j=1}^p$ and $\{(N_t^j)_{t\geq0}\}_{j=1}^q$ and with initial conditions $\rho_0$ and $v_0=\Rc_{\Ns^\perp}(\rho_0)$. Then on $(\Omega,\mathcal F,(\mathcal F_t),\mathbb Q),$ we have $\mathbb{Q}$-almost surely:
    \begin{itemize}
        \item $v_t=\Rc_{\Ns^\perp}(\tau_t)$, for all $t\geq0$;
        \item $\tr[O\tau_t] = \inner{\Jc_{\Ns}^\perp(O)}{v_t}_{\Cb^\kappa}$, for all initial conditions, $\rho_0\in\Df(\Hc)$, for all $t\geq0$ and for all $O\in\Ns^\perp$. 
    \end{itemize} 

    Furthermore, the reduced linear filter \eqref{eq:linear_filter} is a linear filter of minimal dimension such that the  conditions above hold.
\end{theorem}
\begin{proof}
    Let us start by recalling that, by Proposition \ref{prop:non-observable}, we have that $\Ns$ is $\Lc$-,  $\Gc_{D_j}$- and $\Kc_{C_j}$-invariant. As a consequence, \begin{eqnarray*}\Rc_{\Ns^\perp}  \Lc &=\,& \Rc_{\Ns^\perp}  \Lc  (\Pi_{\Ns^\perp} + \Pi_\Ns)\\ &=\,& \Rc_{\Ns^\perp}  \Lc  \Pi_{\Ns^\perp} + \Rc_{\Ns^\perp}  \Lc  \Pi_{\Ns}\\ &=\,& \Rc_{\Ns^\perp}  \Lc  \Pi_{\Ns^\perp} + \cancel{\Rc_{\Ns^\perp}  \Pi_{\Ns}}  \Lc  \Pi_{\Ns}\\& =\,& \Rc_{\Ns^\perp}  \Lc  \Pi_{\Ns^\perp}
    \end{eqnarray*}where $\Pi_\Ns$ is the orthogonal projector onto $\Ns$ and we used the fact that $\Ns\in\ker\Pi_{\Ns^\perp}$ and $\Rc_{\Ns^\perp}$ is a full rank factor of $\Pi_{\Ns^\perp}$ hence $\Rc_{\Ns^\perp}\Pi_\Ns = 0$. Similarly, $\Rc_{\Ns^\perp}\Gc_{D_j} = \Rc_{\Ns^\perp}\Gc_{D_j}\Pi_{\Ns^\perp}$ and $\Rc_{\Ns^\perp}\Kc_{C_j} = \Rc_{\Ns^\perp}\Kc_{C_j}\Pi_{\Ns^\perp}$.
    
    Recalling than that 
    \[\tau_t = \tau_0 + \int_{0}^t \Qc(\tau_{s-}) ds + \sum_{j=1}^p \int_{0}^t \Gc_{D_j}(\tau_{s-}) dY_s^j + \sum_{j=1}^q \int_{0}^t (\Kc_{C_j}-\Ic)(\tau_{s-}) dN_s^j \]
    and defining $\Tilde{v}_t = \Rc_{\Ns^\perp}(\tau_t)$ we have 
    \begin{align*}
        \Tilde{v}_t =\,& \underbrace{\Rc_{\Ns^\perp}(\tau_0)}_{{v}_0}  + \int_0^t \Rc_{\Ns^\perp}  \Qc (\tau_{s-}) d_s + \sum_{j=1}^p \int_{0}^t \Rc_{\Ns^\perp}\Gc_{D_j}(\tau_{s-}) dY_s^j + \sum_{j=1}^q \int_{0}^t \Rc_{\Ns^\perp}(\Kc_{C_j}-\Ic)(\tau_{s-}) dN_s^j\\
        =\,& {v}_0 + \int_0^t \underbrace{\Rc_{\Ns^\perp}  \Qc  \Jc_{\Ns^\perp}}_{Q}  \underbrace{\Rc_{\Ns^\perp} (\tau_{s-})}_{\Tilde{v}_{s-}} d_s + \sum_{j=1}^p \int_{0}^t \underbrace{\Rc_{\Ns^\perp}\Gc_{D_j}  \Jc_{\Ns^\perp}}_{G_j}  \underbrace{\Rc_{\Ns^\perp} (\tau_{s-})}_{\Tilde{v}_{s-}} dY_s^j + \\
        &+ \sum_{j=1}^q \int_{0}^t \left(\underbrace{\Rc_{\Ns^\perp}\Kc_{C_j}\Jc_{\Ns^\perp}}_{K_j}-\underbrace{\Rc_{\Ns^\perp}\Jc_{\Ns^\perp}}_{\one_\kappa}\right)\underbrace{\Rc_{\Ns^\perp}(\tau_{s-})}_{\Tilde{v}_{s-}} dN_s^j\\
        =\,& {v}_0 + \int_0^t {Q} {\Tilde{v}_{s-}} d_s + \sum_{j=1}^q \int_{0}^t {G_j} {\Tilde{v}_{s-}} dY_s^j + \sum_{j=1}^q \int_{0}^t ({K_j}-\one_\kappa){\Tilde{v}_{s-}} dN_s^j
    \end{align*}
    and thus \(v_t=\Tilde{v}_t=\Rc_{\Ns^\perp}(\tau_t)\), for all $t\geq0$ $\mathbb{Q}$-almost surely (this comes from the uniqueness of strong solution of involved SDEs \cite[Section V.3]{protter2005stochastic}). 
    We can then notice that, we have for all $O\in\Ns^\perp$  for all $t\geq 0$ 
    \[\tr[O \tau_t] = \tr[O \Pi_{\Ns^\perp}(\tau_t)] = \tr[O \Jc_{\Ns^\perp}\Rc_{\Ns^\perp}(\tau_t)] = \tr[\Jc_{\Ns^\perp}^*(O) \Rc_{\Ns^\perp}(\tau_t)] = \inner{\Jc_{\Ns^\perp}^*(O)}{v_t}_{\Cb^\kappa}.\]

To prove the minimality of the reduced filter we proceed by contradiction. 
Assume there exists an operator space $\Ws\subsetneq\Bf(\Hc)$  such that $\dim(\Ws)<\dim(\Ns^\perp)$ and allows for an exact model reduction. That is, let the orthogonal projector $\Pi_\Ws$ on $\Ws$ to be factorized as $\Pi_\Ws = \Jc_\Ws\Rc_\Ws$ 
such that, with $w_0 = \Rc_{\Ws}(\tau_0)$ and $w_t$ obtained from the dynamics reduced on $\Ws,$ we have a valid reduction such that \begin{equation}\label{eq:redx1}\tr[O^*\tau_t] = \inner{\Jc_\Ws^*(O)}{w_t}_{\Cb^s},\end{equation} for all $\tau_0\in\Df(\Hc)$, all $O\in\Ns^\perp$ and for all time $t\geq0$. Since \eqref{eq:redx1} holds for all density operators, by linearity it also holds for any $X_0\in\Bf(\Hc)$ as density operators generate the full operator space.

Since we assumed that $\dim(\Ws)<\dim(\Ns^\perp)$ we have that $\dim(\ker\Rc_\Ws)>\dim(\Bf(\Hc))-\dim(\Ns^\perp)$ and $\ker\Rs_\Ws\cap\Ns^\perp \neq\{0\}$. Therefore, there exists $X_0\neq 0 \in \Ns^\perp,$  such that $\Rc_{\Ws}(X_0)=0.$
Consider as initial condition $\tau_0=\frac{\one}{\tr{(\one )}}.$ Given any realization of the model noises, the output generated by the reduced model on $\Ws$ with initial condition $w_0=\Rc_\Ws(\tau_0)$ is identical to that generated by the initial condition \(\tilde w_0=\Rc_\Ws(\tilde\tau_0)\) with \[ \tilde\tau_0=\tau_0+X_0,\]
since $X_0\in\ker(\Rc_\Ws).$
Notice that despite $\tilde \tau_0$ does not correspond to a density operator, the linear reductions must still work as we noted above.
On the other hand, Proposition \ref{prop:non-observable}  ensures that $\tr[O_j^* X_t] = 0$ for all $j$ and $t\geq 0$ if and only if $X_0\in \Ns,$ where $(X_t)$ is the solution of linear evolution corresponding to initial condition $X_0$. Since we chose $X_0\in\Ns^\perp,$ there must exist a $j$ such that $\tr[O_j^* X_t]\neq 0$ for some $t\geq 0.$ 
Thus, the outputs of the original model differ when considering  initial conditions $\tau_0$ and $\tilde\tau_0.$ 
This contradicts the fact that the evolution on $\Ws$ were an exact reduction of the original model.
\end{proof}

With the previous theorem we showed that the reduced process $(v_t)_{t\geq0}$ reproduces the processes $\tr[O_j \tau_t]$. However we are actually interested in reproducing the processes $\tr[O_j \rho_t]$ which can be retrieved by a-posteriori re-normalization.
\begin{corollary}
\label{cor:normalization}
    Under the assumptions of Theorem \ref{thm:linear_filter_reduction} and assuming $\one\in\Ns^\perp$ we have  $$\tr[O_j \rho_t] = \frac{\inner{\zeta_j}{v_t}}{\inner{e}{v_t}},$$ for all $t\geq 0$ and all $j=1,\ldots,r$, where we recall that $\zeta_j \equiv \Jc_{\Ns^\perp}^*(O_j)\in\Cb^{f}$ and where we put $e\equiv\Jc_{\Ns^\perp}^*(\one)$.
\end{corollary}

\begin{proof}
    First we shall notice that, for all $X\in\Ns^\perp$, we have \[\inner{X}{\tau_t}_{HS} = \tr[X\tau_t] = \tr[X \Pi_{\Ns^\perp}(\tau_t)] = \tr[\Jc^*_{\Ns^\perp}(X) \Rc_{\Ns^\perp}(\tau_t)] = (\Jc^*_{\Ns^\perp}(X))^* v_t = \inner{\Jc^*_{\Ns^\perp}(X)}{v_t}_{\Cb^\kappa}.\] The rest follows from the fact that $\one\in\Ns^\perp$ and $\tr[O_j\rho_t] = \frac{\tr[O_j\tau_t]}{\tr[\tau_t]}$. 
\end{proof}

This Corollary implies that, in order to reproduce the processes $\{(\Theta_t^j)_{t\geq0}\}_{j=1}^r$ using the reduced process $(v_t)_{t\geq0}$, it is sufficient to ensure that $\one\in\Ns^\perp$. As we saw in Lemma \ref{prop:observable_construction}, the easiest way to ensure this holds, is to assume that $\one\in\{O_j\}_{j=1}^r$, that is Assumption \ref{ass:identity}. This allows us to ensure that, under Assumption \ref{ass:identity}, the reduced filter $\Sigma_L$ is capable of reproducing the processes $\{(\Theta_t^j)_{t\geq0}\}_{j=1}^r$.
 
\section{Reduced quantum filters}
\label{sec:reduced_quantum_filters}
We can notice that Theorem \ref{thm:linear_filter_reduction} provides a partial solution to the model reduction problem presented in the Introduction. In fact, if one is only interested in finding a linear filter that correctly reproduces the expectation values of the observables $\{O_j\}_{j=1}^r$ then one can find a minimal one by choosing as state $v_t=\Rc_{\Ns^\perp}(\tau_t)$. 
The main limitation of this approach is the fact that the reduced filter described by equation \eqref{eq:linear_filter}  does not correspond to a valid SME, i.e. if $v_0$ is not properly initialized, the output one obtains might be non-physical, in the sense that they might not be replicated by any evolution in the density operator set, and the resulting equation is hard to interpret.

In order to ensure that the reduced filter is a valid quantum model we leverage the algebraic framework developed in recent works, \cite{tit2023, letter2024, prxq2024}. We next collect and summarize the results that are necessary to continue for the reader's convenience.

\subsection{Properties of $*$-algebras and conditional expectations}
We here consider operator *-algebras  $\As\subseteq\Bf(\Hc),$ that in this finite dimensional settings, are operator subspaces closed with respect to the standard matrix product and adjoint, i.e. for all $X,Y\in\As$ then $\alpha X+\beta Y\in \As$ for all $\alpha,\beta\in\Cb$, $XY\in\As$, and $X^*, Y^* \in\As$ \cite{blackadar2006operator}. An algebra $\As$ is unital if it contains the identity operator, i.e. $\one \in\As$. A key result regarding the structure  of *-algebras is known as Wedderburn decomposition \cite{wedderburn1908hypercomplex}. 
Given a unital algebra $\As\subset\Bf(\Hc)$, there exists a decomposition of the Hilbert space
\[\Hc = \bigoplus_{k=1}^K \Hc_{F,k}\otimes\Hc_{G,k} \]
for some $K\in\Nb$, 
and there exists a unitary operator $U\in\Bf(\Hc)$, $U U^* = \one$ that decomposes the algebra $\As$:
\begin{align}
    \As &= U\left(\bigoplus_{k=1}^K \Bf(\Hc_{F,k}) \otimes \one_{G,k}\right)U^*\label{eq:wedderburn}.
\end{align}

A {\em conditional expectation} onto a unital $*$-algebra $\As$, $\CE_\As:\Bf(\Hc)\to\Bf(\Hc)$ with ${\rm Im}\CE_\As = \As$, is a completely positive projector, i.e. $\CE_\As[X] = X$ for all $X\in\As$ and $\CE_\As[X^* X]\geq 0$ for all $X\in\Bf(\Hc)$ \cite{blackadar2006operator}.
Note, that while idempotent, a conditional expectation need not be self-adjoint, and in general is thus not an orthogonal projection. The dual (w.r.t. $\inner{\cdot}{\cdot}_{HS}$) of a conditional expectation $\CE_\As^*$ is called a \textit{state extension} and is a CPTP projector onto its image. Conditional expectations also assume a specific structure related to the Wedderburn decomposition \eqref{eq:wedderburn}. In particular, there exists a set of full-rank density operators $\{\sigma_k\in\Df(\mathcal{H}_{G,k})\}$ such that, for all $X\in\mathcal{B(H)}$
    \begin{align}
         \CE_\mathscr{A}[X] 
         = U\left(\bigoplus_{k=1}^{K} \tr_{\Hc_{G,k}}\left[(V_k^* X V_k)(\one_{F,k}\otimes\sigma_k)\right]\otimes\one_{G,k}\right)U^*
        \label{eqn:cond_exp_blocks}
    \end{align}
where $V_k$ are non-square isometries from $\Hc_{F,k}\otimes\Hc_{G,k}$ to $\Hc$ and such that $V_k V_k^*=\one_{F,k}\otimes \one_{G,k}$, see, e.g. \cite{wolf2012quantum}.

Trough the Wedderburn decomposition of $\As$ one can further reduce the representation of $\As$ by avoiding the repeated blocks created by the tensor products $\otimes\one_{G,k}$. More precisely, given a unital algebra $\As= U\left(\bigoplus_{k=1}^K \Bf(\Hc_{F,k}) \otimes \one_{G,k}\right)U^*$ 
we can observe that it is isomorphic to $\check{\As} = \bigoplus_{k=1}^K \Bf(\Hc_{F,k})$, $\check{\As}\subseteq \Cb^{m\times m}$ with $m=\sum_{k=1}^K \dim(\Hc_{F,k})$. 
As proven in \cite[Theorem 1]{tit2023}, the dual of a conditional expectation $\CE_\As^*$ can be factorized in two CPTP factors $\Rc_{\As}:\Bf(\Hc)\to\check{\As}$ and $\Jc_{\As}: \check{\As}\to\Bf(\Hc)$ with ${\rm Im}\Jc_{\As} = {\rm Im}\CE_\As^*$ such that $\CE_\As^* = \Jc_{\As}\Rc_{\As}$ and $\Rc_{\As}\Jc_{\As} = \Ic_{\check{\As}}$.
Explicitly, one can derive the block-structure of the CPTP  linear maps of interest to be 
\begin{equation}
\begin{split}
    \Rc_{\As}(X) &=\bigoplus_{k=1}^K \tr_{\Hc_{G,k}}(V_k^* X V_k)=\bigoplus_{k=1}^K X_{F,k}=\check X,\\
    \Jc_{\As}(\check X)&=U\left(\bigoplus_{k=1}^K X_{F,k} \otimes\sigma_{k}\right)U^*.
\end{split}
    \label{eq:cond_exp_factorization}
\end{equation}

The main reason why we are interested in conditional expectations, their duals, and their factorizations, is the following. Consider a CPTP map $\Ac:\Bf(\Hc)\to\Bf(\Hc)$, an algebra $\As$ and a conditional expectation $\CE_\As$. Then the restriction of $\Ac$ onto ${\rm Im}(\CE_\As^*)$ is $\check{\Ac} = \Rc_{\As}  \Ac  \Jc_{\As}$ with $\check{\Ac}:\check{\As}\to\check{\As}$ and most importantly is CPTP since $\Rc_{\As}$, $\Jc_{\As}$ and $\Ac$ are CPTP. A similar result, proven in \cite{prxq2024}, holds for Lindblad generators and we report the statement here for the reader's convenience.

\begin{theorem}[\cite{prxq2024}]
\label{thm:Lindblad reduction}
 Consider a $*$-subalgebra $\mathscr{A}$ of $\mathcal{B(H)}.$ 
 Let then $\Rc_\As$ and $\Jc_\As$ be the CPTP factorization of $\CE_\mathscr{A}^* = \es_\As\rs_\As$ as defined in equation \eqref{eq:cond_exp_factorization}. 
 Then for any Lindblad generator $\Lc(\cdot)$, its reduction onto ${\rm Im}(\CE_\As^*)$, $\check{\Lc}(\cdot)\equiv\Rc_\As\Lc\Jc_{\As}(\cdot)$ is also a Lindblad generator, i.e. $\check{\Lc}:\check{\As}\to\check{\As}$ and $\{e^{\check{\Lc}t}\}_{t\geq 0}$ is a quantum dynamical semigroup. 
\end{theorem}

Furthermore, any unital algebra $\As$ admits a CPTP and a unital conditional expectation $\CE_\As$ which is an orthogonal projector onto $\As$, i.e. $\CE_\As^2 = \CE_\As = \CE_\As^*$ and $\CE_\As(\one) = \one$. For orthogonal conditional expectations, the representation \eqref{eqn:cond_exp_blocks} holds with $\sigma_k = \frac{\one_{G_k}}{\dim(\Hc_{G,k})}$. In the rest of this work we will only focus on orthogonal conditional expectations.

The main idea that we use in the next subsection, is to use this property to ensure that the reduced model is a valid quantum model, i.e. $\check{\Lc}$ is a Lindblad generator. In particular we compute an algebra $\As$ that contains $\Ns^\perp$, i.e. $\As\supseteq\Ns^\perp$, and define $\CE_\As$ to be the orthogonal conditional expectation onto $\As$. We then use its factor $\Rc_{\As}$ and $\Jc_{\As}$ to compute the reduced model. Notice that we can pick $\As$ to be the smallest algebra that contains $\Ns^\perp$, i.e. $\As=\alg(\Ns^\perp)$ in case we are interested in the smallest model but we could also consider larger algebras if for example we want to preserve other properties of the model. 
Moreover, from Assumption \ref{ass:identity}, we have $\one\in\Ns^\perp$, hence any algebra $\As$ such that $\As\supseteq \Ns^\perp$, including $\alg(\Ns^\perp)$, is unital and we will thus work under this assumption in the following.

The relation between the dimensions of the involved spaces is clearly $\dim\Ns^\perp\leq\dim\As\leq\dim\Bf(\Hc),$ while the dimension of the smallest algebra that contains the compact representation of the algebra $\check{\As}\subseteq \Cb^{m\times m}$ is entirely determined by the Wedderburn decomposition, since it depends on the number of blocks and their dimensions. For this reason, $m$ cannot be determined a-priori.

Notice that, while considering a conditional expectation $\CE_\As$ onto a unital algebra is sufficient to ensure that the reduced generator $\check{\Lc}$ is Lindblad, this is not necessary. It is thus, in principle, possible to find a reduced Lindblad generator without using conditional expectations and their duals, but to the best of the author's knowledge necessary conditions for a reduced generator to be Lindblad are not known. 

\begin{remark}
    Proving that $\alg(\Ns^\perp)$ is the smallest operator space that admits the existence of (the dual of) a conditional expectation is non trivial. While in classical probability theory the existence of a conditional expectation fixing any state is always guaranteed, for quantum systems this need not be the case. Necessary and sufficient conditions for this to be the case are provided in Takesaki's modular theory and its specializations to the finite dimensional case \cite{TAKESAKI1972306,petz2007quantum,tit2023}. These issues have been addressed in \cite[Theorem 4]{tit2023}. We refer the reader to \cite[Sections IV and V]{tit2023} for more details on the matter. Interestingly, the construction of minimal algebras that contain a set of states of interest emerges also when one aims to generalize the notion of sufficient statistics to the quantum domain \cite{jenvcova2006sufficiency}, or when one considers the set of maps that do not disturb a given subset of states \cite{PhysRevA.66.022318}.
\end{remark}

\subsection{Reduced quantum filters}

The main idea of the following is to construct a stochastic process which we denote  $(\check{\tau}_t)_{t\geq0}$ by restricting the SDE \eqref{eq:Belavkin-Zakai} onto a $*$-algebra $\As$ that contains $\Ns^\perp$.
Let us consider $\As$ to be a $*$-algebra such that $\Ns^\perp\subseteq\As$ (e.g. $\As=\alg(\Ns^\perp)$). Let then $\CE_{\As}$ be the orthogonal conditional expectation onto $\As$ and let $\Rc_{\As}$ and $\Jc_{\As}$ be its full rank factors as defined in \eqref{eq:cond_exp_factorization}, i.e. $\CE_\As = \Jc_{\As}\Rc_{\As}$ and $\Rc_{\As}\Jc_{\As} = \Ic_{\check{\As}}$. We can then consider the stochastic process $(\check{\tau}_t)\in\check{\As}$ defined over the $*$-algebra $\check{\As}\subseteq\Bf(\check{\Hc})$ with initial condition $\check{\tau}_0 = \Rc_{\As}(\rho_0)\in\check{\As}$ and evolving through the linear SDE 
\begin{equation}
    d\check{\tau}_t = \check{\Qc}(\check{\tau}_{t-}) dt + \sum_{j=1}^p \check{\Gc}_{D_j}(\check{\tau}_{t-}) dY_t^j + \sum_{j=1}^q \left[\check{\Kc}_{C_j}-\Ic_{\check{\As}}\right](\check{\tau}_{t-}) dN_t^j,
    \label{eq:quantum_reduced_filter_linear}
\end{equation}
where
\begin{align}
\begin{split}
    \check{\Lc}:\Df(\check{\As})&\to\check{\As}\\
    \check{\rho}&\mapsto 
    \Rc_{\As} {\Lc}\Jc_{\As}(\check{\rho})
\end{split}&&
\begin{split}
    \check{\Gc}_{D_j}:\Df(\check{\As})&\to\check{\As}\\
    \check{\rho}&\mapsto 
    \Rc_{\As} {\Gc}_{D_j}\Jc_{\As}(\check{\rho})
\end{split}&&
\begin{split}
    \check{\Kc}_{C_j}:\Df(\check{\As})&\to\check{\As}\\
    \check{\rho}&\mapsto
    \Rc_{\As} {\Gc}_{D_j}\Jc_{\As}(\check{\rho})
\end{split}
\end{align}
and $$\check{\Qc} \equiv \check{\Lc} +\Ic_{\check{\As}} - \sum_{j=1}^q \check{\Kc}_{C_j}$$ where $\Ic_{\check{\As}}$ is the identity superoperator over the algebra $\check{\As}$. Note that: the SDE \eqref{eq:quantum_reduced_filter_linear} is driven with the same innovation processes $\{(dY_t^j)_{t\geq0}\}_{j=1}^p$, $\{(dN_t^j)_{t\geq0}\}_{j=1}^q$ as the SDE \eqref{eq:Belavkin-Zakai}, and, by Theorem \ref{thm:Lindblad reduction}, $\check{\Lc}$ is a Lindblad generator.
The observables of interest $O_j$ are also reduced to obtain 
$\check{O}_j\equiv \Jc_{\As}^*(O_j)\in\check{\As}$. We next prove that $\tr[O_j\tau_t] = \tr[\check{O}_j \check{\tau}_t]$ for all $O_j$ and for all $t\geq0$. 
\begin{theorem}
\label{thm:SDE_reduction}
    Consider the two processes $(\tau_t)_{t\geq0}$ and $(\check{\tau}_t)_{t\geq0}$ evolving through the linear SDEs \eqref{eq:Belavkin-Zakai} and \eqref{eq:quantum_reduced_filter_linear} with initial conditions $\tau_0 = \rho_0$ and $\check{\tau}_0 = \Rc_\As(\rho_0)$. 
    Then, for all initial conditions, $\rho_0\in\Df(\Hc)$, for all $t\geq0$ and for all $O_j, j=1,\ldots,r$ we have $$\tr[O_j\tau_t] = \tr[\check{O}_j \check{\tau}_t]$$ with $\check{O}_j = \Jc_{\As}^*(O_j)$. 
\end{theorem}
\begin{proof}
    Let us denote by $\Pi_{\Ns^\perp}$ the orthogonal projection onto $\Ns^\perp$. 
    Recalling that by Proposition \ref{prop:non-observable} we have that $\Ns$ is both $\Lc$- and $\Gc_{D_j}$-invariant for all $\{D_j\}_{j=1}^p$ and $\Kc_{C_j}$-invariant for all $\{C_j\}_{j=1}^q$ we have $\Pi_{\Ns^\perp}  \Lc  = \Pi_{\Ns^\perp}  \Lc  \Pi_{\Ns^\perp}$ and, similarly, $\Pi_{\Ns^\perp} \Gc_{D_j} = \Pi_{\Ns^\perp}\Gc_{D_j}\Pi_{\Ns^\perp}$, $\Pi_{\Ns^\perp} \Kc_{C_j} = \Pi_{\Ns^\perp}\Kc_{C_j}\Pi_{\Ns^\perp}$. 

    We first want to prove that $\Pi_{\Ns^\perp}(\tau_t) = \Pi_{\Ns^\perp}\Jc_\As(\check{\tau}_t)$ for all $t\geq0$ and for all initial conditions $\rho_0\in\Df(\Hc)$. Defining $X_t \equiv \Pi_{\Ns^\perp}(\tau_t)\in\Ns^\perp$ we then have:
    \begin{align*}
         X_t =\,&\Pi_{\Ns^\perp}(\rho_0) + \int_{0}^t \Pi_{\Ns^\perp}\Qc(\tau_{s-}) ds + \sum_{j=1}^p \int_{0}^t \Pi_{\Ns^\perp}\Gc_{D_j}(\tau_{s-}) dY_s^j + \sum_{j=1}^q \int_{0}^t \Pi_{\Ns^\perp}[\Kc_{C_j}-\Ic](\tau_{s-}) dN_s^j \\
        =\,& \Pi_{\Ns^\perp}(\rho_0) + \int_{0}^t \Pi_{\Ns^\perp}\Qc\underbrace{\Pi_{\Ns^\perp} (\tau_{s-})}_{X_{s-}} ds + \sum_{j=1}^p \int_{0}^t \Pi_{\Ns^\perp}\Gc_{D_j}\underbrace{\Pi_{\Ns^\perp}(\tau_{s-})}_{X_{s-}} dN_s^j +\\&
        +\sum_{j=1}^q \int_{0}^t \Pi_{\Ns^\perp}[\Kc_{C_j}-\Ic]\underbrace{\Pi_{\Ns^\perp}(\tau_{s-})}_{X_{s-}} dN_s^j \\
        =\,& \Pi_{\Ns^\perp}(\rho_0) + \int_{0}^t \Pi_{\Ns^\perp}\Qc ({X_{s-}}) ds + \sum_{j=1}^p \int_{0}^t \Pi_{\Ns^\perp}\Gc_{D_j}({X_{s-}}) dN_s^j 
        +\sum_{j=1}^q \int_{0}^t \Pi_{\Ns^\perp}[\Kc_{C_j}-\Pi_{\Ns^\perp}]({X_{s-}}) dN_s^j.
    \end{align*}
    On the other hand, by defining $\Gamma_t\equiv \Pi_{\Ns^\perp}\Jc_\As(\check{\tau}_t)\in\Ns^\perp$, we have: 
    \begin{align*}
        \Gamma_t =\,& \Pi_{\Ns^\perp}\Jc_{\As}(\check{\tau}_0) + \int_{0}^t \Pi_{\Ns^\perp}\Jc_{\As}\check{\Qc}(\check{\tau}_{s-}) ds + \sum_{j=1}^p \int_{0}^t \Pi_{\Ns^\perp}\Jc_{\As}\check{\Gc}_{D_j}(\check{\tau}_{s-}) dY_s^j + \\&+ \sum_{j=1}^q \int_0^t \Pi_{\Ns^\perp}\Jc_{\As} [\check{\Kc}_{C_j}-\Ic_{\check{\As}}](\check{\tau}_{s-}) dN_t^j\\
        =\,& \Pi_{\Ns^\perp}\CE_\As (\rho_0) + \int_{0}^t \Pi_{\Ns^\perp}\CE_\As{\Qc}  {\Jc_{\As}(\check{\tau}_{s-})} ds + \sum_{j=1}^p \int_{0}^t \Pi_{\Ns^\perp}\CE_\As{\Gc}_{C_j} {\Jc_{\As}(\check{\tau}_{s-})} dY_s^j  + \\&+ \sum_{j=1}^q  \int_0^t \Pi_{\Ns^\perp} \CE_\As [{\Kc}_{C_j}-\Ic_{\check{\As}}]{\Jc_{\As}(\check{\tau}_{s-})} dN_t^j \\
        =\,& \Pi_{\Ns^\perp}(\rho_0) + \int_{0}^t \Pi_{\Ns^\perp}{\Qc}  \underbrace{\Pi_{\Ns^\perp}\Jc_{\As}(\check{\tau}_{s-})}_{\Gamma_{s-}} ds + \sum_{j=1}^p \int_{0}^t \Pi_{\Ns^\perp}{\Gc}_{C_j} \underbrace{\Pi_{\Ns^\perp}\Jc_{\As}(\check{\tau}_{s-})}_{\Gamma_{s-}} dY_s^j  + \\&+ \sum_{j=1}^q  \int_0^t \Pi_{\Ns^\perp} [{\Kc}_{C_j}-\Ic_{\check{\As}}]\underbrace{\Pi_{\Ns^\perp}\Jc_{\As}(\check{\tau}_{s-})}_{\Gamma_{s-}} dN_t^j \\
        =\,& \Pi_{\Ns^\perp}(\rho_0) + \int_{0}^t \Pi_{\Ns^\perp}{\Qc} (\Gamma_{s-}) ds + \sum_{j=1}^p \int_{0}^t \Pi_{\Ns^\perp}{\Gc}_{C_j} (\Gamma_{s-}) dY_s^j 
        +\sum_{j=1}^q \int_{0}^t \Pi_{\Ns^\perp}[\Kc_{C_j}-\Pi_{\Ns^\perp}](\Gamma_{s-}) dN_s^j,
    \end{align*}
    where we used the fact that, since $\As\supseteq\Ns^\perp$, we have $\CE_\As\Pi_{\Ns^\perp} = \Pi_{\Ns^\perp}$. We can then notice that the two processes $(X_t)_{t\geq0}$ and $(\Gamma_t)_{t\geq0}$ coincide by performing the change of variable $\Gamma_s = X_s$. This proves that $\Pi_{\Ns^\perp}(\tau_t) = \Pi_{\Ns^\perp}(\Tilde{\tau}_t)$ for all $t\geq0$ and for all initial conditions $\rho_0$. We can then recall that $\Ns$ is contained in $\ker \tr[O_j \cdot]$ for all $O_j$ and thus for all $O_j$, $\tr[O_j \tau_t] = \tr[O_j \Pi_{\Ns^\perp}(\tau_t)]$ while, on the other hand, $\tr[\check{O}_j\check{\tau}_t] = \tr[O_j \Jc_{\As}(\check{\tau}_t)] = \tr[O_j \Pi_{\Ns^\perp}  \Jc_{\As}(\check{\tau}_t)] = \tr[O_j \Pi_{\Ns^\perp}(\Tilde{\tau}_t)]$.
\end{proof}

Theorem \ref{thm:SDE_reduction} proves that the reduced process $(\check{\tau}_t)_{t\geq0}$ is such that $\tr[O_j\tau_t] = \tr[\check{O}_j\check{\tau}_t]$ for all $t\geq0$ and for all $\rho_0\in\Df(\Hc)$. This is however not sufficient to ensure that the process $\check{\tau}_t$ is capable of reproducing the processes $\{\Theta_t^j\}$. A sufficient (not necessary) condition to obtain this is to require that $\tr[\tau_t]=\tr[\check{\tau}_t]$ for all $t\geq0$ and all initial conditions. The next corollary shows that, under Assumption \ref{ass:identity}, since $\one\in\Ns^\perp$, we have that the reduced process $(\check{\tau}_t)_{t\geq0}$ correctly reproduces the processes $(\Theta_t^j),j=1,\ldots,r$. 
\begin{corollary}
\label{cor:normalization_quantum}
    Under the assumptions of Theorem \ref{thm:SDE_reduction} and assuming that $\one\in\Ns^\perp$ we have that $$\Theta_t^j = \frac{\tr[\check{O}_j\check{\tau_t}]}{\tr[\check{\tau}_t]},$$ for all $t\geq0$ and all $j=1,\ldots,r$.
\end{corollary}
The proof of this corollary is identical to the proof of Corollary \ref{cor:normalization}.

Up to this point we proved that the reduced process $(\check{\tau}_t)_{t\geq0}$ is capable of reproducing the output processes $\{\Theta_t^j\}$ under Assumption \ref{ass:identity}. Note that we have not yet shown that the linear SDE \eqref{eq:quantum_reduced_filter_linear} generates a valid quantum process. 

Before we move to deriving the reduced quantum filter we need to focus on the operators that define the reduced super-operators $\check{\Lc}(\cdot)$ (or more precisely $\Rc_\As[H,\Jc_\As(\cdot)]$  and $\check{\Dc}_L(\cdot)$), $\check{\Gc}_{D_j}(\cdot)$ and $\check{\Kc}_{C_j}(\cdot)$. For example, we know from Theorem \ref{thm:Lindblad reduction} that $\check{\Lc}$ is a Lindblad generator but we did not specify yet how to compute the reduced Hamiltonian and noise operators. The following Proposition takes care of this.
\begin{proposition}
\label{prop:operator_reduction}
    Given an unital algebra $\As\subseteq\Bf(\Hc)$ which is isomorphic to $\check{\As}\subseteq\Bf(\check{\Hc})$, and given the CPTP factorization of $\CE_\mathscr{A}^* = \Eb|_\As$, $\Rc_\As$ and $\Jc_\As$ as defined in equation \eqref{eq:cond_exp_factorization} we have that:
    \begin{enumerate}
        \item For any Hamiltonian $H\in\Hf(\Hc)$, there exists an operator $\check{H} = \check{H}^*\in\check{\As}$ such that $[\check{H},\check{\rho}] = \Rc_\As [H,\Jc_\As(\check{\rho})]$;
        \item For any operator $L\in\Bf(\Hc)$, there exist a set of operators $\{\check{L}_k\}_{k=1}^d\subset\Bf(\check{\Hc})$ such that       
        \[\check{\Dc}_{L} = \Rc_{\As} \Dc_{L} \Jc_\As = \sum_{k=1}^d \Dc_{\check{L}_k}; \]
        \item For any operator $D\in\Bf(\Hc)$, there exist an operator $\check{D}\in\check{\As}$ and a set of operators $\{\check{D}'_k\}_{k=1}^d\subset \Bf(\check{\Hc})$ such that \[\check{\Gc}_{D} = \Rc_\As \Gc_D \Jc_{\As} = \Gc_{\check{D}} \qquad\text{ and }\qquad \check{\Dc}_{D} = \Rc_{\As} \Dc_{D} \Jc_\As = \Dc_{\check{D}} + \sum_{k=1}^d \Dc_{\check{D}'_k};\]
        \item For any operator $C\in\Bf(\Hc)$, there exist a set of operators $\{\check{C}_k\}_{k=1}^d\subset \Bf(\check{\Hc})$ such that \[ \check{\Kc}_{C} = \Rc_\As \Kc_C \Jc_{\As} = \sum_{k=1}^d \Kc_{\check{C}_k} \qquad\text{ and }\qquad \check{\Dc}_{C} = \Rc_{\As} \Dc_{C} \Jc_\As = \sum_{k=1}^d \Dc_{\check{C}_k};\]
    \end{enumerate}
   where $d\equiv\max_{j=1,\dots,N} \dim(\Hc_{G,j})^2$ of the Wedderburn decomposition \eqref{eq:wedderburn} of $\As$.
\end{proposition}
This Proposition is constructively proven in Appendix \ref{sec:noise_operators_reduction}. 

Applying this Proposition to the superoperators that appear in the SDE \eqref{eq:quantum_reduced_filter_linear} we obtain: 
\begin{itemize}
    \item A reduced Hamiltonian $\check{H}=\check{H}^*\in\check{\As}$; 
    \item A set $\bigcup_{j=1}^m\{\check{L}_{j,k}\}_{k=1}^{d}$ of reduced noise operators that describe the interaction with an un-monitored Markovian bath; 
    \item Two sets of noise operators $\{\check{D}_j\}_{j=1}^{p}$ and $\bigcup_{j=1}^p\{\check{D}_{j,k}'\}_{k=1}^{d}$ that describe the effects of homodyne-type measurement on the reduced model; 
    \item A set $\bigcup_{j=1}^q\{\check{C}_{j,k}\}_{k=1}^{d}$ of noise operators that describe the effects of counting-type measurements on the reduced model.
\end{itemize}
Combining this into SDE \eqref{eq:quantum_reduced_filter_linear} we obtain:
\begin{align}
\begin{split}
    d\check{\tau}_t =& \left[\check{\tau}_{t-}-i[{\check{H}},\check{\tau}_{t-}] + \sum_{j=1}^m \sum_{k=1}^{d} \Dc_{\check{L}_{j,k}}(\check{\tau}_{t-}) + \sum_{j=1}^p \left( \underbrace{\sum_{k=1}^{d} \Dc_{\check{D}_{j,k}'}}_{\star} +  \underbrace{\Dc_{\check{D}_j}}_{\star\star} \right)(\check{\tau}_{t-}) + \sum_{j=1}^{q} \sum_{k=1}^{d} \left(\Dc_{\check{C}_{j,k}}-\Kc_{\check{C}_{j,k}}\right)(\check{\tau}_{t-}) \right]dt \\
    &+ \sum_{j=1}^{p}\Gc_{\check{D}_j}(\check{\tau}_{t-}) dY_t^j + \sum_{j=1}^{q}  \left[ \sum_{k=1}^{d} \Kc_{\check{C}_{j,k}} - \Ic_{\check{\As}} \right](\check{\tau}_{t-})dN_t^j.
\end{split}
    \label{eq:reduced_sde_extended}
\end{align}
From this equation we can observe that, for every noise operator $\check{D}_j$ associated to a diffusive innovation process $dY_t^j$, the reduced SDE correctly include a dissipative term $\Dc_{\check{D}_j}$ (denoted in equation \eqref{eq:reduced_sde_extended} by $\star\star$), but can also, in principle, include more dissipative terms that are not directly associated to the innovation process $dY_t^j$ (in the equation \eqref{eq:reduced_sde_extended} denoted by $\star$). Note that, by Corollary \ref{cor:reduced_operators}, this term (denoted by $\star$) is null if $D_j\in\As$. 

We can further observe that to every noise operator $C_j$ associated to a jump innovation process $dN_t^j$ in the original model, we have, in the reduced model, more than one noise operators that is associated to the same jump process $dN_t^j$, namely $\{\check{C}_{j,k}\}$. In other words, in the reduced models multiple jumps are observed simultaneously whenever a single jump is observed in the original model. 

Applying the Kallianpur-Striebel formula \cite[Theorem 6.2]{bouten2007introduction} (or a-posteriori re-normalization) we have $\check{\rho}_t = \frac{\check{\tau}_t}{\tr(\check{\tau}_t)}$. The process $\check{\rho}_t\in\check{\As}$ has initial condition $\check{\rho}_0 = \Rc_\As(\rho_0)$ and, applying the Ito rules to equation \eqref{eq:reduced_sde_extended} one obtains the SME 
\begin{equation}
\begin{split}
    d\check{\rho}_t =\,& \check{\Lc}(\check{\rho}_{t-}) dt + \sum_{j=1}^{p}\left[\Gc_{\check{D}_j}(\check{\rho}_{t-}) - \tr\left[\Gc_{\check{D}_j}(\check{\rho}_{t-})\right]\check{\rho}_{t-}\right](dY_t^j - \tr[\Gc_{\check{D}_j}(\check{\rho}_{t-})]) \\
    &+ \sum_{j=1}^{q}  \left[ \frac{\sum_{k=1}^{d}\Kc_{\check{C}_{j,k}}(\check{\rho}_{t-})}{\sum_{k=1}^{d}\tr[\Kc_{\check{C}_{j,k}}(\check{\rho}_{t-})]} - \check{\rho}_{t-} \right]\left(dN_t^{j} - \sum_{k=1}^{d}\tr[\Kc_{\check{C}_{j,k}}(\check{\rho}_{t-})]\right).
\end{split}
\label{eq:reduced_SME}
\end{equation}

\begin{theorem}
    \label{thm:linear_solves_prob_1}
    Consider the process $(\rho_t)_{t\geq0}$ evolving through the SME \eqref{eq:SME} with initial condition $\rho_0\in\Df(\Hc)$ and with driving processes $\{(Y_t^j)_{t\geq0}\}_{j=1}^p$ and $\{(N_t^j)_{t\geq0}\}_{j=1}^q$. Consider then the reduced process $(\check{\rho}_t)_{t\geq0}$ evolving through the SME \eqref{eq:reduced_SME}, initial condition $\check{\rho}_0=\Rc_\As(\rho_0)$ and the same two driving processes $\{(Y_t^j)_{t\geq0}\}_{j=1}^p$ and $\{(N_t^j)_{t\geq0}\}_{j=1}^q$. Then, under Assumptions \ref{ass:observables} and \ref{ass:identity}, we have that:
    \begin{itemize}
        \item for all $\rho_0\in\Df(\Hc)$ and for all $t\geq0$, $\tr[O_j\rho_t]=\tr[\check{O}_j \check{\rho}_t]$ for all $j=1,\dots,r$;
        \item $\tr[\Gc_{D_j}(\rho_t)] = \tr[\Gc_{\check{D}_j}(\check{\rho}_t)]$ for all $t\geq0$ and for all $j=1,\dots,p$;
        \item $\tr[\Kc_{C_j}(\rho_t)] = \sum_{k=1}^{d} \tr[\Kc_{\check{C}_{j,k}}(\check{\rho}_t)]$ for all $t\geq0$ and for all $j=1,\dots,q$.
    \end{itemize} 
\end{theorem}
\begin{proof}
    From Theorem \ref{thm:SDE_reduction} and Corollary \ref{cor:normalization_quantum}, under Assumption \ref{ass:identity}, we have that for all $\rho_0\in\Df(\Hc)$, for all $t\geq0$, we have that \[\tr[O_j \rho_t] = \frac{\tr[O_j \tau_t]}{\tr[\tau_t]} = \frac{\tr[\check{O}_j \check{\tau}_t]}{\tr[\check{\tau}_t]} = \tr[\check{O}_j\check{\rho}_t]\]
    for all $j=1,\dots,r$. Furthermore, using Assumption \ref{ass:observables}, Lemma \ref{lem:operator_reduction_dissipative}, and the above we have 
    \begin{align*}
        \tr[(D_j+D_j^*)\rho_t] &= \tr[(\check{D}_j+\check{D}_j^*)\check{\rho}_t]\\
        \tr[C_j\rho_t C_j^*] &= \tr[\Jc_\As^*(C_j^* C_j)\check{\rho}_t] = \sum_{\check{C}\in\Xc(C_j)}\tr[\check{C}\check{\rho}_t\check{C}^*],
    \end{align*}
    thus concluding the proof.
\end{proof}

\begin{remark}
    Notice that the proof of Theorem \ref{thm:SDE_reduction} holds for any algebra $\As$ that contains $\Ns^\perp$. $\alg(\Ns^\perp)$ is, by definition, the smallest such algebra so it is possibly the best choice in case one aim to find the smallest quantum model that reproduces the output dynamics, however, it is also possible to consider larger algebras than $\alg(\Ns^\perp)$ if this is necessary to impose further properties on the reduced model. For example we might have that $\alg(\Ns^\perp)$ is not $\Lc^*$-invariant but there might exist an algebra $\As\supseteq\Ns^\perp$ which is $\Lc^*$-invariant which might be preferable in certain cases. See the next section for an example of why one might desire invariant algebras. 
\end{remark}

\section{Algebra invariance and filter stability}

An important difference between the two reduced filters $\Sigma_L$ and $\Sigma_Q$ shall be noted. From Theorem \ref{thm:linear_filter_reduction} we have $v_t = \Rc_{\Ns^\perp}(\rho_t)$ for all $t\geq0$ while, from Theorem \ref{thm:SDE_reduction}, we have that $\check{\tau}_0 = \Rc_{\As}(\tau_0)$ but this does not hold in general for all $t$, i.e. there might exist some time $t>0$ such that $\check{\tau}_t \neq \Rc_{\As}(\tau_t)$. This is due to the fact that the algebra $\As$ is not necessarily $\Lc^*$-, $\Gc_{D_j}^*$ and $\Kc_{C_j}^*$-invariant. This prompts the following proposition.

\begin{proposition}
\label{thm:invariant_algebra}
    Under the assumptions of Theorem \ref{thm:SDE_reduction} we have that: $\check{\tau}_t = \Rc_\As(\tau_t)$ for all $t\geq0$ and for all $\tau_0\in\Df(\Hc)$ if and only if $\As$ is  $\Lc^*$-, $\Gc_{D_j}^*$- and $\Kc_{C_j}^*$-invariant for all $\{D_j\}_{j=1}^p$ and all $\{C_j\}_{j=1}^q$. 
\end{proposition}
\begin{proof}
    Let us denote $\Tilde{\tau}_t \equiv \Rc_\As(\tau_t)$. Assume that $\As$ is $\Lc^*$-, $\Gc_{D_j}^*$ and $\Kc_{C_j}^*$-invariant for all $\{D_j\}_{j=1}^p$ and all $\{C_j\}_{j=1}^q$, then we have that $\CE_{\As}  \Lc  = \CE_{\As}  \Lc  \CE_{\As}$, $\CE_{\As} \Gc_{D_j} = \CE_{\As}\Gc_{D_j}\CE_{\As}$ and $\CE_{\As} \Kc_{C_j} = \CE_{\As}\Kc_{C_j}\CE_{\As}$. Recalling that $\Rc_\As = \Rc_\As\CE_\As$ we have 
    \begin{align*}
        \Tilde{\tau}_t&=\Rc_\As(\tau_t) \\&= \underbrace{\Rc_\As(\tau_0)}_{\check{\tau}_0} + \int_0^t \Rc_\As\Qc(\tau_{s-}) ds + \sum_{j=1}^p \int_0^t \Rc_\As\Gc_{D_j}(\tau_{s-}) dW_s^j + \sum_{j=1}^q\int_0^t \Rc_\As [\Kc_{C_j}-\Ic](\tau_{s-}) dN_j^s\\
        &= \check{\tau}_0 + \int_0^t \Rc_\As\CE_\As\Qc\CE_\As(\tau_{s-}) ds + \sum_{j=1}^p \int_0^t \Rc_\As\CE_\As\Gc_{D_j}\CE_\As(\tau_{s-}) dW_s^j \\&\hphantom{=}+ \sum_{j=1}^q\int_0^t \Rc_\As \CE_\As [\Kc_{C_j}-\Ic]\CE_\As(\tau_{s-}) dN_j^s\\
        &= \check{\tau}_0 + \int_0^t \underbrace{\Rc_\As\Qc \Jc_\As}_{\check{\Qc}}\underbrace{\Rc_\As(\tau_{s-})}_{\Tilde{\tau}_t} ds + \sum_{j=1}^p \int_0^t \underbrace{\Rc_\As \Gc_{D_j} \Jc_\As}_{\check{\Gc}_{D_j}}\underbrace{\Rc_\As (\tau_{s-})}_{\Tilde{\tau}_t} dW_s^j \\&\hphantom{=}+ \sum_{j=1}^q\int_0^t \underbrace{\Rc_\As [\Kc_{C_j}-\Ic]\Jc_\As}_{\check{\Kc}_{C_j}-\Ic_{\check{\As}}}\underbrace{\Rc_\As(\tau_{s-})}_{\Tilde{\tau}_s} dN_j^s\\
        &= \check{\tau}_0 + \int_0^t {\check{\Qc}}({\Tilde{\tau}_{s-}}) ds + \sum_{j=1}^p \int_0^t {\check{\Gc}_{D_j}}({\Tilde{\tau}_{s-}}) dW_s^j + \sum_{j=1}^q\int_0^t [{\check{\Kc}_{C_j}-\Ic_{\check{\As}}}]({\Tilde{\tau}_{s-}}) dN_j^s,
    \end{align*}
    where 
    \begin{align*}
        \check{\tau}_t &= \check{\tau}_0 + \int_0^t {\check{\Qc}}(\check{\tau}_{s-}) ds + \sum_{j=1}^q \int_0^t {\check{\Gc}_{D_j}}(\check{\tau}_{s-}) dW_s^j + \sum_{j=1}^p \int_0^t {\check{\Gc}_{D_j}}({\Tilde{\tau}_{s-}}) dW_s^j + \sum_{j=1}^q\int_0^t [{\check{\Kc}_{C_j}}-\Ic_{\check{\As}}]({\check{\tau}_{s-}}) dN_j^s,
    \end{align*}
    which implies $\Tilde{\tau}_t = \check{\tau}_t$ for all $t\geq0$ by strong uniqueness of the solution.

    Assume now that $\Rc_\As(\tau_t) = \check{\tau}_t$ for all $t\geq0$ and for all $\tau_0\in\Df(\Hc)$. Then, it holds 
    \begin{align*}
    &\cancel{\Rc_\As(\tau_0)} + \int_0^t \Rc_\As\Qc(\tau_{s-}) ds + \sum_{j=1}^q \int_0^t \Rc_\As\Gc_{D_j}(\tau_{s-}) dW_s^j+\sum_{j=1}^q \int_{0}^t \Rc_\As(\Kc_{C_j}-\Ic)(\tau_{s-}) dN_s^j \\&= \cancel{\check{\tau}_0} + \int_0^t {\check{\Qc}}(\check{\tau}_{s-}) ds + \sum_{j=1}^q \int_0^t {\check{\Gc}_{D_j}}(\check{\tau}_{s-})dW_s^j +\sum_{j=1}^q \int_{0}^t (\check{\Kc}_{C_j}-\Ic_{\check{\As}})(\check{\tau}_{s-}) dN_s^j.
    \end{align*}
    We then have
     \begin{align*}
         &\int_0^t \Rc_\As\Qc(\tau_{s-}) ds + \sum_{j=1}^p \int_0^t \Rc_\As\Gc_{D_j}(\tau_{s-}) dW_s^j +\sum_{j=1}^q \int_{0}^t \Rc_\As(\Kc_{C_j}-\Ic)(\tau_{s-}) dN_s^j\\ &= \int_0^t \Rc_\As{{\Qc}}\CE_\As({\tau}_s) ds + \sum_{j=1}^q \int_0^t \Rc_\As{{\Gc}_{D_j}}\CE_\As({\tau}_s) dW_s^j +\sum_{j=1}^q \int_{0}^t \Rc_\As(\Kc_{C_j}-\Ic_{\check{\As}})\CE_\As(\tau_{s-}) dN_s^j .
     \end{align*}
    In particular this yields
    \begin{align*}
        \int_0^t [\Rc_\As\Qc -\Rc_\As{{\Qc}}\CE_\As](\tau_{s-}) ds + \sum_{j=1}^p \int_0^t [\Rc_\As\Gc_{D_j}-\Rc_\As{{\Gc}_{D_j}}\CE_\As](\tau_{s-}) dW_s^j \\+ \sum_{j=1}^q \int_{0}^t [\Rc_\As\Kc_{C_j}-\cancel{\Rc_\As} - \Rc_\As\Kc_{C_j}\CE_\As + \cancel{\Rc_\As}](\tau_{s-}) dN_s^j = 0. 
    \end{align*}
    Taking again conditional quadratic variation  (as in the proof of Lemma \ref{lem:trick}) in the above equality, we have for all $s\geq0$
    \begin{eqnarray*}
     \Rc_\As\Qc -\Rc_\As{{\Qc}}\CE_\As(\tau_{s}) &=& 0 \\
     \Rc_\As\Gc_{D_j}-\Rc_\As{{\Gc}_{D_j}}\CE_\As(\tau_{s})&=&0\\
     \Rc_\As\Kc_{C_j} - \Rc_\As\Kc_{C_j}\CE_\As  (\tau_{s})&=&0,
    \end{eqnarray*}
    for all $\{D_j\}_{j=1}^p$ and for all $\{C_j\}_{j=1}^q$. Since this is true for arbitrary $\tau_0$, we have $\Rc_\As\Qc -\Rc_\As{{\Qc}}\CE_\As = 0$, $\Rc_\As\Gc_{D_j}-\Rc_\As{{\Gc}_{D_j}}\CE_\As=0$ and $\Rc_\As\Kc_{C_j} - \Rc_\As\Kc_{C_j}\CE_\As  =0$. Using the definition of $\Qc$, we deduce $\Rc_\As\Lc -\Rc_\As{{\Lc}}\CE_\As = 0$.
 
 Then, left-applying $\Jc_\As$ to both sides of both equations we obtain $\CE_\As\Lc =\CE_\As{{\Lc}}\CE_\As$, $\CE_\As\Gc_{D_j}=\CE_\As{{\Gc}_{D_j}}\CE_\As$ and $\CE_\As\Kc_{C_j} = \CE_\As\Kc_{C_j}\CE_\As$ or, in other words, $\As$ is $\Lc^*$-, $\Gc_{D_j}^*$-and $\Kc_{C_j}^*$-invariant for all $\{D_j\}_{j=1}^p$and all $\{C_j\}_{j=1}^q$, concluding the proof.
\end{proof}

\begin{figure}
    \centering
    \begin{subfigure}[t]{0.45\textwidth}
        \includegraphics[width=0.95\textwidth]{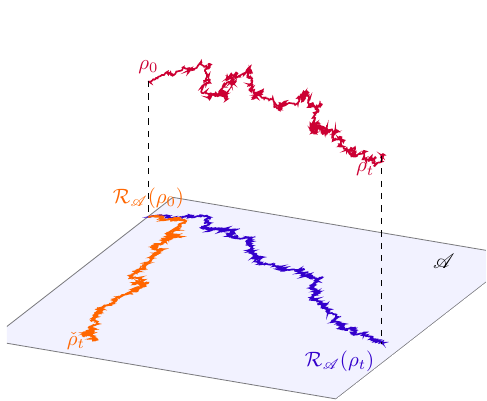}
        \caption{Case of $\As$ that is NOT invariant.}
    \end{subfigure}
    \begin{subfigure}[t]{0.45\textwidth}
        \includegraphics[width=0.95\textwidth]{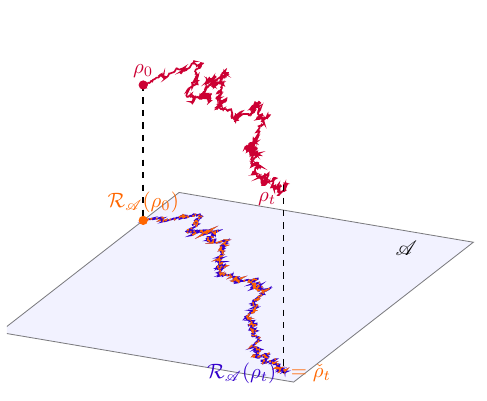}
        \caption{Case of $\As$ that is invariant.}
    \end{subfigure}
    \caption{Graphical representation of the reduced model when the algebra $\As$ is or is not $\Lc^*$-, $\Gc_{D_j}^*$- and $\Kc_{C_j}^*$-invariant.}
    \label{fig:algebra_invariance}
\end{figure}

Albeit the two trajectories $(\Rc_\As(\tau_t))_{t\geq0}$ and $(\check{\tau}_t)_{t\geq0}$ might differ in principle, as shown in Figure \ref{fig:algebra_invariance}, the important property to ensure that the two filters produce the same expectation values is that their projection onto $\Ns^\perp$ is identical, which was proven in the previous sections.

\subsection{Filter stability}
Given a quantum filter of the type \eqref{eq:SME}, we might imagine it is initialized in an estimated state $\rho_0^e$ of the true initial condition of the system $\rho_0$. This is typically the situation where the initial state $\rho_0$ is unknown. We then update a process $(\rho_t^e)$, initialized with an arbitrary state $\rho_0^e$, following the result of the measurement as if it was the true trajectory. Note that the results of measurement rely on the true trajectory. Assuming that the original filter is stable, i.e. the estimate $(\rho_t^e)$ converges to the true state $(\rho_t)_{t\geq0}$, we want to investigate if the reduced model \eqref{eq:reduced_SME}, initialized in the state $\check{\rho}_0^e =\Rc_\As(\rho_0^e)$ converges to the quantum trajectory that is initialized in $\check{\rho}_0 = \Rc_\As(\rho_0)$.
Let $\Fs$ denote the quantum fidelity, that is $\Fs(\rho,\sigma)=\tr\left[\sqrt{\sqrt{\rho}\sigma\sqrt{\rho}}\right],$
for any two states $\rho,\sigma\in\Df(\Hc).$ Recalling that  that $\Fs(\rho,\sigma)=\Fs(\sigma,\rho)$ and $0\leq\Fs(\rho,\sigma)\leq1$, with $\Fc(\rho,\sigma)=1$ if and only if $\rho=\sigma$ and $\Fs(\rho,\sigma)=0$ if and only if $\rho$ and $\sigma$ have orthogonal images we have that the fidelity $\Fs$ provides a good similarity measure between states. We refer the reader to \cite[Section 3.2]{watrous2018theory} for more details on the quantum fidelity function.

\begin{proposition}
\label{prop:filter_stability}
    Assume that the algebra $\As$ on which we restrict our model is $\Lc^*$-, $\Gc_{D_j}^*$ and $\Kc_{C_j}^*$-invariant for all $j$. Then, for all $t\geq 0$, $$\Fs(\check{\rho}_t,\check{\rho}_t^e)\geq\Fs(\rho_t,\rho_t^e).$$
\end{proposition}
\begin{proof}
From the assumption that $\As$ is $\Lc^*$- and $\Gc_{D_j}^*$ and $\Kc_{C_j}^*$-invariant for all $j$ and from Proposition \ref{thm:invariant_algebra} we have that $\check{\rho}_t=\Rc_\As(\rho_t)$ and $\check{\rho}_t^e=\Rc_\As(\rho_t^e)$ with $\Rc_\As$ CPTP for all $t\geq0$. Then, it is known that for any CPTP map $\Psi$ and any states $\rho,\sigma$, we have $\Fs(\Psi(\rho),\Psi(\sigma))\geq \Fs(\rho,\sigma)$, hence $\Fs(\check{\rho}_t,\check{\rho}_t^e)\geq\Fs(\rho_t,\rho_t^e)$.
\end{proof}

The following corollary expresses that under the condition of the previous proposition the stability of a quantum filter is preserved by the reduction.

\begin{corollary}
Under the assumptions of Proposition \ref{prop:filter_stability} and assuming that the original quantum filter is stable, i.e. $\lim_{t\to\infty}\Fs(\rho_t,\rho_t^e) =1$ we have $\lim_{t\to\infty}\Fs(\check{\rho}_t,\check{\rho}_t^e)=1$.
\end{corollary}

The above convergences can be in any sense, almost surely, in $L^p$, in probability. In particular, we know that if purification occurs
$$\lim_{t\to\infty} \Fs(\rho_t,\rho^{e}_t)=1$$
almost surely and in $L^1$\cite{amini2021asymptotic}, then the results holds for the reduced model.

\section{Illustrative applications}
\label{sec:examples}
\subsection{Generalized quantum non demolition}\label{sec:qnd}
Let us now consider a case where all the operators are block-diagonal in the same basis. This represents a generalization of the Quantum Non-Demolition (QND) condition, which will be later derived as a special case. 

Consider an Hilbert space $$\Hc=\bigoplus_{k=1}^K \Hc_k.$$ This decomposition of the Hilbert space induces a natural block-decomposition of the set of operators $\Bf(\Hc)$. In the following we consider operators that are block-diagonal in the block-decomposition induced by the decomposition of $\Hc$.

As announced, an important model which fits this situation is the so called \emph{quantum non demolition measurement} model. This case is given by one dimensional component $\mathcal H_k=\mathbb C\vert e_k\rangle$, where $\{e_1,\ldots,e_K\}$ is an orthonormal basis for $\Hc$. This basis is called a pointer basis and the involved operators are diagonal in this basis. When studying the long time behaviour of the monitored system under non demolition measurements the quantities $(\tr(\vert e_l\rangle\langle e_l\vert\rho_t)), l=1,\ldots,K$ play a crucial role. It is then natural to consider our reduction model in this situation. In particular we study the general case where block diagonal are of any dimension (not only one)

\subsubsection{Generalized QND}

Consider a Hamiltonian, noise operators and measurement operators that are block-diagonal in the basis provided by this decomposition of the Hilbert space $\Hc$, i.e. 
\[H= \bigoplus_{k=1}^K H_k, \qquad L_j= \bigoplus_{k=1}^K L_{j,k}, \qquad D_j = \bigoplus_{k=1}^K D_{j,k}, \qquad\text{and}\qquad C_j = \bigoplus_{k=1}^K C_{j,k} \]
with $H_k,L_{j,k},D_{j,k},C_{j,k}\in\Bf(\Hc_k)$.
We further assume to be interested in reproducing the expectation value observables $O_j$ that are also block-diagonal:
\[O_j = \bigoplus_{k=1}^K O_{j,k}\]
where $O_{j,k}\in\Bf(\Hc_k)$. As an example one can consider to be interested in reproducing the probability of the state being in each of the subspaces $\Hc_j,$ which implies the observables of interest are the orthogonal projectors onto each $\Hc_j$. 
In order to satisfy both Assumptions \ref{ass:observables} and  \ref{ass:identity}, we need to include $\one$, $D_j+D_j^*$ and $C_j^* C_j$ to the set of observables of interest. One can easily verify that these operators are block-diagonal as well.

For generic choices of the diagonal blocks, the space orthogonal to the non-observable space $\Ns^\perp$ generates the entire block-diagonal algebra
\[\As = \bigoplus_{k=1}^K \Bf(\Hc_k).\]
Note that, it is possible, for specific choices of the diagonal blocks, that the non-observable space is such that $\Ns^\perp\subseteq\alg(\Ns^\perp)\subsetneq\As$ or, in other words, a smaller reduction could exist. An example of this is shown in the next subsection.
Furthermore, for any choice of the diagonal blocks of $L_j,D_j,C_j$ and $O_j$, the algebra $\As$ contains $\Span\{O_j\}$, is $\Lc^*$-, $\Gc_{D_j}^*$- $\Kc^*_{C_j}$-invariant (trivially since sums and products of block-diagonal matrices remain block-diagonal) hence, in general, we have $\Ns^\perp \subseteq \alg(\Ns^\perp) \subseteq \As$. This shows that, albeit in certain cases the reduction of the filter onto $\As$ could be non-minimal, it is always possible to reduce the filter onto $\As$. 

In such a case the reduction and injection super-operators result to be 
\begin{align*}
    \Rc_\As(X) &\equiv \bigoplus_{k=1}^K V_k^* X V_k = \bigoplus_{k=1}^K X_k = \check{X},\\
    \Jc_\As(\check{X}) &\equiv \check{X}
\end{align*}
where $V_k = \ket{k} \otimes \one_{\Hc_k}$ are the isometries $V_k:\Hc_k\to\Hc$. 

The reduced un-normalized state $$\check{\tau}_t = \Rc_\As(\tau_t) = \bigoplus_{k=1}^K \check{\tau}_{t,k}$$ then evolves according to the linear stochastic differential equation 
\begin{align*}
    d\check{\tau}_{t,k} = \check{\Lc}_{k}(\check{\tau}_{t-,k})dt + \sum_{j=1}^{p}\Gc_{{D}_{j,k}}(\check{\tau}_{t-,k}) dY_t^j +\sum_{j=1}^q \left[\Kc_{{C}_{j,k}}-\Ic_{\check{\As}}\right](\check{\tau}_{t-,k}) (dN_t^j-dt)  \qquad \forall k=1,\dots,M
\end{align*}
where 
\[\check{\Lc}_{k}(\check{\tau}) = -i[H_k,\check{\tau}] + \sum_{j=1}^{m}\Dc_{L_{j,k}}(\check{\tau}) + \sum_{j=1}^{p}\Dc_{D_{j,k}}(\check{\tau}) + \sum_{j=1}^{q}\Dc_{C_{j,k}}(\check{\tau}).\]
With this, one can observe that each block $\check{\tau}_{t,k}$ of the un-normalized state $\check{\tau}_{t}$ evolves independently of all the others. This is not the case when considering the normalized state $$\check{\rho}_t = \frac{\check{\tau}_t}{\tr(\check{\tau}_t)}= \Rc_\As(\rho_t) = \bigoplus_{k=1}^K \check{\rho}_{t,k}$$ where each block evolves trough the SME 
\begin{align*}
        d\check{\rho}_{t,k} = & \check{\Lc}_k(\check{\rho}_{t-,k}) dt + \sum_{j=1}^{p} \left[ \Gc_{D_{j,k}}(\check{\rho}_{t-,k}) - \check{\rho}_{t-,k}\sum_{h=1}^K\tr[\Gc_{D_{j,h}}(\check{\rho}_{t-,h})] \right] \left(dY_t^j - \sum_{h=1}^K\tr[\Gc_{D_{j,h}}(\check{\rho}_{t-,h})] dt\right)\\ & +  \sum_{j=1}^{q}\left[\frac{\Kc_{C_{j,k}}(\check{\rho}_{t-,k})}{\sum_{h=1}^K\tr[\Kc_{C_{j,h}}(\check{\rho}_{t-,h})]}-\check{\rho}_{t-,k}\right]\left(dN_t^j-\sum_{h=1}^K\tr[\Kc_{C_{j,h}}(\check{\rho}_{t-,h})]dt\right),
\end{align*}
which clearly depends on all the blocks, (see also the discussion in \cite{amini2021asymptotic}).

The expectation values of interest is then obtained as
\[\tr[O_j \rho_t] = \sum_{k=1}^K\tr[O_{j,k}\check{\rho}_{t,k}] =  \frac{\sum_{k=1}^K\tr[O_{j,k}\check{\tau}_{t,k}]}{\sum_{h=1}^K\tr[\check{\tau}_{t,h}]}. \]

Note that the fact that each block of the reduced un-normalized state $\check{\tau}_t$ evolves independently of the others might provide a computational simulation advantage. In fact, to simulate the expectation values of interest,  one can simulate each block $\check{\tau}_{t,k}$ independently, either in parallel or in series depending on the available resources, and then sum the results to obtain the desired expectation values. Notice that the potential for independent block simulation for the average semigroup dynamics was also found in \cite{prxq2024} in presence of strong symmetries, albeit in that case the evolution is already linear and can be simulated directly in block form. 

\subsubsection{Quantum non-demolition continuous measurement}
As a special case of the example we just presented, we can focus on quantum non-demolition measurements in continuous time.
Consider an Hamiltonian and noise operators that are block-diagonal in the basis provided by the decomposition of the Hilbert space $\Hc$, i.e. 
\[H= \bigoplus_{k=1}^K H_k, \qquad\text{and}\qquad L_j= \bigoplus_{k=1}^K L_{j,k}\]
with $H_k,L_{j,k}\in\Bf(\Hc_k)$.
Furthermore, let us consider measurement operators that, in each diagonal block, are proportional to the identity operator acting on the relative subspace 
$\Hc_k$, i.e. 
\[\qquad D_j = \bigoplus_{k=1}^K d_{j,k}\one_{\Hc_{k}}, \qquad\text{and}\qquad C_j = \bigoplus_{k=1}^K c_{j,k} \one_{\Hc_k} \]
with $d_{j,k},c_{j,k}\in\Rb$.
Assume then that we are interested in reproducing the probability of the state being in each of the subspaces $\Hc_j$, i.e. we consider the observables of interest $O_j$ to be the orthogonal projectors onto $\Hc_j$:
\[O_j = \bigoplus_{k=1}^K \one_{\Hc_k}\delta_{j,k}\qquad \forall j=1,\dots,K,\]
where $\delta_{j,k}$ denotes the Kronecker delta.
With this, one can verify that both Assumptions \ref{ass:observables} and \ref{ass:identity} are satisfied, and we can thus proceed with our proposed procedure. 
Note that the typical QND setting \cite{Benoist_2014,LIANG2024111590} can be seen as a special case where $\dim(\Hc_k)=1$ for all $k$ or where $H, L_j$ are also diagonal in the standard basis.

As a first step, we shall compute the operator subspace $\Ns^\perp$. From the properties of the non-observable subspace presented in Proposition \ref{prop:non-observable} we know that $\Span\{O_j\}_{j=1}^K\subseteq\Ns^\perp$. We also know that $\Ns^\perp$ is also $\Lc^*$- $\Gc_{D_j}^*$- and $\Kc_{C_j}^*$-invariant. One can then observe that $\Span\{O_j\}_{j=1}^K$ is $\Gc_{D_j}^*$- and $\Kc_{C_j}^*$-invariant since 
\begin{align*}
    \Gc_{D_j}^*(O_h) &= D_j O_h + O_h D_j^* = \bigoplus_{k=1}^K 2d_{j,k}\delta_{h,k}\one_{\Hc_k},\\
    \Kc_{C_j}^*(O_h) &= C_j^* O_h C_j = \bigoplus_{k=1}^K c_{j,k}^2\delta_{h,k}\one_{\Hc_k}.
\end{align*}
Verifying that $\Span\{O_j\}_{j=1}^K$ is also $\Lc^*$-invariant is also straightforward. One can in fact verify that $\Span\{O_j\}_{j=1}^K$ is contained in $\ker\Lc^*$ by computing:
\begin{align*}
    [H,O_h] &= HO_h - O_hH = \bigoplus_{k=1}^K [H_k, \one_{\Hc_k}]\delta_{h,k} = 0\\
    \Dc_{L_j}^*(O_h) &= L_j^* O_h L_j -\frac{1}{2}\{L_j^* L_j, O_h\} = \bigoplus_{k=1}^K (L_{j,k}^* \one_{\Hc_k} L_{j,k} - \frac{1}{2} \{L_{j,k}^*  L_{j,k}, \one_{\Hc_k}\} )\delta_{h,k}=0,
\end{align*}
and $\Dc_{D_j}^*(O_h) = 0$ and $\Dc_{C_j}^*(O_h)=0$.

Then, since $\Span\{O_j\}_{j=1}^K$ is $\Lc^*$-, $\Gc_{D_j}^*$-, and $\Kc_{C_j}^*$-invariant and is trivially the smallest operators subspace that contains itself we find
\[\Ns^\perp = \Span\{O_j\}_{j=1}^K. \]
To proceed with the reduction procedure one should then find the algebra $\alg(\Ns^\perp)$ generated by the subspace $\Ns^\perp$. One can however notice that $\Ns^\perp$ is already an abelian algebra of dimension $M$ hence \[\As = \Span\{O_j\}_{j=1}^K = \bigoplus_{k=1}^K \Cb  \one_{\Hc_k}\]
where we also expressed its Wedderburn decomposition. This allows us to write the reduction and injection super-operators that factor the conditional expectation:
\begin{align*}
    \Rc_\As(X) &\equiv \bigoplus_{k=1}^K \tr[ V_k^* X V_k]  = \bigoplus_{k=1}^K x_k = \check{X} \quad \in\Cb^K \\
    \Jc_\As(\check{X}) &\equiv \bigoplus_{k=1}^K x_k \one_{\Hc_k},
\end{align*}
where $V_k = \ket{k} \otimes \one_{\Hc_k}$ are the isometries $V_k:\Hc_k\to\Hc$. 

We finally have all the elements to compute the reduced model. 
First of all one can notice that the action of the Lindblad generator on the algebra is null, i.e. $\Lc\circ\Rc = 0$ hence there is no need in computing the reduced Hamiltonian and noise operators since $\check{\Lc}=0$.

For the measurement operators the reduction process is quite simple. First of all we can notice that, by assumption, both $D_j$ and $C_j$ are all block-diagonal. This implies that the reduced operators will also result (block-)diagonal (in the terminology used in Appendix \ref{sec:noise_operators_reduction} we only have $e=0$). We then need to construct an orthonormal operator basis for the spaces that get factored out and express the original operators in that basis. Let us consider a set of operator basis $\{E_j^k\}$ for the spaces $\Bf(\Hc_k)$ such that $E_0^k=\one_{\Hc_k}/\dim(\Hc_k)^\um$ for all $k$. Then the original measurement operators can be written as 
\[ \qquad D_j = \bigoplus_{k=1}^K d_{j,k}\dim(\Hc_k)^\um E_0^k, \qquad C_j = \bigoplus_{k=1}^K c_{j,k}\dim(\Hc_k)^\um E_0^k.\]
Using then Proposition \ref{prop:operator_reduction} we can then directly compute the reduced measurement operators:
\[D_j \,\to\, \check{D}_j = \bigoplus_{k=1}^K d_{j,k},\qquad C_j \,\to\, \check{C}_j = \bigoplus_{k=1}^K c_{j,k} .\]

The reduced state \[\check{\rho}_t = \Rc_\As(\rho_t) = \bigoplus_{k=1}^K p_{t,k}\] with $p_{t,k}\in[0,1]$ and $\sum_{k=1}^K p_{t,k} =1$, $\forall t\geq0$ then evolves according to the SME 
\[d\check{\rho}_t = \sum_{j=1}^p (\Gc_{\check{D}_j}(\check{\rho}_{t-}) - \tr[\Gc_{\check{D}_j}(\check{\rho}_{t-})]\check{\rho}_{t-}) dW_t^j + \sum_{j=1}^q\left(\frac{\Kc_{\check{C}_j}(\check{\rho}_{t-})}{\tr[\Kc_{\check{C}_j}(\check{\rho}_{t-})]}  - \check{\rho}_{t-} \right) (dN_t^j-\tr[\Kc_{\check{C}_j}(\check{\rho}_{t-})]).\]

Because the reduced state $\check{\rho}_t$ is diagonal, one can also represent the same evolution explicitly expressing each diagonal element, obtaining
\[dp_{t,k} = \sum_{j=1}^p \left(2d_{j,k}p_{t-,k}- p_{t-,k} \sum_{h=1}^K 2d_{j,h}p_{t-,h} \right)dW_t^j + \sum_{j=1}^q \left(\frac{c_{j,k}^2 p_{t-,k}}{\sum_{h=1}^K c_{j,h}^2 p_{t-,h}}-p_{t-,k}\right)\left(dN_t^j - \sum_{h=1}^K c_{j,h}^2 p_{t-,h}\right)\]
rediscovering the form found in \cite{Benoist_2014}.
The probability of the state being in a subspace $\Hc_j$ can then be computed as
\[\tr[O_j \rho_t] = \tr[\check{O}_j \check{\rho}_t] = p_{t,j}\]
where $\check{O}_j = \ketbra{j}{j}$. 

\subsection{Measured spin chains}
\label{sec:measured_spin_chains}
We  consider here a model consisting of a spin chain undergoing both homodyne- and counting-type measurement. 
Specifically, we here consider a model composed by $N$ qubits, i.e. $\Hc= \otimes_{j=1}^N \Hc_j$ with $\Hc_j \simeq \Cb^2$. Let then $\sigma_q$ with $q\in\{0,x,y,z\}$ denote the usual Pauli matrices with $\sigma_0\equiv\one_2$ and \[\sigma_q^{(j,N)} \equiv \one_{2^{j-1}} \otimes \sigma_q \otimes \one_{2^{N-j}}\] the operators in $\Bf(\Hc) = \Cb^{2^{N}\times 2^N}$ that act nontrivially only on the $j$-th qubit. Similarly, we denote by $\sigma_\pm^{(j,N)}\equiv \frac{1}{2}\left(\sigma_x^{(j,N)}\pm i\sigma_y^{(j,N)}\right)$ local raising and lowering operators. Whenever there is no confusion on the space onto which $\sigma_q^{(j,N)}$ acts we drop the dependence on $N$ using the symbol $\sigma_q^{(j)}$. 

\subsubsection{Model description}
We assume that the spins in the chain interact trough an inhomogeneous Ising Hamiltonian with transverse field, which reads 
\begin{equation}
    H = \sum_{j=1}^{N-1} \delta_j \sigma_x^{(j)}\sigma_x^{(j+1)} + \sum_{j=1}^N \mu_j \sigma_z^{(j)},
\end{equation}
with $\delta_j,\mu_j\in\Rb$.
Furthermore, the entire system undergoes continuous-time local measurement described by the operator
\[ D_j \equiv \gamma_j \sigma_z^{(j)}, \qquad \forall j = 1,\dots, N\]
as well as counting-type measurements described by the operators
\[ C_j \equiv \alpha_j \sigma_-^{(j)}, \qquad \forall j=1,\dots,N.\]

These measurement operators can be considered either as physical description of measurement processes such as photon emission (or absorption\footnote{Note that if one were to consider $C_j=\sigma_+^{(j)}$ the derivation that follows would not change.}), or as unravelings of Lindblad generators \cite{PhysRevB.105.064305,turkeshi_measurement-induced_2021}. 
While the simultaneous continuous- and counting-type measurement considered in this example might not be physically realistic, handling both of them presents no mathematical challenge.
More importantly, as we shall see, this model allows us to study the effects of off-diagonal blocks that can be present only in counting-type measurement operators. For these reasons, we consider and reduce the model with both processes. Removing either counting-type or continuous type measurement is possible by simply setting the parameters $\alpha_j$ or $\gamma_j$ to zero. 

One could then be interested, for example, in reproducing the probability distribution in the standard basis, i.e. the considered observables of interest are
\[ O_k = \ketbra{k}{k}, \qquad \forall k = 1,\dots,2^{N}.\]
With this, one can verify that, $D_j+D_j^*\in\Span\{O_k\}_{k=1}^{2^N}$, since $D_j$ are diagonal in the basis given by $\ket{k}$,
as well as $C_j^* C_j\in\Span\{O_k\}_{k=1}^{2^N}$ since 
\begin{align*}
    {\sigma_+^{(j)}}^*\sigma_{+}^{(j)} = \sigma_{-}^{(j)}\sigma_+^{(j)} = \one_{2^{j-1}} \otimes \ketbra{1}{1}\otimes \one_{2^{N-j}},\\ 
    {\sigma_-^{(j)}}^*\sigma_{-}^{(j)} = \sigma_{+}^{(j)}\sigma_-^{(j)} = \one_{2^{j-1}} \otimes \ketbra{0}{0}\otimes \one_{2^{N-j}} 
\end{align*}
hence both assumptions \ref{ass:observables} and \ref{ass:identity} are satisfied. Note that, since $D_j\in\Span\{O_k\}_{k=1}^{2^N}$, the reduced model is also able to reproduce the expectation value of the local magnetization $\expect{\sigma_z^{(j)}}$ as well as the probability distribution $\expect{O_k}$. 

\subsubsection{Numerical reduction for $N=2,3$}
To perform the proposed model reduction procedure one can then compute the super-operator algebra $\Ts$, the space orthogonal to the non-observable space $\Ns^\perp$ defined in Proposition \ref{prop:non-observable}, compute the algebra $\alg(\Ns^\perp)$, The conditional expectation and its two factors $\Eb_{|\As}$ and then compute the reduced model as described in Appendix \ref{sec:noise_operators_reduction}. 

Note that, in principle, these tasks can be performed on a (classical) computer by obtaining numerically the required spaces, algebras and operators. This however, can be computationally demanding (the computation of $\Ts$ in particular) and depends on the size and the complexity of the dynamics for the system at hand. 
For this toy model, for example, we were able to compute numerically the spaces of interest, $\Ns^\perp$ and $\alg(\Ns^\perp)$ as well as the reduced Hamiltonian and noise operators only for $N=2,3$. 

Specifically, for $N=2,3$ we numerically verified that
\begin{equation}
    \alg(\Ns^\perp) =  \alg\left(\{ \sigma_z^{(j)}\}_{j=1}^N\cup\{\sigma_x^{(j)}\sigma_x^{(j+1)}\}_{j=1}^{N-1}\right)\simeq \Cb^{2^{N-1} \times 2^{N-1}} \oplus \Cb^{2^{N-1} \times 2^{N-1}}.
    \label{eq:example_algebra}
\end{equation}

While the computational complexity of the numerical methods might constrain us to reduce numerically only systems of small sizes, the intuition we can develop for systems of small sizes, as well as the theoretical results we developed in this work, allow us to extend some of the results theoretically to larger systems. This is in fact the approach undertaken in the next subsection, where we compute a sub-optimal reduced model that can provably be defined for any system size $N$ and is inspired by the numerical computations we just described.    

\subsubsection{Sub-optimal reduction for any $N$}\label{sec:scalableex}

Although it may be difficult to analytically compute $\Ns^\perp$, we can take inspiration from the numerically computed $*$-algebra given in \eqref{eq:example_algebra} and define the algebra
\[\As \equiv \alg\left(\{ \sigma_z^{(j)}\}_{j=1}^N\cup\{\sigma_x^{(j)}\sigma_x^{(j+1)}\}_{j=1}^{N-1}\right)\] 
with $N\geq2$. We might then wonder if, at least, such an algebra contains $\Ns^\perp,$ which would be a sufficient condition for reduction.
The following Lemma answers this question. 
\begin{lemma} 
    For the measured spin chain example described above and for $\As$ we have:
    \begin{enumerate}
        \item $\Span\{O_j\}_{j=1}^{2^N}\subset \As$;
        \item $\As$ is $\Lc^*$-, $\Gc_{D_j}^*$- and $\Kc_{C_j}^*$-invariant for all $j=1,\dots,N$;
        \item $\alg(\Ns^\perp)\subseteq\As$.
    \end{enumerate}
\end{lemma}
\begin{proof}
Let us start by defining $\As_0\equiv\alg\{\sigma_z^{(j)}\}_{j=1}^N = \Span\{O_j = \ketbra{j}{j}\}_{j=1}^{2^{N}}$ where $\{\ket{j}\}_{k=1}^{2^N}$ forms the standard basis for $\Hc$ and $O_j$ are the observables of interest. By definition of $\As$ we have that $\As_0\subset\As$ hence $\Span\{O_j\}_{j=0}^N\subset \As$.

To prove the second claim, we can observe that $H\in\As$ hence, $\As$ is invariant under the action of $[H,\cdot]$, i.e. $[H,X]\in\As$ for all $X\in\As$. Similarly, since $D_j\in\As$, we have that $\As$ is also $\Dc_{D_j}^*$- and $\Gc_{D_j}^*$-invariant. 
It thus remains to prove that $\As$ is $\Kc_{C_j}^*$- and $\Dc_{C_j}^*$-invariant. Let us first observe that for any operator $X\in\Bf(\Hc)$, we have $Q C_j^* X C_j Q^* = C_j^* Q X Q^* C_j$ where $Q\equiv \prod_{j=1}^N \sigma_z^{(j)}$. Direct calculations lead to 
\begin{align*}
    Q C_j^* X C_j Q^* &= F \underbrace{\sigma_z^{(j)} \sigma_+^{(j)}}_{\sigma_+^{(j)}} X \underbrace{\sigma_-^{(j)} \sigma_z^{(j)}}_{\sigma_-^{(j)}} F = F  \sigma_+^{(j)} X \sigma_-^{(j)} F\\
    C_j^* Q X Q^* C_j &= \underbrace{\sigma_+^{(j)} \sigma_z^{(j)}}_{-\sigma_+^{(j)}} F X F \underbrace{\sigma_z^{(j)} \sigma_-^{(j)}}_{-\sigma_-^{(j)}} = F \sigma_+^{(j)} X \sigma_-^{(j)} F
\end{align*}
where we defined $F\equiv \prod_{k\neq j} \sigma_z^{(k)}$ for convenience. Then, since the superoperators $\Kc_{C_j}^*$ and $\Sc(X)\equiv Q X Q^*$ commute, they share the same eigen-decomposition and, more importantly, every eigenspace of $\Sc$ is $\Kc_{C_j}^*$-invariant, see e.g. \cite[Sec. V.B]{prxq2024} or \cite{Buča_2012}. In particular, we shell note that the $1$-eigenspace of $\Sc$ coincides with $\{Q\}' = \As$ and is $\Kc_{C_j}^*$-invariant. This proves that $\As$ is $\Kc_{C_j}^*$-invariant. The proof of the fact that $\As$ is also $\Dc_{C_j}^*$-invariant follows from this fact and from the fact that $C_j^* C_j\in\As$.    

From the first and second claim we have that $\Ns^\perp\subseteq\As$ by  Proposition \ref{prop:non-observable}, i.e. the fact that $\Ns^\perp$ is the smallest operator space that contains $\Span\{O_j\}$ and that is $\Lc^*$-, $\Gc_{D_j}^*$- and $\Kc_{C_j}^*$-invariant for all $j$. Then, by definition of $\alg$ we have that $\alg(\Ns^\perp)$ is the smallest operator $*$-algebra that contains $\Ns^\perp$ hence $\Ns^\perp\subseteq\alg(\Ns^\perp)\subseteq\As$. 
\end{proof}

This Lemma allows us to conclude that we can reduce the spin-chain quantum filter onto the algebra $\As$, regardless of the number of spins $N$. Note that we here only proved that $\alg(\Ns^\perp) \subseteq \As$ hence, in principle, there could be smaller model than the one we compute next. 
We next show how to unitarily obtain the Wedderburn decomposition of the operators in the algebra $\As,$ which we need in order to find the reduced model. 

\begin{lemma}
    Let us define the permutation matrix \[P \equiv \one_4 + \sigma_x\otimes\one_2 + \one_2\otimes\sigma_z  -\sigma_x\otimes\sigma_z = \begin{bmatrix}
        1&0&0&0\\0&0&0&1\\0&0&1&0\\0&1&0&0
    \end{bmatrix} \in\Cb^{4\times4}\] and define, for any $N\geq1$, the unitary operator \[U_N = \begin{cases}
        \one_2 & \text{if }N=1\\
        (P\otimes\one_{2^{N-2}})(\one_2\otimes U_{N-1}) &\text{if } N\geq2
    \end{cases}.\]

    Then for all $N$: 
    \begin{enumerate}
        \item \[U_N \sigma_z^{(j)} U_N^* = \begin{cases}
            \sigma_z^{(j)}\sigma_z^{(j+1)}&\text{if }j<N\\
            \sigma_z^{(N)}&\text{if }j=N
        \end{cases};\]
        \item \(U_N \sigma_x^{(j)}\sigma_x^{(j)} U_N^* = \sigma_x^{(j+1)}\);
        \item $U_N \As U_N^* = \Bf(\Cb^{2^{N-1}}) \bigoplus \Bf(\Cb^{2^{N-1}})$;
        \item 
        \[U_N \sigma_-^{(j)} U_N^* =\begin{cases}
            \frac{1}{2}\left[\sigma_x^{(0)}\sigma_x^{(1)}\dots\sigma_x^{(j-1)}\sigma_x^{(j)} -i \sigma_x^{(0)}\sigma_x^{(1)}\dots\sigma_x^{(j-1)}\sigma_y^{(j)} \sigma_z^{(j+1)} \right]&\text{if }j < N\\
            \sigma_x^{(0)}\sigma_x^{(1)}\dots\sigma_x^{(N-1)}\sigma_-^{(N)}&\text{if }j=N
        \end{cases}. \]
    \end{enumerate}
\end{lemma}
\begin{proof}
Note that the fact that $U_N$ is unitary can easily proven by induction.
    The first claim of this Lemma is also proven by induction. 
    We start by proving the case $N=2$. Simple calculations show that 
    \begin{align*}
        U_2\sigma_z^{(1,2)} U_2^* &= P (\sigma_z\otimes\one_2) P^* = \sigma_z\otimes\sigma_z\\
        U_2\sigma_z^{(2,2)} U_2^* &= P (\one_2\otimes\sigma_z) P^* = \one_2\otimes\sigma_z.
    \end{align*}
    Now assume that $U_{N-1}\sigma_z^{(j,N-1)}U_{N-1}^* = \sigma_z^{(j,N-1)}\sigma_z^{(j,N-1)}$ for $j<N-1$ and $U_{N-1}\sigma_z^{(N-1,N-1)}U_{N-1}^* = \sigma_z^{(N-1,N-1)}$. Then,
    \begin{align*}
        U_N\sigma_z^{(1,N)} U_N^* &= (P\otimes\one_{2^{N-2}})(\one_2\otimes U_{N-1}) (\sigma_z\otimes \one_{2^{N-1}}) (\one_2\otimes U_{N-1}^*) (P^*\otimes\one_{2^{N-2}})\\
        &= P(\sigma_z\otimes\one_2)P^* \otimes\one_{2^{N-2}} = \sigma_z\otimes\sigma_z\otimes\one_{2^{N-2}} = \sigma_z^{(1,N)}\sigma_z^{(2,N)},
    \end{align*}
    and, for $j\geq2$
    \begin{align*}
        U_N\sigma_z^{(j,N)} U_N^* &= (P\otimes\one_{2^{N-2}})(\one_2\otimes U_{N-1}) (\one_2\otimes \sigma_z^{(j-1,N-1)}) (\one_2\otimes U_{N-1}^*) (P^*\otimes\one_{2^{N-2}})\\
        &= (P\otimes\one_{2^{N-2}})(\one_2\otimes \sigma_z^{(j-1,N-1)}\sigma_z^{(j,N-1)}) (P^*\otimes\one_{2^{N-2}}) = (\one_2\otimes \sigma_z^{(j-1,N-1)}\sigma_z^{(j,N-1)}) = \sigma_z^{(j,N)}\sigma_z^{(j+1,N)}
    \end{align*}
    concluding the proof of the third point. 

    The next point is also proven by induction. 
    We then prove the base case $N=2$. Simple calculations show that \(U_2\sigma_x^{(1,2)}\sigma_x^{(2,2)} U_2^* = P (\sigma_x\otimes\sigma_x) P^* = \one_2\otimes\sigma_x\) concluding the base case. Similar calculations show that \(P(\sigma_x\otimes\one_2)P^* = \sigma_x\otimes \one_2 \) and \(P(\one_2\otimes\sigma_x)P^* = \sigma_x\otimes \sigma_x \) which are useful in what comes next. We then assume that $U_{N-1} \sigma_x^{(j,N-1)}\sigma_x^{(j,N-1)} U_{N-1^*} = \sigma_x^{(j+1,N-1)}$ and consider
    \begin{align*}
        U_N \sigma_x^{(1,N)}\sigma_x^{(2,N)}U_N^* &= (P\otimes\one_{2^{N-2}})(\one_2\otimes U_{N-1}) (\sigma_x\otimes \sigma_x^{(2,N-1)}) (\one_2\otimes U_{N-1}^*) (P^*\otimes\one_{2^{N-2}})\\
        &= (P\otimes\one_{2^{N-2}}) (\sigma_x\otimes \sigma_x^{(1,N-1)}) (P^*\otimes\one_{2^{N-2}}) = (P\otimes\one_{2^{N-2}}) (\sigma_x\otimes \sigma_x \otimes \one_{2^{N-2}} ) (P^*\otimes\one_{2^{N-2}}) \\&= \one_2\otimes\sigma_x\otimes \one_{2^{N-2}} = \sigma_x^{(2,N)}
    \end{align*}
    where we used the fact that $U_N\sigma_x^{(j,N)}U_N^* = \sigma_x^{(j,N)}$ which can also be proved by induction, and, for $j>1$:
    \begin{align*}
        U_N \sigma_x^{(j,N)}\sigma_x^{(j+1,N)}U_N^* &= (P\otimes\one_{2^{N-2}})(\one_2\otimes U_{N-1}) (\one_2\otimes \sigma_x^{(j-1,N-1)}\sigma_x^{(j,N-1)}) (\one_2\otimes U_{N-1}^*) (P^*\otimes\one_{2^{N-2}})\\
        &= (P\otimes\one_{2^{N-2}})(\one_2\otimes \sigma_x^{(j,N-1)}) (P^*\otimes\one_{2^{N-2}}) = \sigma_x^{(j+1,N-1)}
    \end{align*}
    which concludes the proof of the second statement. 

    The third claim is a direct consequence of the first two claims as, by definition $\As = \alg(\{\sigma_z^{(j)}\}_{j=1}^N\cup\{\sigma_x^{(j)}\sigma_x^{(j+1)}\}_{j=1}^{N-1})$ and the fact that $U_N \sigma_z^{(j)} U_N^*$ and $U_N \sigma_x^{(j)}\sigma_x^{(j)} U_N^*$ act on the first qubit either as the identity operator or as $\sigma_z$, thus the off-diagonal blocks must be zero, e.g. $\bra{k}\otimes\one_{2^{N-1}} \sigma_x^{(j)} \ket{l}\otimes\one_{2^{N-1}} = 0$.

    The fourth claim is proven by induction as the first two claims and its proof is here omitted.
\end{proof}

Using the previous lemma, not only do we know the Wedderburn decomposition of the algebra, $\As = \simeq \Cb^{2^{N-1} \times 2^{N-1}} \oplus \Cb^{2^{N-1} \times 2^{N-1}}$, but we can also compute the reduced Hamiltonian and noise operators. 
Specifically, the operators that define $\As$, for any $N$ have the form:
\begin{align*}
    U_N \sigma_z^{(1)} U_N^* = \sigma_z^{(1)}\sigma_z^{(2)} &= \left[\begin{array}{c|c}
         \sigma_z^{(1,N-1)}&0  \\\hline
         0&-\sigma_z^{(1,N-1)} 
    \end{array}\right],\\
    U_N \sigma_z^{(j)} U_N^* = \sigma_z^{(j)}\sigma_z^{(j+1)} &= \left[\begin{array}{c|c}
         \sigma_z^{(j-1,N-1)}\sigma_z^{(j+1,N-1)}&0  \\\hline
         0&\sigma_z^{(j,N-1)}\sigma_z^{(j+1,N-1)} 
    \end{array}\right], & j=2,\dots,N-1,\\
    U_N \sigma_z^{(N)} U_N^* = \sigma_z^{(N)} &= \left[\begin{array}{c|c}
         \sigma_z^{(N-1,N-1)}&0  \\\hline
         0&\sigma_z^{(N-1,N-1)} 
    \end{array}\right],\\
    U_N \sigma_x^{(j)}\sigma_x^{(j)} U_N^* = \sigma_x^{(j+1)} &= \left[\begin{array}{c|c}
         \sigma_x^{(j-1,N-1)}&0  \\\hline
         0&\sigma_x^{(j-1,N-1)} 
    \end{array}\right] & j=1,\dots,N-1.\\
\end{align*}
Hence the reduced Hamiltonian takes the form
\[\check{H} = U_N H U_N^* = \check{H}_1\bigoplus \check{H}_2 = \left[\begin{array}{c|c}
         \check{H}_1&0  \\\hline
         0&\check{H}_2
    \end{array}\right] \]
where 
\begin{align*}
    \check{H}_1 &= \sum_{j=1}^{N-1}\delta_j\sigma_x^{(j-1,N-1)} + \mu_1 \sigma_z^{(1,N-1)} + \sum_{j=2}^{N-1} \mu_j \sigma_z^{(j-1,N-1)}\sigma_z^{(j,N-1)} + \mu_N \sigma_z^{(N-1,N-1)}\\
    \check{H}_2 &= \sum_{j=1}^{N-1}\delta_j\sigma_x^{(j-1,N-1)} -\mu_1 \sigma_z^{(1,N-1)} + \sum_{j=2}^{N-1} \mu_j \sigma_z^{(j-1,N-1)}\sigma_z^{(j,N-1)} + \mu_N \sigma_z^{(N-1,N-1)}
\end{align*}
while the reduced measurement operators associated to homodyne-type measurements are $\check{D}_j = U_N D_j U_N^* = \gamma_j U_N \sigma_z^{(j)} U_N^*$. Note that, as expected, both operators are block-diagonal in the basis defined by $U_N$.

The noise operators associated to counting-type operators instead do not belong to $\As$ and are thus not block diagonal in the basis provided by $U_N$. None the less they present a particular block-structure that can be described as follows:
\begin{align*}
    \check{C}_j = U_N C_j U_N^* &= \left[\begin{array}{c|c}
         0&F  \\\hline
         F&0
    \end{array}\right] & j<N \\
    \check{C}_N = U_N C_N U_N^* &= \left[\begin{array}{c|c}
         0&\sigma_x^{(0,N-1)}\dots\sigma_x^{(N-1,N-1)}  \\\hline
         \sigma_x^{(0,N-1)}\dots\sigma_x^{(N-1,N-1)}&0
    \end{array}\right] &
\end{align*}
where $F = \frac{1}{2}\left[\sigma_x^{(0,N-1)}\dots\sigma_x^{(j-1,N-1)} -i \sigma_x^{(0,N-1)}\dots\sigma_x^{(j-2,N-1)}\sigma_y^{(j-1,N-1)} \sigma_z^{(j,N-1)}\right]$.
Here, it is interesting to observe that the noise operators $C_j$ have a structure which is not block-diagonal in the basis given by $U_N,$ and yet they leave the algebra $\alg(\Ns^\perp)$ $\Kc_{C_j}$-invariant. This shows a departure from the generalized QND example where all operators were block-diagonal in the same basis. Furthermore, one should notice that this type of structure is only possible with operators associated to counting-type measurement as off-diagonal blocks that are non-zero in measurement operators associated to homodyne-type measurement would break the invariance of $\Ns^\perp$.

If we then represent the state $U_N\rho U_N^*$ into its block diagonal structure, i.e. 
\begin{equation*}
    U_N\rho U_N^* = \left[\begin{array}{c|c}
    \rho_{0,0} & \rho_{0,1}\\\hline
    \rho_{1,0} & \rho_{1,1}
    \end{array}\right], 
\end{equation*}
with $\rho_{j,k}\in\Cb^{2^{N-1}\times 2^{N-1}}$, the reduced density operator results to be 
\[\check{\rho} \equiv \Rc(\rho) = \bigoplus_{j=0,1} {\rho}_{j,j}
= \left[\begin{array}{c|c}
    \rho_{0,0} & 0\\\hline
    0 & \rho_{1,1}
    \end{array}\right].
\]

The observables of interest also belong to the algebra by construction, i.e. \(O_j\in\As\), and remain diagonal in the basis given by $U_N$ and thus can easily be put in block-diagonal form $\check{O} = U_N O_j U_N ^\dag$.

To conclude, \textit{for any number of qubits $N$}, it is possible to reduce the filter onto $\As$ obtaining a reduction of the dimension of the state by a factor $1/2$ and, while it could be non-minimal, can be effectively computed for any number of spins. Furthermore note that while the involved operators remain of the same dimension as the original ones, the reduction comes from the fact only the diagonal blocks of the state are populated at any time. 

\begin{figure}[th]
    \centering
    \begin{subfigure}[t]{0.48\textwidth}
        \includegraphics[width=0.9\textwidth]{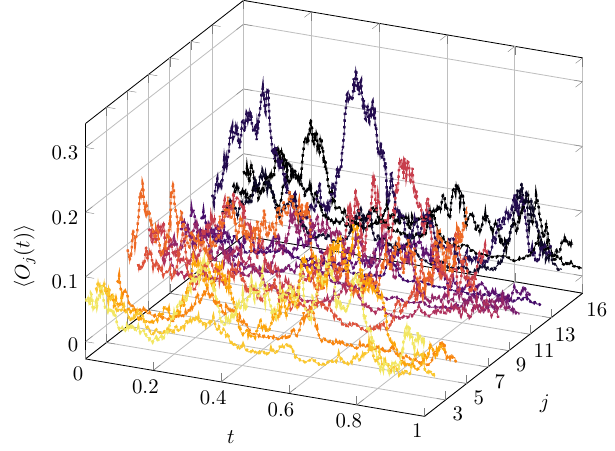}
        \label{fig:prob_diffusive}
    \end{subfigure}
    \begin{subfigure}[t]{0.48\textwidth}
        \includegraphics[width=0.9\textwidth]{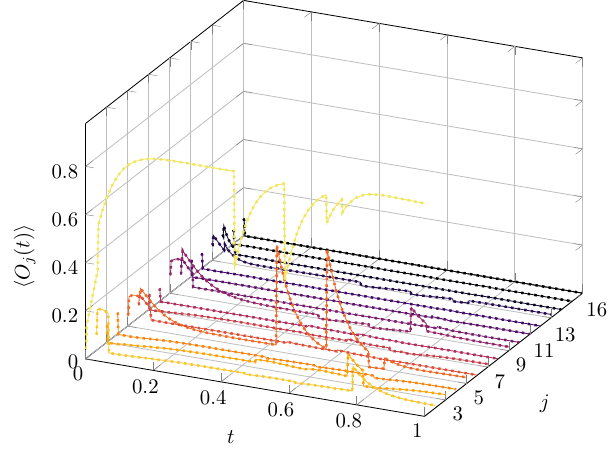}
        \label{fig:prob_jump}
    \end{subfigure}
    \begin{subfigure}[t]{0.48\textwidth}
        \includegraphics[width=0.9\textwidth]{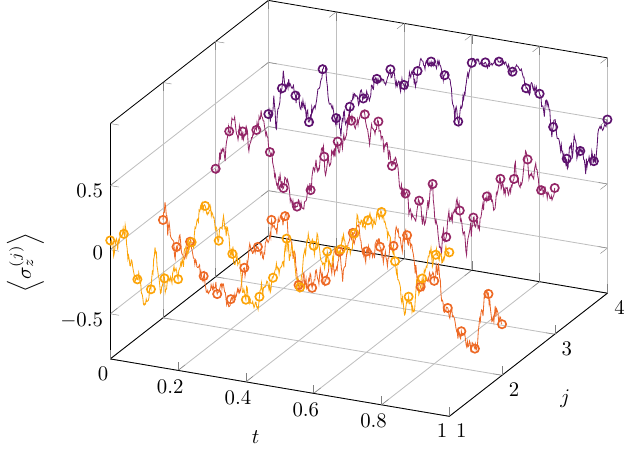}
        \label{fig:mag_diffusive}
    \end{subfigure}
    \begin{subfigure}[t]{0.48\textwidth}
        \includegraphics[width=0.9\textwidth]{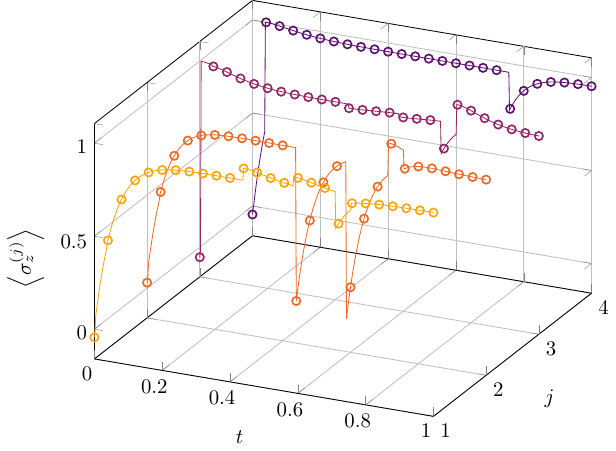}
        \label{fig:mag_jump}
    \end{subfigure}
    \begin{subfigure}[t]{0.48\textwidth}
        \includegraphics[width=0.9\textwidth]{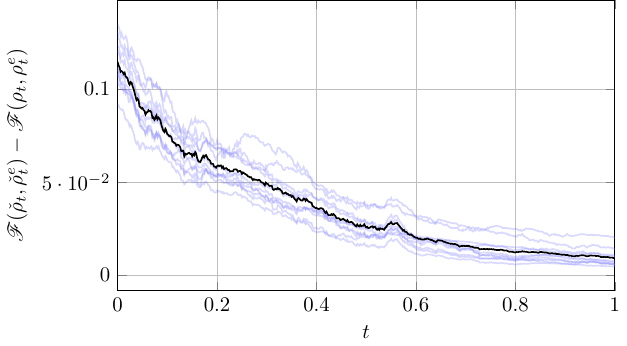}
        \label{fig:fidelity_diffusive}
    \end{subfigure}
    \begin{subfigure}[t]{0.48\textwidth}
        \includegraphics[width=0.9\textwidth]{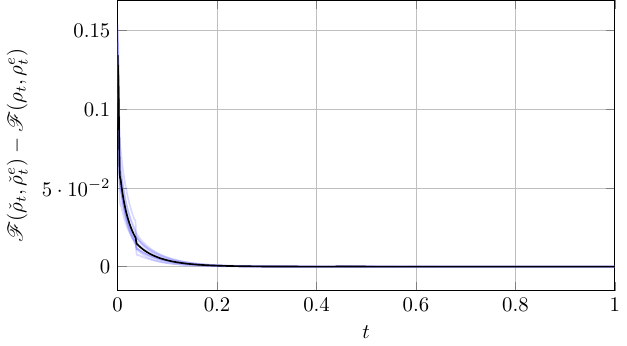}
        \label{fig:fidelity_jump}
    \end{subfigure}
    \caption{Numerical simulations for the measured spin chain with $N=4$, $\delta_j$ sampled from a Gaussian distribution with mean $2$ and standard deviation $0.2$, $\mu_j$ sampled from a Gaussian distribution with mean $1$ and standard deviation $0.2$, and $\gamma_j =\gamma$ and $\alpha_j=\alpha$ for all $j$. The left column shows a diffusive-type evolution, i.e. $\gamma=0.5$, $\alpha=0$, while the right column shows a counting-type evolution, i.e. $\gamma=0$, $\alpha=4$. From top to bottom we have: Comparison of the population in the standard basis versus time $\expect{O_j(t)}$ for the original (dots) and reduced (continuous curves) filters; Comparison of the local magnetization versus time $\expect{\sigma_z^{(j)}(t)}$ for the original (empty circles) and reduced (continuous curves) filters; Difference between the fidelity between a filter initialized in the correct initial condition $\rho_0$ and a filter initialized in a random initial condition $\rho_0^e$ for the original and reduced model. }
    \label{fig:numerical_simulation}
\end{figure}

\subsubsection{Numerical simulations}
In order to further test the validity of the reduced filter we performed numerical simulations of the described model for $N=4$. 
The numerical experiments have been performed as follows. Starting from a random initial condition $\rho_0$ we simulated the stochastic evolution of the filter using the technique proposed by \cite{rouchon_efficient_2015} and obtaining a realization of the quantum trajectory $(\rho_t)_{t\geq 0}$ as well as the measurement records $(Y_t^j)_{t\geq0}$ and $(N_t^j)_{t\geq 0}$. Using the measurement records we then simulated the evolution of the original quantum filter starting from the initial condition $\rho_0$ and the evolution of the reduced filter starting from the initial condition $\check{\rho}_0$ and computed the expectation values for the observables of interest $\expect{O_j(t)}$ and for the local magnetization $\langle{\sigma_z^{(j)}}\rangle$. In Fig.\,\ref{fig:numerical_simulation} (first two rows) one can in fact observe that the population obtained in the standard basis $\expect{O_j(t)}$ and local magnetization $\langle\sigma_z^{(j)}\rangle$ for both the full (dotted curves and empty circles) and reduced quantum model (continuous curves) are identical hence, as expected, the reduced filter correctly reproduces the expectation of the observables of interest. 

To conclude, since we have that $\As$ is $\Lc^*$-, $\Gc_D^*$- and $\Kc_{C_j}^*$- invariant, Proposition \ref{prop:filter_stability} applies and hence $\Fs(\check{\rho}_t,\check{\rho}_t^e)-\Fs(\rho_t,\rho_t^e)\geq0$. This is depicted in the last row of Fig.\,\ref{fig:numerical_simulation} where the simulation has been run ten times with the original filter initialized in a random density operator $\rho_0^e$ and the reduced filter has been initialized $\check{\rho}_0^e = \Rc(\rho_0^e)$ while using the measurement records $(Y_t^j)_{t\geq0}$ and $(N_t^j)_{t\geq 0}$ computed from the evolution of the original filter initialized in $\rho_0$. One can see that the difference between the fidelity of the original model and that of the reduced one is always greater than 0. This shows that the reduced filter is less sensitive to initialization errors. 

\section{Conclusion and outlook}

In this paper, we presented a model reduction method for quantum filters that is capable of exactly reproducing the stochastic processes associated to the expectations of observables of interest, while maintaining the reduced model in the form of a quantum filtering equation. This ensures complete positivity and trace preservation of the state evolution as well as physical interpretability. The results derived here build on the notion of observability of linear systems from control and system theory, and leverage results from quantum probability, specifically the theory of non-commutative conditional expectations. The method also offers a way to compute the minimal linear realization of the filter, which is not necessarily in the form \eqref{eq:Belavkin-Zakai}, but may be used for numerical simulation.
While the numerical complexity of the proposed method might limit the capability of reducing large models, the theoretical framework can still be useful to derive sub-optimal reduced model that work even for large systems. This has been showcased in a concrete example in Section \ref{sec:scalableex}.

This work significantly extends the findings of previous studies \cite{tit2023,prxq2024,letter2024} in several ways. First, the non-observable space defined and utilized in this paper represents a non-trivial generalization of the non-observable subspace introduced in \cite{prxq2024}. Specifically, our framework incorporates the effects of conditioning into the non-observable space, which were not accounted for previously. In particular, the presence of diffusion terms in the filtering equation, which act as a stochastic input for a dynamical system, necessitates a novel approach to defining the relevant invariant operator subspaces.
Second, this paper provides explicit methods for deriving the reduced Hamiltonian, noise and measurement operators from their original counterparts. This addresses a key open problem left unresolved in prior works.

Possible extensions of this work include developing a reduction method based on the reachability analysis of the model, that is, leveraging the knowledge of initial conditions. Since in many practical cases only a few initial conditions are considered when studying or simulating a quantum system, one can devise a dual approach to the one presented here that reduces quantum models to the algebra that contain the trajectories of interest, see e.g. \cite{prxq2024,tit2023}. Such an extension is, however, more challenging: first, one needs to introduce distorted (or smeared) algebras to ensure that the subspaces on which one projects are minimal, see \cite{tit2023} for more detail on the matter; second, as discussed in Section \ref{sec:problem_setting}, the Girsanov change of measure necessary to connect the linear stochastic evolution \eqref{eq:Belavkin-Zakai} and the SME \eqref{eq:SME}, depends on the initial condition of the model, and this should be taken into account when considering a reduction based on initial conditions. 

Other extensions of this method include {\em approximate} model-reduction protocols that are still capable of guaranteeing the complete positivity and preservation of total probability properties of the reduced model. These types of approaches promise to have direct applications in practical scenarios, where the reduced filters obtained through the exact method we propose here are too large to be implemented. Lastly, to further develop these results towards applications in feedback control, where the dynamics is made nonlinear by a state-dependent Hamiltonian perturbation, it could be convenient to extend it to a hybrid quantum-classical scenario, as those presented in \cite{barchielli2023markovian,barchielli2024hybridquantumclassicalsystemsquasifree}.

\section*{Acknowledgment}
The authors would like to thank Tristan Benoist and Lorenza Viola for useful discussions.

\noindent\textbf{Funding:}
F.T. acknowledges funding from the European Union - NextGenerationEU, within the National Center for HPC, Big Data and Quantum Computing (Project No. CN00000013, CN 1,Spoke 10).
T.G. and F.T. were partially supported by the Italian Ministry of University and Research under the PRIN project ‘‘Extracting essential information and dynamics from complex networks’’, grant  no. 2022MBC2EZ.

\noindent\textbf{Conflict of interest:} The authors declare that they have no conflict of interest. 

\bibliographystyle{imsart-number} 
\bibliography{ref}

\begin{thebibliography}{72}

\bibitem{Adler_2008}
\begin{barticle}[author]
\bauthor{\bsnm{Adler},~\bfnm{Stephen~L}\binits{S.~L.}} \AND
  \bauthor{\bsnm{Bassi},~\bfnm{Angelo}\binits{A.}}
(\byear{2008}).
\btitle{Collapse models with non-white noises: II. Particle-density coupled
  noises}.
\bjournal{Journal of Physics A: Mathematical and Theoretical}
\bvolume{41}
\bpages{395308}.
\bdoi{10.1088/1751-8113/41/39/395308}
\end{barticle}
\endbibitem

\bibitem{adler2001martingale}
\begin{barticle}[author]
\bauthor{\bsnm{Adler},~\bfnm{Stephen~Louis}\binits{S.~L.}},
  \bauthor{\bsnm{Brody},~\bfnm{DC}\binits{D.}},
  \bauthor{\bsnm{Brun},~\bfnm{TA}\binits{T.}} \AND
  \bauthor{\bsnm{Hughston},~\bfnm{LP}\binits{L.}}
(\byear{2001}).
\btitle{Martingale models for quantum state reduction}.
\bjournal{Journal of Physics A: Mathematical and General}
\bvolume{34}
\bpages{8795}.
\end{barticle}
\endbibitem

\bibitem{altafini2012modeling}
\begin{barticle}[author]
\bauthor{\bsnm{Altafini},~\bfnm{Claudio}\binits{C.}} \AND
  \bauthor{\bsnm{Ticozzi},~\bfnm{Francesco}\binits{F.}}
(\byear{2012}).
\btitle{Modeling and control of quantum systems: An introduction}.
\bjournal{IEEE Transactions on Automatic Control}
\bvolume{57}
\bpages{1898--1917}.
\end{barticle}
\endbibitem

\bibitem{amini2012stabilization}
\begin{barticle}[author]
\bauthor{\bsnm{Amini},~\bfnm{Hadis}\binits{H.}},
  \bauthor{\bsnm{Mirrahimi},~\bfnm{Mazyar}\binits{M.}} \AND
  \bauthor{\bsnm{Rouchon},~\bfnm{Pierre}\binits{P.}}
(\byear{2012}).
\btitle{Stabilization of a delayed quantum system: the photon box case-study}.
\bjournal{IEEE Transactions on Automatic Control}
\bvolume{57}
\bpages{1918--1930}.
\end{barticle}
\endbibitem

\bibitem{1amini2014stability}
\begin{barticle}[author]
\bauthor{\bsnm{Amini},~\bfnm{Hadis}\binits{H.}},
  \bauthor{\bsnm{Pellegrini},~\bfnm{Cl{\'e}ment}\binits{C.}} \AND
  \bauthor{\bsnm{Rouchon},~\bfnm{Pierre}\binits{P.}}
(\byear{2014}).
\btitle{Stability of continuous-time quantum filters with measurement
  imperfections}.
\bjournal{Russian Journal of Mathematical Physics}
\bvolume{21}
\bpages{297--315}.
\end{barticle}
\endbibitem

\bibitem{amini2021asymptotic}
\begin{barticle}[author]
\bauthor{\bsnm{Amini},~\bfnm{Nina~H}\binits{N.~H.}},
  \bauthor{\bsnm{Bompais},~\bfnm{Ma{\"e}l}\binits{M.}} \AND
  \bauthor{\bsnm{Pellegrini},~\bfnm{Cl{\'e}ment}\binits{C.}}
(\byear{2021}).
\btitle{On asymptotic stability of quantum trajectories and their Cesaro mean}.
\bjournal{Journal of Physics A: Mathematical and Theoretical}
\bvolume{54}
\bpages{385304}.
\end{barticle}
\endbibitem

\bibitem{amini2019estimation}
\begin{barticle}[author]
\bauthor{\bsnm{Amini},~\bfnm{Nina~H}\binits{N.~H.}} \AND
  \bauthor{\bsnm{Gough},~\bfnm{John~E}\binits{J.~E.}}
(\byear{2019}).
\btitle{The estimation Lie algebra associated with quantum filters}.
\bjournal{Open Systems \& Information Dynamics}
\bvolume{26}
\bpages{1950004}.
\end{barticle}
\endbibitem

\bibitem{barchielli2023markovian}
\begin{barticle}[author]
\bauthor{\bsnm{Barchielli},~\bfnm{Alberto}\binits{A.}}
(\byear{2023}).
\btitle{Markovian master equations for quantum-classical hybrid systems}.
\bjournal{Phys. Lett. A}
\bvolume{492}
\bpages{129230}.
\end{barticle}
\endbibitem

\bibitem{barchielli1991measurements11}
\begin{barticle}[author]
\bauthor{\bsnm{Barchielli},~\bfnm{Alberto}\binits{A.}} \AND
  \bauthor{\bsnm{Belavkin},~\bfnm{Viacheslav~P}\binits{V.~P.}}
(\byear{1991}).
\btitle{Measurements continuous in time and a posteriori states in quantum
  mechanics}.
\bjournal{Journal of Physics A: Mathematical and General}
\bvolume{24}
\bpages{1495}.
\end{barticle}
\endbibitem

\bibitem{barchielli2009quantum}
\begin{bbook}[author]
\bauthor{\bsnm{Barchielli},~\bfnm{Alberto}\binits{A.}} \AND
  \bauthor{\bsnm{Gregoratti},~\bfnm{Matteo}\binits{M.}}
(\byear{2009}).
\btitle{Quantum trajectories and measurements in continuous time: the diffusive
  case}
\bvolume{782}.
\bpublisher{Springer Science \& Business Media}.
\end{bbook}
\endbibitem

\bibitem{barchielliholevo}
\begin{barticle}[author]
\bauthor{\bsnm{Barchielli},~\bfnm{A.}\binits{A.}} \AND
  \bauthor{\bsnm{Holevo},~\bfnm{A.~S.}\binits{A.~S.}}
(\byear{1995}).
\btitle{Constructing quantum measurement processes via classical stochastic
  calculus}.
\bjournal{Stochastic Processes and their Applications}
\bvolume{58}
\bpages{293-317}.
\bdoi{https://doi.org/10.1016/0304-4149(95)00011-U}
\end{barticle}
\endbibitem

\bibitem{barchielli2024hybridquantumclassicalsystemsquasifree}
\begin{bmisc}[author]
\bauthor{\bsnm{Barchielli},~\bfnm{Alberto}\binits{A.}} \AND
  \bauthor{\bsnm{Werner},~\bfnm{Reinhard}\binits{R.}}
(\byear{2024}).
\btitle{Hybrid quantum-classical systems: Quasi-free Markovian dynamics}.
\end{bmisc}
\endbibitem

\bibitem{bauer2013repeated}
\begin{binproceedings}[author]
\bauthor{\bsnm{Bauer},~\bfnm{Michel}\binits{M.}},
  \bauthor{\bsnm{Benoist},~\bfnm{Tristan}\binits{T.}} \AND
  \bauthor{\bsnm{Bernard},~\bfnm{Denis}\binits{D.}}
(\byear{2013}).
\btitle{Repeated quantum non-demolition measurements: convergence and
  continuous time limit}.
In \bbooktitle{Annales Henri Poincar{\'e}}
\bvolume{14}
\bpages{639--679}.
\bpublisher{Springer}.
\end{binproceedings}
\endbibitem

\bibitem{bauer2011convergence}
\begin{barticle}[author]
\bauthor{\bsnm{Bauer},~\bfnm{Michel}\binits{M.}} \AND
  \bauthor{\bsnm{Bernard},~\bfnm{Denis}\binits{D.}}
(\byear{2011}).
\btitle{Convergence of repeated quantum nondemolition measurements and
  wave-function collapse}.
\bjournal{Physical Review A—Atomic, Molecular, and Optical Physics}
\bvolume{84}
\bpages{044103}.
\end{barticle}
\endbibitem

\bibitem{belavkin2004towards}
\begin{barticle}[author]
\bauthor{\bsnm{Belavkin},~\bfnm{VP}\binits{V.}}
(\byear{2004}).
\btitle{Towards the theory of control in observable quantum systems}.
\bjournal{arXiv preprint quant-ph/0408003}.
\end{barticle}
\endbibitem

\bibitem{BELAVKIN1992171}
\begin{barticle}[author]
\bauthor{\bsnm{Belavkin},~\bfnm{Viacheslav~P}\binits{V.~P.}}
(\byear{1992}).
\btitle{Quantum stochastic calculus and quantum nonlinear filtering}.
\bjournal{Journal of Multivariate Analysis}
\bvolume{42}
\bpages{171-201}.
\bdoi{https://doi.org/10.1016/0047-259X(92)90042-E}
\end{barticle}
\endbibitem

\bibitem{Benoist2021}
\begin{barticle}[author]
\bauthor{\bsnm{Benoist},~\bfnm{T.}\binits{T.}},
  \bauthor{\bsnm{Fraas},~\bfnm{M.}\binits{M.}},
  \bauthor{\bsnm{Pautrat},~\bfnm{Y.}\binits{Y.}} \AND
  \bauthor{\bsnm{Pellegrini},~\bfnm{C.}\binits{C.}}
(\byear{2021}).
\btitle{Invariant Measure for Stochastic Schr{\"o}dinger Equations}.
\bjournal{Annales Henri Poincar{\'e}}
\bvolume{22}
\bpages{347--374}.
\bdoi{10.1007/s00023-020-01001-4}
\end{barticle}
\endbibitem

\bibitem{benoist2024exponentially}
\begin{barticle}[author]
\bauthor{\bsnm{Benoist},~\bfnm{Tristan}\binits{T.}},
  \bauthor{\bsnm{Greggio},~\bfnm{Linda}\binits{L.}} \AND
  \bauthor{\bsnm{Pellegrini},~\bfnm{Cl{\'e}ment}\binits{C.}}
(\byear{2024}).
\btitle{Exponentially fast selection of sectors for quantum trajectories beyond
  non demolition measurements}.
\bjournal{arXiv preprint arXiv:2407.18864}.
\end{barticle}
\endbibitem

\bibitem{Benoist_2014}
\begin{barticle}[author]
\bauthor{\bsnm{Benoist},~\bfnm{Tristan}\binits{T.}} \AND
  \bauthor{\bsnm{Pellegrini},~\bfnm{Clément}\binits{C.}}
(\byear{2014}).
\btitle{Large Time Behavior and Convergence Rate for Quantum Filters Under
  Standard Non Demolition Conditions}.
\bjournal{Communications in Mathematical Physics}
\bvolume{331}
\bpages{703–723}.
\bdoi{10.1007/s00220-014-2029-6}
\end{barticle}
\endbibitem

\bibitem{benoist2017exponential}
\begin{binproceedings}[author]
\bauthor{\bsnm{Benoist},~\bfnm{Tristan}\binits{T.}},
  \bauthor{\bsnm{Pellegrini},~\bfnm{Cl{\'e}ment}\binits{C.}} \AND
  \bauthor{\bsnm{Ticozzi},~\bfnm{Francesco}\binits{F.}}
(\byear{2017}).
\btitle{Exponential stability of subspaces for quantum stochastic master
  equations}.
In \bbooktitle{Annales Henri Poincar{\'e}}
\bvolume{18}
\bpages{2045--2074}.
\bpublisher{Springer}.
\end{binproceedings}
\endbibitem

\bibitem{blackadar2006operator}
\begin{bbook}[author]
\bauthor{\bsnm{Blackadar},~\bfnm{Bruce}\binits{B.}}
(\byear{2006}).
\btitle{Operator algebras: theory of C*-algebras and von Neumann algebras}
\bvolume{122}.
\bpublisher{Springer Science \& Business Media}.
\end{bbook}
\endbibitem

\bibitem{bonato2016optimized}
\begin{barticle}[author]
\bauthor{\bsnm{Bonato},~\bfnm{Cristian}\binits{C.}},
  \bauthor{\bsnm{Blok},~\bfnm{Machiel~S}\binits{M.~S.}},
  \bauthor{\bsnm{Dinani},~\bfnm{Hossein~T}\binits{H.~T.}},
  \bauthor{\bsnm{Berry},~\bfnm{Dominic~W}\binits{D.~W.}},
  \bauthor{\bsnm{Markham},~\bfnm{Matthew~L}\binits{M.~L.}},
  \bauthor{\bsnm{Twitchen},~\bfnm{Daniel~J}\binits{D.~J.}} \AND
  \bauthor{\bsnm{Hanson},~\bfnm{Ronald}\binits{R.}}
(\byear{2016}).
\btitle{Optimized quantum sensing with a single electron spin using real-time
  adaptive measurements}.
\bjournal{Nature nanotechnology}
\bvolume{11}
\bpages{247--252}.
\end{barticle}
\endbibitem

\bibitem{bouten_quantum_2006}
\begin{bmisc}[author]
\bauthor{\bsnm{Bouten},~\bfnm{Luc}\binits{L.}} \AND \bauthor{\bparticle{van}
  \bsnm{Handel},~\bfnm{Ramon}\binits{R.}}
(\byear{2006}).
\btitle{Quantum filtering: a reference probability approach}.
\bnote{arXiv:math-ph/0508006}.
\end{bmisc}
\endbibitem

\bibitem{bouten2007introduction}
\begin{barticle}[author]
\bauthor{\bsnm{Bouten},~\bfnm{Luc}\binits{L.}},
  \bauthor{\bsnm{Van~Handel},~\bfnm{Ramon}\binits{R.}} \AND
  \bauthor{\bsnm{James},~\bfnm{Matthew~R}\binits{M.~R.}}
(\byear{2007}).
\btitle{An introduction to quantum filtering}.
\bjournal{SIAM Journal on Control and Optimization}
\bvolume{46}
\bpages{2199--2241}.
\end{barticle}
\endbibitem

\bibitem{Bucy}
\begin{barticle}[author]
\bauthor{\bsnm{Bucy},~\bfnm{R.}\binits{R.}}
(\byear{1965}).
\btitle{Nonlinear filtering theory}.
\bjournal{IEEE Transactions on Automatic Control}
\bvolume{10}
\bpages{198-198}.
\bdoi{10.1109/TAC.1965.1098109}
\end{barticle}
\endbibitem

\bibitem{Buča_2012}
\begin{barticle}[author]
\bauthor{\bsnm{Buča},~\bfnm{Berislav}\binits{B.}} \AND
  \bauthor{\bsnm{Prosen},~\bfnm{Tomaž}\binits{T.}}
(\byear{2012}).
\btitle{A note on symmetry reductions of the Lindblad equation: transport in
  constrained open spin chains}.
\bjournal{New Journal of Physics}
\bvolume{14}
\bpages{073007}.
\bdoi{10.1088/1367-2630/14/7/073007}
\end{barticle}
\endbibitem

\bibitem{cardona2020exponential}
\begin{barticle}[author]
\bauthor{\bsnm{Cardona},~\bfnm{Gerardo}\binits{G.}},
  \bauthor{\bsnm{Sarlette},~\bfnm{Alain}\binits{A.}} \AND
  \bauthor{\bsnm{Rouchon},~\bfnm{Pierre}\binits{P.}}
(\byear{2020}).
\btitle{Exponential stabilization of quantum systems under continuous
  non-demolition measurements}.
\bjournal{Automatica}
\bvolume{112}
\bpages{108719}.
\end{barticle}
\endbibitem

\bibitem{carmichael1993open}
\begin{bbook}[author]
\bauthor{\bsnm{Carmichael},~\bfnm{H.}\binits{H.}}
(\byear{1993}).
\btitle{An Open Systems Approach to Quantum Optics: Lectures Presented at the
  Universit{\'e} Libre de Bruxelles, October 28 to November 4, 1991}.
\bseries{An Open Systems Approach to Quantum Optics: Lectures Presented at the
  Universit{\'e} Libre de Bruxelles, October 28 to November 4, 1991}
\bvolume{v. 18}.
\bpublisher{Springer Berlin Heidelberg}.
\end{bbook}
\endbibitem

\bibitem{elliott2009bilinear}
\begin{bbook}[author]
\bauthor{\bsnm{Elliott},~\bfnm{David~LeRoy}\binits{D.~L.}}
(\byear{2009}).
\btitle{Bilinear control systems: matrices in action}
\bvolume{169}.
\bpublisher{Springer}.
\end{bbook}
\endbibitem

\bibitem{gao_design_2020}
\begin{barticle}[author]
\bauthor{\bsnm{Gao},~\bfnm{Qing}\binits{Q.}},
  \bauthor{\bsnm{Dong},~\bfnm{Daoyi}\binits{D.}},
  \bauthor{\bsnm{Petersen},~\bfnm{Ian~R.}\binits{I.~R.}} \AND
  \bauthor{\bsnm{Ding},~\bfnm{Steven~X.}\binits{S.~X.}}
(\byear{2020}).
\btitle{Design of a {Quantum} {Projection} {Filter}}.
\bjournal{IEEE Transactions on Automatic Control}
\bvolume{65}
\bpages{3693--3700}.
\bnote{Conference Name: IEEE Transactions on Automatic Control}.
\bdoi{10.1109/TAC.2019.2953457}
\end{barticle}
\endbibitem

\bibitem{gao_exponential_2018}
\begin{bmisc}[author]
\bauthor{\bsnm{Gao},~\bfnm{Qing}\binits{Q.}},
  \bauthor{\bsnm{Zhang},~\bfnm{Guofeng}\binits{G.}} \AND
  \bauthor{\bsnm{Petersen},~\bfnm{Ian~R.}\binits{I.~R.}}
(\byear{2018}).
\btitle{An {Exponential} {Quantum} {Projection} {Filter} for {Open} {Quantum}
  {Systems}}.
\bnote{arXiv:1705.09114 [math-ph, physics:quant-ph]}.
\bdoi{10.48550/arXiv.1705.09114}
\end{bmisc}
\endbibitem

\bibitem{gao_improved_2020}
\begin{barticle}[author]
\bauthor{\bsnm{Gao},~\bfnm{Qing}\binits{Q.}},
  \bauthor{\bsnm{Zhang},~\bfnm{Guofeng}\binits{G.}} \AND
  \bauthor{\bsnm{Petersen},~\bfnm{Ian~R.}\binits{I.~R.}}
(\byear{2020}).
\btitle{An improved quantum projection filter}.
\bjournal{Automatica}
\bvolume{112}
\bpages{108716}.
\bdoi{10.1016/j.automatica.2019.108716}
\end{barticle}
\endbibitem

\bibitem{ghirardi1990markov}
\begin{barticle}[author]
\bauthor{\bsnm{Ghirardi},~\bfnm{Gian~Carlo}\binits{G.~C.}},
  \bauthor{\bsnm{Pearle},~\bfnm{Philip}\binits{P.}} \AND
  \bauthor{\bsnm{Rimini},~\bfnm{Alberto}\binits{A.}}
(\byear{1990}).
\btitle{Markov processes in Hilbert space and continuous spontaneous
  localization of systems of identical particles}.
\bjournal{Physical Review A}
\bvolume{42}
\bpages{78}.
\end{barticle}
\endbibitem

\bibitem{ghirardi1985model}
\begin{binproceedings}[author]
\bauthor{\bsnm{Ghirardi},~\bfnm{Gian~Carlo}\binits{G.~C.}},
  \bauthor{\bsnm{Rimini},~\bfnm{Alberto}\binits{A.}} \AND
  \bauthor{\bsnm{Weber},~\bfnm{Tullio}\binits{T.}}
(\byear{1985}).
\btitle{A model for a unified quantum description of macroscopic and
  microscopic systems}.
In \bbooktitle{Quantum Probability and Applications II: Proceedings of a
  Workshop held in Heidelberg, West Germany, October 1--5, 1984}
\bpages{223--232}.
\bpublisher{Springer}.
\end{binproceedings}
\endbibitem

\bibitem{prxq2024}
\begin{barticle}[author]
\bauthor{\bsnm{Grigoletto},~\bfnm{Tommaso}\binits{T.}},
  \bauthor{\bsnm{Tao},~\bfnm{Yukuan}\binits{Y.}},
  \bauthor{\bsnm{Ticozzi},~\bfnm{Francesco}\binits{F.}} \AND
  \bauthor{\bsnm{Viola},~\bfnm{Lorenza}\binits{L.}}
(\byear{2025}).
\btitle{Exact Model Reduction for Continuous-Time Open Quantum Dynamics}.
\bjournal{Quantum}
\bvolume{9}
\bpages{1814}.
\end{barticle}
\endbibitem

\bibitem{grigoletto2021stabilization}
\begin{barticle}[author]
\bauthor{\bsnm{Grigoletto},~\bfnm{Tommaso}\binits{T.}} \AND
  \bauthor{\bsnm{Ticozzi},~\bfnm{Francesco}\binits{F.}}
(\byear{2021}).
\btitle{Stabilization via feedback switching for quantum stochastic dynamics}.
\bjournal{IEEE Control Systems Letters}
\bvolume{6}
\bpages{235--240}.
\end{barticle}
\endbibitem

\bibitem{letter2024}
\begin{barticle}[author]
\bauthor{\bsnm{Grigoletto},~\bfnm{Tommaso}\binits{T.}} \AND
  \bauthor{\bsnm{Ticozzi},~\bfnm{Francesco}\binits{F.}}
(\byear{2024}).
\btitle{Exact Model Reduction for Discrete-Time Conditional Quantum Dynamics}.
\bjournal{IEEE Control Systems Letters}
\bvolume{8}
\bpages{550-555}.
\bdoi{10.1109/LCSYS.2024.3399100}
\end{barticle}
\endbibitem

\bibitem{tit2023}
\begin{barticle}[author]
\bauthor{\bsnm{Grigoletto},~\bfnm{Tommaso}\binits{T.}} \AND
  \bauthor{\bsnm{Ticozzi},~\bfnm{Francesco}\binits{F.}}
(\byear{2025}).
\btitle{Model Reduction for Quantum Systems: Discrete-time Quantum Walks and
  Open Markov Dynamics}.
\bjournal{IEEE Transactions on Information Theory}.
\bnote{(Accepted. Preprint available at arXiv:2307.06319)}.
\end{barticle}
\endbibitem

\bibitem{handel_quantum_2005}
\begin{barticle}[author]
\bauthor{\bsnm{Handel},~\bfnm{Ramon~van}\binits{R.~v.}} \AND
  \bauthor{\bsnm{Mabuchi},~\bfnm{Hideo}\binits{H.}}
(\byear{2005}).
\btitle{Quantum projection filter for a highly nonlinear model in cavity
  {QED}}.
\bjournal{Journal of Optics B: Quantum and Semiclassical Optics}
\bvolume{7}
\bpages{S226}.
\bdoi{10.1088/1464-4266/7/10/005}
\end{barticle}
\endbibitem

\bibitem{hartman2020}
\begin{barticle}[author]
\bauthor{\bsnm{Hartmann},~\bfnm{Carsten}\binits{C.}},
  \bauthor{\bsnm{Neureither},~\bfnm{Lara}\binits{L.}} \AND
  \bauthor{\bsnm{Sharma},~\bfnm{Upanshu}\binits{U.}}
(\byear{2020}).
\btitle{Coarse Graining of Nonreversible Stochastic Differential Equations:
  Quantitative Results and Connections to Averaging}.
\bjournal{SIAM Journal on Mathematical Analysis}
\bvolume{52}
\bpages{2689-2733}.
\bdoi{10.1137/19M1299852}
\end{barticle}
\endbibitem

\bibitem{kalman}
\begin{barticle}[author]
\bauthor{\bsnm{HO},~\bfnm{BL}\binits{B.}} \AND
  \bauthor{\bsnm{K{\'a}lm{\'a}n},~\bfnm{Rudolf~E}\binits{R.~E.}}
(\byear{1966}).
\btitle{Effective construction of linear state-variable models from
  input/output functions}.
\bjournal{at-Automatisierungstechnik}
\bvolume{14}
\bpages{545--548}.
\end{barticle}
\endbibitem

\bibitem{jenvcova2006sufficiency}
\begin{barticle}[author]
\bauthor{\bsnm{Jen{\v{c}}ov{\'a}},~\bfnm{Anna}\binits{A.}} \AND
  \bauthor{\bsnm{Petz},~\bfnm{D{\'e}nes}\binits{D.}}
(\byear{2006}).
\btitle{Sufficiency in quantum statistical inference}.
\bjournal{Communications in mathematical physics}
\bvolume{263}
\bpages{259--276}.
\end{barticle}
\endbibitem

\bibitem{kalman1969topics}
\begin{bbook}[author]
\bauthor{\bsnm{Kalman},~\bfnm{Rudolf~Emil}\binits{R.~E.}},
  \bauthor{\bsnm{Falb},~\bfnm{Peter~L}\binits{P.~L.}} \AND
  \bauthor{\bsnm{Arbib},~\bfnm{Michael~A}\binits{M.~A.}}
(\byear{1969}).
\btitle{Topics in mathematical system theory}
\bvolume{1}.
\bpublisher{McGraw-Hill New York}.
\end{bbook}
\endbibitem

\bibitem{PhysRevA.66.022318}
\begin{barticle}[author]
\bauthor{\bsnm{Koashi},~\bfnm{Masato}\binits{M.}} \AND
  \bauthor{\bsnm{Imoto},~\bfnm{Nobuyuki}\binits{N.}}
(\byear{2002}).
\btitle{Operations that do not disturb partially known quantum states}.
\bjournal{Physical Review A}
\bvolume{66}
\bpages{022318}.
\bdoi{10.1103/PhysRevA.66.022318}
\end{barticle}
\endbibitem

\bibitem{kry31}
\begin{barticle}[author]
\bauthor{\bsnm{Krylov},~\bfnm{A.~N.}\binits{A.~N.}}
(\byear{1931}).
\btitle{On the numerical solution of equations whose solution determine the
  frequencies of small vibrations of material systems}.
\bjournal{Izvestija AN {SSSR} (News of Academy of Sciences of the {USSR})}
\bvolume{VII}
\bpages{491}.
\bnote{(In Russian)}.
\end{barticle}
\endbibitem

\bibitem{Kushner}
\begin{barticle}[author]
\bauthor{\bsnm{Kushner},~\bfnm{Harold~J.}\binits{H.~J.}}
(\byear{1964}).
\btitle{On the Differential Equations Satisfied by Conditional Probablitity
  Densities of Markov Processes, with Applications}.
\bjournal{Journal of the Society for Industrial and Applied Mathematics Series
  A Control}
\bvolume{2}
\bpages{106-119}.
\bdoi{10.1137/0302009}
\end{barticle}
\endbibitem

\bibitem{legoll2010effective}
\begin{barticle}[author]
\bauthor{\bsnm{Legoll},~\bfnm{Fr{\'e}d{\'e}ric}\binits{F.}} \AND
  \bauthor{\bsnm{Lelievre},~\bfnm{Tony}\binits{T.}}
(\byear{2010}).
\btitle{Effective dynamics using conditional expectations}.
\bjournal{Nonlinearity}
\bvolume{23}
\bpages{2131}.
\end{barticle}
\endbibitem

\bibitem{LIANG2024111590}
\begin{barticle}[author]
\bauthor{\bsnm{Liang},~\bfnm{Weichao}\binits{W.}} \AND
  \bauthor{\bsnm{Amini},~\bfnm{Nina~H.}\binits{N.~H.}}
(\byear{2024}).
\btitle{Model robustness for feedback stabilization of open quantum systems}.
\bjournal{Automatica}
\bvolume{163}
\bpages{111590}.
\bdoi{https://doi.org/10.1016/j.automatica.2024.111590}
\end{barticle}
\endbibitem

\bibitem{liang2022switching}
\begin{barticle}[author]
\bauthor{\bsnm{Liang},~\bfnm{Weichao}\binits{W.}},
  \bauthor{\bsnm{Grigoletto},~\bfnm{Tommaso}\binits{T.}} \AND
  \bauthor{\bsnm{Ticozzi},~\bfnm{Francesco}\binits{F.}}
(\byear{2024}).
\btitle{Dissipative feedback switching for quantum stabilization}.
\bjournal{Automatica}
\bvolume{165}
\bpages{111659}.
\end{barticle}
\endbibitem

\bibitem{marrobasile}
\begin{barticle}[author]
\bauthor{\bsnm{Marro},~\bfnm{G.}\binits{G.}} \AND
  \bauthor{\bsnm{Basile},~\bfnm{G.}\binits{G.}}
(\byear{1994}).
\btitle{Controlled and Conditioned Invariants in Linear System Theory}.
\bjournal{Automatica}
\bvolume{30}
\bpages{369--370}.
\end{barticle}
\endbibitem

\bibitem{nurdin_structures_2014}
\begin{barticle}[author]
\bauthor{\bsnm{Nurdin},~\bfnm{Hendra~I.}\binits{H.~I.}}
(\byear{2014}).
\btitle{Structures and {Transformations} for {Model} {Reduction} of {Linear}
  {Quantum} {Stochastic} {Systems}}.
\bjournal{IEEE Transactions on Automatic Control}
\bvolume{59}
\bpages{2413--2425}.
\bdoi{10.1109/TAC.2014.2322731}
\end{barticle}
\endbibitem

\bibitem{parthasarathy2012introduction}
\begin{bbook}[author]
\bauthor{\bsnm{Parthasarathy},~\bfnm{Kalyanapuram~R}\binits{K.~R.}}
(\byear{2012}).
\btitle{An introduction to quantum stochastic calculus}
\bvolume{85}.
\bpublisher{Birkh{\"a}user}.
\end{bbook}
\endbibitem

\bibitem{pellegrini2010markov}
\begin{binproceedings}[author]
\bauthor{\bsnm{Pellegrini},~\bfnm{Cl{\'e}ment}\binits{C.}}
(\byear{2010}).
\btitle{Markov chains approximation of jump-diffusion stochastic master
  equations}.
In \bbooktitle{Annales de l'IHP Probabilit{\'e}s et statistiques}
\bvolume{46}
\bpages{924--948}.
\end{binproceedings}
\endbibitem

\bibitem{percival1998quantum}
\begin{bbook}[author]
\bauthor{\bsnm{Percival},~\bfnm{I.}\binits{I.}}
(\byear{1998}).
\btitle{Quantum State Diffusion}.
\bpublisher{Cambridge University Press}.
\end{bbook}
\endbibitem

\bibitem{petz2007quantum}
\begin{bbook}[author]
\bauthor{\bsnm{Petz},~\bfnm{D{\'e}nes}\binits{D.}}
(\byear{2007}).
\btitle{Quantum information theory and quantum statistics}.
\bpublisher{Springer Science \& Business Media}.
\end{bbook}
\endbibitem

\bibitem{PhysRevB.105.064305}
\begin{barticle}[author]
\bauthor{\bsnm{Piccitto},~\bfnm{Giulia}\binits{G.}},
  \bauthor{\bsnm{Russomanno},~\bfnm{Angelo}\binits{A.}} \AND
  \bauthor{\bsnm{Rossini},~\bfnm{Davide}\binits{D.}}
(\byear{2022}).
\btitle{Entanglement transitions in the quantum Ising chain: A comparison
  between different unravelings of the same Lindbladian}.
\bjournal{Phys. Rev. B}
\bvolume{105}
\bpages{064305}.
\bdoi{10.1103/PhysRevB.105.064305}
\end{barticle}
\endbibitem

\bibitem{protter2005stochastic}
\begin{bbook}[author]
\bauthor{\bsnm{Protter},~\bfnm{Philip~E}\binits{P.~E.}} \AND
  \bauthor{\bsnm{Protter},~\bfnm{Philip~E}\binits{P.~E.}}
(\byear{2005}).
\btitle{Stochastic differential equations}.
\bpublisher{Springer}.
\end{bbook}
\endbibitem

\bibitem{ramadan_exact_2023}
\begin{bmisc}[author]
\bauthor{\bsnm{Ramadan},~\bfnm{Ibrahim}\binits{I.}},
  \bauthor{\bsnm{Amini},~\bfnm{Nina~H.}\binits{N.~H.}} \AND
  \bauthor{\bsnm{Mason},~\bfnm{Paolo}\binits{P.}}
(\byear{2023}).
\btitle{Exact solution and projection filters for open quantum systems subject
  to imperfect measurements}.
\bnote{arXiv:2311.15015 [quant-ph]}.
\end{bmisc}
\endbibitem

\bibitem{rouchon_efficient_2015}
\begin{barticle}[author]
\bauthor{\bsnm{Rouchon},~\bfnm{Pierre}\binits{P.}} \AND
  \bauthor{\bsnm{Ralph},~\bfnm{Jason~F.}\binits{J.~F.}}
(\byear{2015}).
\btitle{Efficient {Quantum} {Filtering} for {Quantum} {Feedback} {Control}}.
\bjournal{Physical Review A}
\bvolume{91}
\bpages{012118}.
\bnote{arXiv:1410.5345 [quant-ph]}.
\bdoi{10.1103/PhysRevA.91.012118}
\end{barticle}
\endbibitem

\bibitem{sayrin2011real}
\begin{barticle}[author]
\bauthor{\bsnm{Sayrin},~\bfnm{Cl{\'e}ment}\binits{C.}},
  \bauthor{\bsnm{Dotsenko},~\bfnm{Igor}\binits{I.}},
  \bauthor{\bsnm{Zhou},~\bfnm{Xingxing}\binits{X.}},
  \bauthor{\bsnm{Peaudecerf},~\bfnm{Bruno}\binits{B.}},
  \bauthor{\bsnm{Rybarczyk},~\bfnm{Th{\'e}o}\binits{T.}},
  \bauthor{\bsnm{Gleyzes},~\bfnm{S{\'e}bastien}\binits{S.}},
  \bauthor{\bsnm{Rouchon},~\bfnm{Pierre}\binits{P.}},
  \bauthor{\bsnm{Mirrahimi},~\bfnm{Mazyar}\binits{M.}},
  \bauthor{\bsnm{Amini},~\bfnm{Hadis}\binits{H.}},
  \bauthor{\bsnm{Brune},~\bfnm{Michel}\binits{M.}} \betal{et~al.}
(\byear{2011}).
\btitle{Real-time quantum feedback prepares and stabilizes photon number
  states}.
\bjournal{Nature}
\bvolume{477}
\bpages{73--77}.
\end{barticle}
\endbibitem

\bibitem{stratonovich1965conditional}
\begin{bincollection}[author]
\bauthor{\bsnm{Stratonovich},~\bfnm{Ruslan~Leont’evich}\binits{R.~L.}}
(\byear{1965}).
\btitle{Conditional markov processes}.
In \bbooktitle{Non-linear transformations of stochastic processes}
\bpages{427--453}.
\bpublisher{Elsevier}.
\end{bincollection}
\endbibitem

\bibitem{TAKESAKI1972306}
\begin{barticle}[author]
\bauthor{\bsnm{Takesaki},~\bfnm{Masamichi}\binits{M.}}
(\byear{1972}).
\btitle{Conditional expectations in von Neumann algebras}.
\bjournal{Journal of Functional Analysis}
\bvolume{9}
\bpages{306-321}.
\bdoi{https://doi.org/10.1016/0022-1236(72)90004-3}
\end{barticle}
\endbibitem

\bibitem{tirrito_full_2023}
\begin{barticle}[author]
\bauthor{\bsnm{Tirrito},~\bfnm{Emanuele}\binits{E.}},
  \bauthor{\bsnm{Santini},~\bfnm{Alessandro}\binits{A.}},
  \bauthor{\bsnm{Fazio},~\bfnm{Rosario}\binits{R.}} \AND
  \bauthor{\bsnm{Collura},~\bfnm{Mario}\binits{M.}}
(\byear{2023}).
\btitle{Full counting statistics as probe of measurement-induced transitions in
  the quantum {Ising} chain}.
\bjournal{SciPost Physics}
\bvolume{15}
\bpages{096}.
\bdoi{10.21468/SciPostPhys.15.3.096}
\end{barticle}
\endbibitem

\bibitem{turkeshi_measurement-induced_2021}
\begin{barticle}[author]
\bauthor{\bsnm{Turkeshi},~\bfnm{Xhek}\binits{X.}},
  \bauthor{\bsnm{Biella},~\bfnm{Alberto}\binits{A.}},
  \bauthor{\bsnm{Fazio},~\bfnm{Rosario}\binits{R.}},
  \bauthor{\bsnm{Dalmonte},~\bfnm{Marcello}\binits{M.}} \AND
  \bauthor{\bsnm{Schiró},~\bfnm{Marco}\binits{M.}}
(\byear{2021}).
\btitle{Measurement-induced entanglement transitions in the quantum {Ising}
  chain: {From} infinite to zero clicks}.
\bjournal{Physical Review B}
\bvolume{103}
\bpages{224210}.
\bdoi{10.1103/PhysRevB.103.224210}
\end{barticle}
\endbibitem

\bibitem{vanhandel2008stabilityquantummarkovfilters}
\begin{bmisc}[author]
\bauthor{\bparticle{van} \bsnm{Handel},~\bfnm{Ramon}\binits{R.}}
(\byear{2008}).
\btitle{The stability of quantum Markov filters}.
\end{bmisc}
\endbibitem

\bibitem{von2018mathematical}
\begin{bbook}[author]
\bauthor{\bsnm{Von~Neumann},~\bfnm{John}\binits{J.}}
(\byear{2018}).
\btitle{Mathematical foundations of quantum mechanics: New edition}.
\bpublisher{Princeton university press}.
\end{bbook}
\endbibitem

\bibitem{watrous2018theory}
\begin{bbook}[author]
\bauthor{\bsnm{Watrous},~\bfnm{John}\binits{J.}}
(\byear{2018}).
\btitle{The theory of quantum information}.
\bpublisher{Cambridge university press}.
\end{bbook}
\endbibitem

\bibitem{wedderburn1908hypercomplex}
\begin{barticle}[author]
\bauthor{\bsnm{Wedderburn},~\bfnm{JH}\binits{J.}}
(\byear{1908}).
\btitle{On hypercomplex numbers}.
\bjournal{Proceedings of the London Mathematical Society}
\bvolume{2}
\bpages{77--118}.
\end{barticle}
\endbibitem

\bibitem{wiseman1994quantum}
\begin{bphdthesis}[author]
\bauthor{\bsnm{Wiseman},~\bfnm{Howard~Mark}\binits{H.~M.}}
(\byear{1994}).
\btitle{Quantum trajectories and feedback},
\btype{PhD thesis},
\bpublisher{University of Queensland}.
\end{bphdthesis}
\endbibitem

\bibitem{wiseman2009quantum}
\begin{bbook}[author]
\bauthor{\bsnm{Wiseman},~\bfnm{Howard~M}\binits{H.~M.}} \AND
  \bauthor{\bsnm{Milburn},~\bfnm{Gerard~J}\binits{G.~J.}}
(\byear{2009}).
\btitle{Quantum measurement and control}.
\bpublisher{Cambridge university press}.
\end{bbook}
\endbibitem

\bibitem{wolf2012quantum}
\begin{barticle}[author]
\bauthor{\bsnm{Wolf},~\bfnm{Michael~M}\binits{M.~M.}}
(\byear{2012}).
\btitle{Quantum channels \& operations: Guided tour}.
\bnote{Lecture notes available at
  \url{https://citeseerx.ist.psu.edu/document?repid=rep1&type=pdf&doi=afa58291b0b8bd47504acb1ab8f553f0b37685cf}}.
\end{barticle}
\endbibitem

\bibitem{wonham}
\begin{bbook}[author]
\bauthor{\bsnm{Wonham},~\bfnm{Walter}\binits{W.}}
(\byear{1979}).
\btitle{Linear Multivariable Control: A Geometric Approach}.
\bdoi{10.1007/978-1-4684-0068-7}
\end{bbook}
\endbibitem

\end{thebibliography}

\section*{Appendix}
\appendix

\section{Reduction at the operator level}
\label{sec:noise_operators_reduction}

In this section we shall prove Proposition \ref{prop:operator_reduction} by constructing the necessary reduced operators. As in the rest of the paper we shall here focus only on the case of orthogonal conditional expectations $\Eb|_\As = \Eb|_{\As}^*$.
The fundamental superoperators that appear in the SDE \eqref{eq:reduced_quantum_filter} and defined in equation \eqref{eq:fundamental_super_operators} are all linear in the state $\rho$ but are either linear or quadratic in the operators that parameterize them, that is $H,C,D,L$. In particular, we have that $[H,\cdot]$ and $\Gc_{D}$ are linear in both the state and the operators $H$ and $D$ while $\Dc_{L}$ and $\Kc_{C}$ are linear in $\rho$ but quadratic in $L$ and $C$. 

We start by computing the reduced operators that define the superoperators that are linear in both arguments, as they are simpler to derive. For that purpose, we recall a useful result.
\begin{proposition}[\cite{prxq2024}, Proposition 5]
\label{prop:combining_reduction_injection}
    Consider a $*$-subalgebra $\mathscr{A}$ of $\mathcal{B(H)}$ with Wedderburn decomposition $$\As = U\left(\bigoplus_{k=1}^K \Bf(\Hc_{F,k})\otimes \one_{G,k} \right)U^* \simeq \bigoplus_{k=1}^K \Bf(\Hc_{F,k}) =: \check{\As}.$$ Let then $\rs_{\As}$ and $\es_{\As}$ be the CPTP factorization of the conditional expectation $\CE_\mathscr{A} = \es_\As\rs_\As$ as defined in equation \eqref{eq:cond_exp_factorization}. Then, for all $A\in\check{\As}$ and for all $X\in\Bf(\Hc)$, we have  
    \begin{equation}
    \rs_\As(X \es_\As(A)) = \es_\As^*(X) A, \qquad 
    \rs_\As(\es_\As(A) X) = A \es_\As^*(X).
    \end{equation}
\end{proposition}

Using this proposition, we can directly compute the reduced operators for the terms $[H,\cdot]$ and $\Gc_{D}(\cdot)$. Specifically, we have: 
\begin{itemize}
    \item $\Rc_\As[H,\Jc_\As(\check{\rho})] = [\Jc_\As^*(H),\check{\rho}];$ 
    \item $\Rc_\As\Gc_{D} \Jc_\As(\check{\rho}) = \Rc_\As(D \Jc_\As(\check{\rho}) + \Jc_\As(\check{\rho}) D^*) = \Jc_\As^*(D)\check{\rho} + \check{\rho} \Jc_\As^*(D)^* = \Gc_{\check{D}_j}(\check{\rho}).$
\end{itemize}
This proves points (1) and part of point (3) of Proposition \ref{prop:operator_reduction} with the reduced Hamiltonian $\check{H} = \Jc_\As^*(H)\in\check{\As}$ and with $\check{D} \equiv \Jc_\As^*(D)$.
It thus remains to prove points (2), (3) and (4). This requires the computation of the reduced operators for the two quadratic superoperators $\Dc_{L}$ and $\Kc_{C}$.

\begin{lemma}
    Let $\As\subseteq\Bf(\Hc)$ and a decomposition
    \(\Hc = \bigoplus_{k=1}^K \Hc_{F,k}\otimes\Hc_{G,k} \)
so that we can write
\begin{align*}
    \As &= U\left(\bigoplus_{k=1}^K \Bf(\Hc_{F,k}) \otimes \one_{G,k}\right)U^*.
\end{align*}
    Then, any operator $X\in\Bf(\Hc)$ admits a decomposition of the form \[X =\sum_{j,k=1}^K \sum_{\ell=1}^{d} V_j (X_{\ell,F}^{(j,k)}\otimes G_{\ell}^{(j,k)}) V_k^*\]
    where: \begin{itemize}
        \item $V_k$ are non-square isometries $V_k:\Hc_k\equiv\Hc_{F,k}\otimes\Hc_{G,k}\to \Hc$ such that $V_k^* V_j = \one_{\dim(\Hc_k)} \delta_{j,k}$ and $V_kV_j^* = \delta_{j,k}\Pi_{\Hc_k},$ with $\Pi_{\Hc_k}$ the orthogonal projector onto $\Hc_k$;
        \item $\{G_{\ell}^{(j,k)}\}_{\ell=1}^{d_{j,k}}$ are orthonormal operator basis for $\Bf(\Hc_{G,j},\Hc_{G,k})$, i.e. $\tr[G_\ell^{(j,k)*} G_f^{(j,k)}] = \delta_{\ell,f}$ with $d_{j,k}= \dim(\Hc_{G,j})\dim(\Hc_{G,k})$;
        \item for convenience we fixed $d\equiv\max_{j} \dim(\Hc_{G,j})^2$ and $X_{\ell,F}^{(j,k)} = 0$ and $G_{\ell}^{(j,k)}=0$ for all $\ell = d_{j,k}+1,\dots, d$. 
    \end{itemize} 
\end{lemma}
\begin{proof}
As mentioned above, for a unital algebra $\As\subseteq\Bf(\Hc)$, there exists a decomposition of the Hilbert space $\Hc = \bigoplus_{k=1}^K \Hc_{F,k} \otimes \Hc_{G,k}$ and a unitary matrix $U$ such that $\As = U\left(\bigoplus_{k=1}^K \Bf(\Hc_{F,k})\otimes \one_k\right)U^*$ with $\one_k\in\Bf(\Hc_{G,k})$. Such a decomposition of the Hilbert space induces a decomposition for $\Bf(\Hc)$, namely, \[\Bf(\Hc) = \bigoplus_{j,k=1}^K \Bf(\Hc_{F,j},\Hc_{F,k})\otimes\Bf(\Hc_{G,j},\Hc_{G,k}) \]
where $\Bf(\Hc_{F,j},\Hc_{F,k})$ ($\Bf(\Hc_{G,j},\Hc_{G,k})$) is the set of bounded operators from $\Hc_{F,k}$ to $\Hc_{F,j}$ (from $\Hc_{G,k}$ to $\Hc_{G,j}$ respectively).
Let us then denote by $V_k$ the non-square isometries $V_k:\Hc_k\equiv\Hc_{F,k}\otimes\Hc_{G,k}\to \Hc$ such that $V_k^* V_j = \one_{\dim(\Hc_k)} \delta_{j,k}$ and $V_kV_j^* = \Pi_{\Hc_k}\delta_{j,k}$ the orthogonal projector onto $\Hc_k$. Then, for any operator $X\in\Bf(\Hc)$ we have $X = \sum_{j,k=1}^K V_jV_j^* X V_kV_k^*$ since $\sum_{j=1}^K V_jV_j^* = \one_{\dim(\Hc)}$ which, by defining $X_{j,k} \equiv V_j^* X V_k \in\Bf(\Hc_j,\Hc_k)$, can be rewritten as $X = \sum_{j,k=1}^K V_j X_{j,k} V_k^*$. 

Let us then consider an orthonormal operator basis $\{G_{\ell}^{(j,k)}\}_{\ell=1}^{d_{j,k}}$ for $\Bf(\Hc_{G,j},\Hc_{G,k})$, i.e. $\tr[G_\ell^{(j,k)*} G_f^{(j,k)}] = \delta_{\ell,f}$.  
Each $X_{j,k}\in\Bf(\Hc_j,\Hc_k)$ can then be represented by its operator Schmidt decomposition, i.e. \[X_{j,k} = \sum_{\ell=1}^{d_{j,k}} X_{F,\ell}^{(j,k)}\otimes G_{\ell}^{(j,k)}\] with $\{X_{F,\ell}^{(j,k)}\}\subset\Bf(\Hc_{F,j},\Hc_{F,k})$.
To practically compute the operators $X_{F,\ell}^{(j,k)}$ one can construct two orthonormal vector basis $\{\ket{\phi_a}\}$ and $\{\ket{\psi_b}\}$ for $\Hc_{F,j}$ and $\Hc_{F,k}$ respectively, so that $\{\ketbra{\phi_a}{\psi_b}\}$ forms an orthonormal operator basis for $\Bf(\Hc_{F,j}, \Hc_{F,k})$ so that \[X_{F,\ell}^{(j,k)} = \sum_{a,b} \ketbra{\phi_a}{\psi_b} \tr\left[(\ketbra{\phi_a}{\psi_b}\otimes G_\ell^{(j,k)})^* X_{j,k} \right],\] which is simply a generalization of the partial trace for operator spaces of the type $\Bf(\Hc_j,\Hc_k)$ with $\Hc_{j} = \Hc_{F,j}\otimes\Hc_{G,j}$ and $\Hc_k=\Hc_{F,k}\otimes\Hc_{G,k}$. To avoid carrying with us the dependence of $d_{j,k}$ on the indexes $j,k$ one can define $d\equiv \max_{j,k} d_{j,k} = \max_{j=1,\dots,N}\dim(\Hc_{G,j})^2$ and complete the sets of operators $\{X_{F,\ell}^{(j,k)}\}_{j=1}^{d_{j,k}}$ and $\{G_\ell^{(j,k)}\}_{\ell=1}^{d_{j,k}}$ with zeros operators until their cardinality is exactly $d$. In this manner $\ell=1,\dots,d$ regardless of $j$ and $k$.
By composing the two representations we then obtain the statement. 
\end{proof}

Similarly, we can obtain a natural representation for the reduced representation of the algebra $\As$, $\check{\As}\subset\Bf(\check{\Hc})$ whose Wedderburn decomposition is $\check{\As} = \bigoplus_{k=1}^K \Bf(\Hc_{F,k})$ with $\check{\Hc} = \bigoplus_{k=1}^K \Hc_{F,k}$. Let us then denote by $W_k$ the non-square isometries $W_k:\Hc_{F,k}\to\check{\Hc}$, such that $W_j W_k^* = \Pi_{\Hc_{F,k}}\delta_{j,k}$ and $W_k^* W_j = \one_{\dim(\Hc_{F,k})}\delta_{j,k}$. A density operator $\check{\rho}\in\Df(\check{\As})$ is of the form $\check{\rho} = \sum_{k=1}^K W_k \check{\rho}_k W_k^*$ with each $\check{\rho}_k\geq0$ and such that $\sum_{k=1}^K\tr[\check{\rho}_k] = 1$ while any operator $\check{X}\in\Bf(\check{\Hc})$ can be represented as \[\check{X} = \sum_{j,k=1}^{N} W_j \check{X}_{j,k} W_k^*.\]

Let us now focus on the term $\Kc_{C}$ for any $C\in\Bf(\Hc)$.
Because $\Kc_C$, $\Rc_{\As}$ and $\Jc_\As$ are CP, we know that $\check{\Kc}_{C}\equiv \Rc_\As \Kc_C \Jc_\As$ is CP, and thus admits a Kraus representation $\check{\Kc}_C(\check{\rho}) = \sum_j \check{C}_j \check{\rho} \check{C}_j^*$ for a set of Kraus operators $\{\check{C}_j\}\subset\Bf(\check{\Hc})$. We now want to obtain one possible set $\{\check{C}_j\}$ starting from $C$ and the Wedderburn decomposition of $\As$.  

 Notice that: 1) there are clearly multiple, equivalent Kraus representations of $\check{\Kc}_C$ and 2) the reduction does not necessarily preserve the Kraus rank, i.e., the Kraus rank of $\check{\Kc}_C$ might be greater or equal than the Kraus rank of $\Kc_C$. This means that, in general, the cardinality of $\{\check{C}_j\}$ could be greater than $1$. 

Among all possible representations of the reduced map $\check{\Kc}_C$ we decided, for convenience, to use a representation where the Kraus operators are (principal) block diagonals. For this reason before the next important result we introduce the block-diagonal index $e$. Let $N$ be the number of blocks in the Wedderburn decomposition of $\check{\As}$ (and thus of $\As$), i.e. $j=1,\dots,N$. 
We can now define the principal block-diagonal index $e\equiv k-j$ and observe that $e\in[-N+1,N-1]$. 
In order to specify a single block of $\check{X}$ we can specify either its row $j$ and column $k$ position or, equivalently, its column $k$ and its principal diagonal $e$, since $j=k-e$. Once a block-diagonal index $e$ is specified, the possible column indexes are $\max(1,1+e)\leq k\leq\min(N,N+e)$.  
This means that we can equivalently rewrite any operator $\check{X}\in\Bf(\check{\Hc})$ as
\[\check{X} = \sum_{j,k=1}^K W_j \check{X}_{j,k} W_k^* = \sum_{e=-N+1}^{N-1} \sum_{k=\max(1,1+e)}^{\min(N,N+e)} W_{k-e} \check{X}_{k-e,k} W_{k}^*.\]
This description divides the operator $\check{X}$ into its principal components as depicted in Figure \ref{fig:diagonals}. In particular we find that $e=0$ is associated with the main block diagonal, i.e. the blocks $\check{X}_{k,k}$, while $e>0$ is associated with upper diagonals and $e<0$ is associated with lower diagonals.

\begin{figure}
    \centering
    \includegraphics[width=0.8\linewidth]{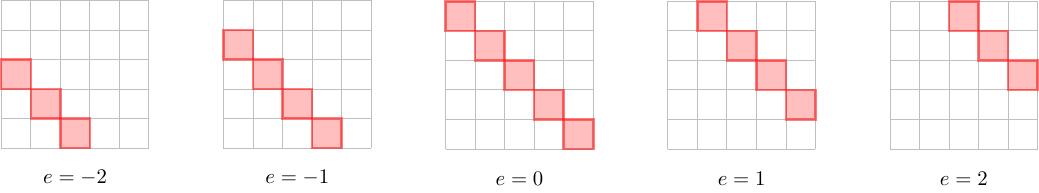}
    \caption{Pictorial representation of the block-diagonal index $e\in[-N+1,N-1]$ of the principal block diagonals.}
    \label{fig:diagonals}
\end{figure}
 
We are now finally ready for the following proposition.
\begin{proposition}
\label{prop:cp_reduction}
    Let $\Kc_{C}$ be the CP map $\Kc(\rho) \equiv C \rho C^*$ and let \[C =\sum_{j,k=1}^K\sum_{\ell=1}^{d} V_j (C_{\ell,F}^{(j,k)}\otimes G_{\ell}^{(j,k)}) V_k^*\] be the representation of the Kraus operator $C\in\Bf(\Hc)$ induced by the Wedderburn decomposition of $\As$ as above. Then, the reduced CP map $\check{\Kc}_{C} \equiv \Rc_\As \Kc_C \Jc_\As$, with $\check{\rho} = \sum_{k=1}^K W_k \check{\rho}_k W_k^*$, can be written as 
    \[\check{\Kc}_C(\check{\rho}) = \sum_{e=-N+1}^{N-1} \sum_{\ell=1}^d \Kc_{\check{C}_{\ell,e}}(\check{\rho})= \sum_{e=-N+1}^{N-1} \sum_{\ell=1}^d \check{C}_{\ell,e} \check{\rho} \check{C}_{\ell,e}^* \qquad \text{with}\qquad 
    \check{C}_{\ell,e} = \sum_{k=\max(1,1+e)}^{\min(N,N+e)} W_{k-e} \frac{C_{\ell,F}^{(k-e,k)}}{\dim(\Hc_{G,k})^\um} W_k^*.\]
\end{proposition}
\begin{proof}
The application of the injection map $\Jc_\As$ to $\check{\rho}$ gives $\Jc_\As[\check{\rho}] = \sum_h V_h \left(\check{\rho}_h\otimes \frac{\one_h}{\dim(\Hc_{G,h})}\right) V_h^*$. 
We can then compute
\begin{align*}
    \check{\Kc}_C[\check{\rho}] &= \Rc_\As\left[ \sum_{\substack{j,k,\\h,s,\\t=1}}^K\sum_{\ell,u=1}^{d} V_j (C_{\ell,F}^{(j,k)}\otimes G_{\ell}^{(j,k)}) \underbrace{V_k^* V_h}_{\one_{\dim(\Hc_h)}\delta_{k,h}} \left(\check{\rho}_h\otimes \frac{\one_h}{\dim(\Hc_{G,h})}\right) \underbrace{V_h^* V_t}_{\one_{\dim(\Hc_h)}\delta_{h,t}} (C_{u,F}^{(s,t)*}\otimes G_{u}^{(s,t)*}) V_s^* \right]\\ 
    &= \Rc_\As\left[ \sum_{j,h,s=1}^K\sum_{\ell,u=1}^{d} V_j \left(C_{\ell,F}^{(j,h)} \check{\rho}_h C_{u,F}^{(s,h)*} \otimes \frac{G_{\ell}^{(j,h)} G_{u}^{(s,h)*}}{\dim(\Hc_{G,h})}\right) V_s^* \right]\\
    &= \sum_{k,j,h,s=1}^K\sum_{\ell,u=1}^{d} \, W_k \,\,\tr_{\Hc_{G,k}} \left[ \underbrace{V_k^* V_j}_{\one_{\dim(\Hc_k)}\delta_{k,j}} \left( \underbrace{L_{\ell,F}^{(j,h)} \check{\rho}_h L_{u,F}^{(s,h)*}}_{\in\Bf(\Hc_{F,k})} \otimes \underbrace{\frac{G_{\ell}^{(j,h)} G_{u}^{(s,h)*}}{\dim(\Hc_{G,h})}}_{\in\Bf(\Hc_{F,k})}\right) \underbrace{V_s^* V_k}_{\one_{\dim(\Hc_k)}\delta_{s,k}} \right] W_k^*\\
    &= \sum_{k,h=1}^K\sum_{\ell,u=1}^{d} W_k \frac{C_{\ell,F}^{(k,h)} \check{\rho}_h C_{u,F}^{(k,h)*}}{\dim(\Hc_{G,h})} \underbrace{\tr\left[G_{\ell}^{(k,h)} G_{u}^{(k,h)*}\right]}_{\delta_{\ell,u}} W_k^*
    = \sum_{k,h=1}^K \sum_{\ell=1}^{d} W_k \frac{C_{\ell,F}^{(k,h)} \check{\rho}_h C_{\ell,F}^{(k,h)*}}{\dim(\Hc_{G,h})} W_k^*.
\end{align*}
On the other hand:

\begin{align*}
    &\sum_{e=-N+1}^{N-1} \sum_{\ell=1}^d\quad \check{C}_{\ell,e} \check{\rho} {\check{C}_{\ell,e}}^* = \\
    & = \sum_{e=-N+1}^{N-1} \sum_{\ell=1}^d\quad \left( \sum_{s=\max(1,1+e)}^{\min(N,N+e)} W_{s-e} \frac{C_{\ell,F}^{(s-e,s)}}{\dim(\Hc_{G,s})^\um} W_s^*\right) \left( \sum_{h=1}^K W_h \check{\rho}_h W_h^*\right) \left( \sum_{t=\max(1,1+e)}^{\min(N,N+e)} W_t \frac{C_{\ell,F}^{(t-e,t)*}}{\dim(\Hc_{G,t})^\um} W_{t-e}^*\right)\\
    &= \sum_{e=-N+1}^{N-1} \sum_{\ell=1}^d \sum_{h=1}^K \sum_{s,t=\max(1,1+e)}^{\min(N,N+e)}  \quad W_{s-e} \frac{C_{\ell,F}^{(s-e,s)}}{\dim(\Hc_{G,s})^\um} \underbrace{W_s^* W_h}_{\one \delta_{s,h}} \check{\rho}_h \underbrace{W_h^* W_t}_{\one \delta_{s,t}} \frac{C_{\ell,F}^{(t-e,t)*}}{\dim(\Hc_{G,t})^\um} W_{t-e}^*\\
    &= \sum_{e=-N+1}^{N-1} \sum_{h=\max(1,1+e)}^{\min(N,N+e)} \sum_{\ell=1}^{d_{h-e,h}}\quad W_{h-e} \frac{C_{\ell,F}^{(h-e,h)} \check{\rho}_h C_{\ell,F}^{(h-e,h)*}}{\dim(\Hc_{G,h})} W_{h-e}^*\\
    &= \sum_{k,h=1}^K \sum_{\ell=1}^{d}\quad W_k \frac{C_{\ell,F}^{(k,h)} \check{\rho}_h C_{\ell,F}^{(k,h)*}}{\dim(\Hc_{G,h})} W_k^* 
\end{align*}
where in the last line we substituted $k=h-e$. This concludes the proof.
\end{proof}

\begin{example}
    To provide better intuition on how the reduced operators are constructed we consider here a simple example. Consider four finite-dimensional Hilbert spaces $\Hc_{F,j},\Hc_{G,j}$ with $j=1,2$, $\dim(\Hc_{F,1}) = \dim(\Hc_{F,2})$ and $\dim(\Hc_{G,1}) = \dim(\Hc_{G,2})$. Assume that the total Hilbert space is given by $\Hc=\bigoplus_{j=1}^2 \Hc_{F,j}\otimes\Hc_{G,j}$ and consider the unital algebra $\As\subset\Bf(\Hc)$ whose Wedderburn decomposition is $\As = \bigoplus_{j=1}^2 \Bf(\Hc_{F,j}) \otimes \one_{G,j}$. Consider then a CP map $\Kc_{C}(\rho) = C\rho C^*$ with Kraus operator $C\in\Bf(\Hc)$ constructed as follows:
    \[C=\left[\begin{array}{c|c}
         F_{1,1}\otimes G_{1,1}& F_{1,2}\otimes G_{1,2} \\\hline
         F_{2,1}\otimes G_{2,1}& F_{2,2}\otimes G_{2,2}
    \end{array}\right]\]
    with generic operators $F_{j,k}\in\Bf(\Hc_{F,j})$ and $G_{j,k}\in\Bf(\Hc_{F,j}).$ The reduced CP map $\check{\Kc}_C(\check{\rho}) \equiv \Rc_\As \Kc_C \Jc_\As(\check{\rho})$ can be written in Kraus form as \[\check{\Kc}_{C}(\check{\rho}) = \sum_{e=-1}^1 \Kc_{\check{C}_e}(\check{\rho}) = \sum_{e=-1}^1 \check{C}_e \check{\rho} \check{C}_e^*\] with 
    \[\check{C}_{-1} = \left[\begin{array}{c|c}
         0& 0 \\\hline
         \gamma_{2,1}F_{2,1}& 0
    \end{array}\right],\quad
    \check{C}_{0} = \left[\begin{array}{c|c}
         \gamma_{1,1}F_{1,1}& 0 \\\hline
         0& \gamma_{2,2}F_{2,2}
    \end{array}\right],\quad
    \check{C}_{1} = \left[\begin{array}{c|c}
         0& \gamma_{1,2}F_{1,2} \\\hline
         0& 0
    \end{array}\right]\]
    where $\gamma_{j,k} = \frac{\tr[G_{j,k}^* G_{j,k}]^\um}{\dim(\Hc_{G,k})^\um}$. 
\qed
\end{example}

Let us denote with $\Xc$ the map that takes the noise operator $C$ and returns the set $\Xc(C)\equiv\{\check{C}_{\ell,e}\}_{\ell,e}$ defined in Proposition \ref{prop:cp_reduction}, and such that $\check{\Kc}_C(\check{\rho}) = \sum_{\check{C}\in\Xc(C)} \check{C} \check{\rho} \check{C}^*$. 

Clearly, for CP maps with Kraus rank $>1$, i.e. $\Psi(\rho) = \sum_j C_j\rho C_j^*,$ it holds that,  $\Rc_\As \Psi \Jc_\As[\check{\rho}] = \sum_j \sum_{\check{C}\in\Xc(C_j)} \check{C} \check{\rho} \check{C}^*$ by linearity. One can further verify that if $\Psi$ is CPTP then, since $\Rc_\As,\Jc_\As$ are CPTP then $\Rc_\As \Psi \Jc_\As$ is CPTP, and thus $\sum_j \sum_{\check{C}\in\Xc(C_j)} \check{C}^*\check{C} = \one$. This solves the problem, left open in \cite{tit2023}, of finding the reduced operators in Kraus representation for the reduction of a CPTP map $\Psi$.

The connection between the map $\Xc$ and the map $\Jc_\As^*$ we used for linear terms follows from the next corollary.
\begin{corollary}
\label{cor:reduced_operators}
    Under the assumptions of Proposition \ref{prop:cp_reduction}, and assuming, without loss of generality, that $G_{1}^{(j,j)} = \frac{\one_{\Hc_{G,j}}}{\dim(\Hc_{G,j})^\um}$ we have that $\Jc_{\As}^*(C)\in\Xc(C)$. Moreover, if $C\in\As$, we have $\Xc(C) = \{\Jc_\As^*(C)\}$.
\end{corollary}
\begin{proof}
Note that, by picking the basis $\{G_\ell^{(j,k)}\}$ so that $G_1^{(j,j)}$ is proportional to $\one_{\Hc_G,j}$ for all $j$, we have $\tr[G_\ell^{(j,j)}] = 0$ for all $\ell\neq1$ and for all $j=1,\dots,N$. Considering then $C =\sum_{j,k=1}^K\sum_{\ell=1}^d V_j (C_{\ell,F}^{(j,k)}\otimes G_{\ell}^{(j,k)}) V_k^*$ and recalling that \[\Jc_\As^*(X) = \sum_{q=1}^K W_q \tr_{\Hc_{G,q}}\left[V_q^* X V_q \left(\one_{\Hc_{F,q}}\otimes\frac{\one_{\Hc_{G,q}}}{\dim(\Hc_{G,q})}\right) \right]W_q^*\] we have:
\begin{align*}
    \Jc_\As^*(C) &= \sum_{j,k,q=1}^K \sum_{\ell=1}^d W_q \tr_{\Hc_{G,q}}\left[\underbrace{V_q^* V_j}_{\one\delta_{j,q}} (C_{\ell,F}^{(j,k)}\otimes G_{\ell}^{(j,k)}) \underbrace{V_k^* V_q}_{\delta_{k,q}} \left(\one_{\dim(\Hc_{F,q})}\otimes\frac{\one_{\dim(\Hc_{G,q})}}{\dim(\Hc_{G,q})}\right) \right]W_q^*\\
    &= \sum_{q=1}^K \sum_{\ell=1}^d W_q \tr_{\Hc_{G,q}}\left[C_{\ell,F}^{(q,q)} \otimes\frac{G_{\ell}^{(q,q)}}{\dim(\Hc_{G,q})} \right]W_q^*
    = \sum_{q=1}^K \sum_{\ell=1}^d W_q \frac{C_{\ell,F}^{(q,q)}} {\dim(\Hc_{G,q})}\underbrace{{\tr\left[G_{\ell}^{(q,q)}\right]}}_{\delta_{l,1} \tr[G_1^{(q,q)}]} W_q^*\\ 
    &= \sum_{q=1}^K W_q \frac{C_{1,F}^{(q,q)}}{\cancel{\dim(\Hc_{G,q})}} \frac{\cancel{\tr[\one_{\Hc_{G,q}}]}}{\dim(\Hc_{G,q})^\um} W_q^*
    = \sum_{q=1}^K W_q \frac{C_{1,F}^{(q,q)}}{\dim(\Hc_{G,q})^\um} W_q^* = \check{L}_{1,0}.
\end{align*}
The second statement follows directly from the fact that $L\in\As$ implies \(L =\sum_{k} V_k (L_{1,F}^{(k,k)}\otimes G_{1}^{(k,k)}) V_k^*.\)
\end{proof}

Proposition \ref{prop:cp_reduction} proves the first half of point (4) of Proposition \ref{prop:operator_reduction} and we next focus on the terms $\Dc_L$ and $\check{\Dc}_L = \Rc_\As\Dc_L\Jc_\As$ for which the following lemma comes useful.
\begin{lemma}
\label{lem:operator_reduction_dissipative}
    Under the assumptions of Proposition \ref{prop:cp_reduction}, we have $\sum_{\check{L}\in\Xc(L)}\check{L}^* \check{L} = \Jc_\As^*(L^* L)$.
\end{lemma}
\begin{proof}
Recalling that $L =\sum_{j,k=1}^K\sum_{\ell=1}^d V_j (L_{\ell,F}^{(j,k)}\otimes G_{\ell}^{(j,k)}) V_k^*$  we have 
\begin{align*}
    \Jc_\As^*(L^* L) &=\sum_{j,k,s,t=1}^K\sum_{\ell,u=1}^d \Jc_\As^* \left[V_k (L_{\ell,F}^{(j,k)*}\otimes G_{\ell}^{(j,k)*}) \underbrace{V_j^* V_s}_{\one_{\Hc_j}\delta_{j,s}} (L_{u,F}^{(s,t)}\otimes G_{u}^{(s,t)}) V_t^*\right] \\
    &=\sum_{j,k,t=1}^K\sum_{\ell,u=1}^d \Jc_\As^*\left[V_k (L_{\ell,F}^{(j,k)*}L_{u,F}^{(j,t)}\otimes G_{\ell}^{(j,k)*}G_{u}^{(j,t)}) V_t^*\right] \\
    &=\sum_{j,k,t, q=1}^K\sum_{\ell,u=1}^d W_q \tr_{\Hc_{G,q}} \left[\underbrace{V_q^* V_k}_{\one_{\Hc_q}\delta_{q,k}} \left(L_{\ell,F}^{(j,k)*}L_{u,F}^{(j,t)}\otimes \frac{G_{\ell}^{(j,k)*}G_{u}^{(j,t)}}{ \dim(\Hc_{G,q})  }\right) \underbrace{V_t^* V_q}_{\one_{\Hc_q}\delta_{q,t}}\right] W_q^* \\
    &=\sum_{j,q=1}^K\sum_{\ell,u=1}^d W_q \tr_{\Hc_{G,q}} \left[L_{\ell,F}^{(j,q)*}L_{u,F}^{(j,q)}\otimes \frac{G_{\ell}^{(j,q)*}G_{u}^{(j,q)}}{\dim(\Hc_{G,q})} \right] W_q^* \\
    &=\sum_{j,q=1}^K\sum_{\ell,u=1}^d W_q  \frac{L_{\ell,F}^{(j,q)*}L_{u,F}^{(j,q)}}{\dim(\Hc_{G,q})}\underbrace{\tr\left[ G_{\ell}^{(j,q)*}G_{u}^{(j,q)} \right]}_{\delta_{\ell,u}} W_q^* 
    =\sum_{j,q=1}^K\sum_{\ell=1}^d W_q  \frac{L_{\ell,F}^{(j,q)*}L_{\ell,F}^{(j,q)}}{\dim(\Hc_{G,q})} W_q^*.
\end{align*}
On the other hand instead we have:
\begin{align*}
    \sum_{e=-N+1}^{N-1}\sum_{\ell=1}^d {\check{L}_{\ell,e}}^* \check{L}_{\ell,e} 
    &=\sum_{e=-N+1}^{N-1} \sum_{\ell=1}^d \sum_{s,t=\max(1,1+e)}^{\min(N,N+e)} W_{s-e} \frac{L_{\ell,F}^{(s-e,s)*}}{\dim(\Hc_{G,s})^\um} \underbrace{W_s^* W_t}_{\delta_{s,t}} \frac{L_{\ell,F}^{(t-e,t)}}{\dim(\Hc_{G,t})^\um} W_{t-e}^* \\
    &=\sum_{e=-N+1}^{N-1} \sum_{\ell=1}^d \sum_{s=\max(1,1+e)}^{\min(N,N+e)} W_{s+e} \frac{L_{\ell,F}^{(s-e,s)*}L_{\ell,F}^{(s-e,s)}}{\dim(\Hc_{G,s})} W_{s-e}^* 
    = \sum_{k,s=1}^K \sum_{\ell=1}^d W_s \frac{L_{\ell,F}^{(k,s)*} L_{\ell,F}^{(k,s)} }{\dim(\Hc_{G,s})} W_s^* 
\end{align*}
where we substituted $s-e=k$, 
which confirms that $\sum_{\check{L}\in\Xc(L)}\check{L}^* \check{L} = \Jc_\As^*(L^* L)$. 
\end{proof}

We are now ready to compute the reduced noise operators for the dissipative terms of the type $\Dc_{L}$. 
\begin{proposition}
\label{prop:dissipative}
    Under the assumptions of Proposition \ref{prop:cp_reduction}, given $\Dc_{L}(\rho) = L\rho L^* -\frac{1}{2}\{L^* L, \rho\}$, we have $\check{\Dc}_{L}(\check{\rho}) = \Rc_\As \Dc_L \Jc_\As(\check{\rho}) = \sum_{\check{L}\in\Xc(L)} \Dc_{\check{L}}(\check{\rho})$. 
\end{proposition}
\begin{proof}
    Combining Propositions \ref{prop:combining_reduction_injection}, \ref{prop:cp_reduction} and Lemma \ref{lem:operator_reduction_dissipative} we have \[\check{\Dc}_{L}(\check{\rho})\equiv \Jc_\As \Dc_\Lc \Rc_\As (\check{\rho}) = \sum_{\check{L}\in\Xc(L)} \check{L}\check{\rho}\check{L}^* -\frac{1}{2}\{\Jc_\As^*(L^* L),\check{\rho}\} = \sum_{\check{L}\in\Xc(L)} \check{L}\check{\rho}\check{L}^* -\frac{1}{2}\{\check{L}^* \check{L},\check{\rho}\} = \sum_{\check{L}\in\Xc(L)} \Dc_{\check{L}}(\check{\rho}).\]
\end{proof}
With this, we have a description of the reduced Lindblad generator $\check{\Lc}$ in terms of the defining operators, namely 
\begin{equation}
    \check{\Lc}(\cdot) = -i [\check{H},\cdot] + \sum_{j=1}^m \sum_{\check{L}\in\Xc(L_j)} \Dc_{\check{L}}(\cdot) + \sum_{j=1}^p \sum_{\check{D}\in\Xc(D_j)} \Dc_{\check{D}}(\cdot) + \sum_{j=1}^q \sum_{\check{C}\in\Xc(C_j)} \Dc_{\check{C}}(\cdot)
\end{equation}
Notice that, this solves the problem left open in \cite{prxq2024} of finding the reduced description of the noise operators in the general case, e.g. in the presence of weak (not strong) symmetries. 

In summary, the proof of Proposition \ref{prop:operator_reduction} then follows from direct applications of the Propositions shown in this subsection.
\begin{proof}[Proof of Proposition \ref{prop:operator_reduction}]
    Point (1) follows from Proposition \ref{prop:combining_reduction_injection}, with $\check{H} = \Jc_\As^*(H)$. Similarly, point (2) follow from Proposition \ref{prop:dissipative} with $\{\check{L}_k\} = \Xc(L)$.
    Point (3) follows from Proposition \ref{prop:combining_reduction_injection}, Proposition \ref{prop:dissipative} and Corollary \ref{cor:reduced_operators} with $\check{D} = \Jc_\As^*(D)$ and $\{\check{D}'_k\} = \Xc(D)/\{\Jc_\As^*(D)\}$. 
    To conclude, point (4) follows from Proposition \ref{prop:dissipative} and Proposition \ref{prop:cp_reduction} with $\{\check{C}_k\} = \Xc(C)$.
\end{proof}

\end{document}